\documentclass[a4paper,UKenglish,final]{lipics}

\usepackage[strict]{changepage}

\usepackage{microtype}
\usepackage{amsmath,amsfonts,amssymb,xspace}
\usepackage{color}
\usepackage[obeyFinal]{todonotes}
\usepackage{wrapfig}
\usepackage[ruled,boxed,vlined,linesnumbered]{algorithm2e}
\usepackage{paralist}
\usepackage{xparse}
\usepackage{enumitem}
\usepackage{hhline}

\usepackage{tikz}
\usetikzlibrary{arrows,shapes,decorations,automata,backgrounds,positioning}

\tikzstyle{player1}=[draw,rounded rectangle, minimum size=5mm]
\tikzstyle{player2}=[draw,rectangle,minimum size=5mm]
\tikzstyle{widget}=[draw,ellipse,dashed,minimum size=6mm]
\tikzset{every loop/.style={looseness=7}}

\renewcommand\geq{\geqslant}
\renewcommand\leq{\leqslant}
\newcommand\Z{\mathbb Z}
\newcommand\Zbar{\Z_{\infty}}
\newcommand\N{\mathbb N}
\DeclareMathOperator*{\argmin}{\arg\!\min}
\DeclareMathOperator*{\argmax}{\arg\!\max}
\newcommand\tuple[1]{\langle #1 \rangle}
\def\bar#1{\ensuremath{\overline{#1}}}

\newcommand\NPcoNP{\ensuremath{\mathrm{NP}\cap\mathrm{co\text{-}NP}}\xspace}

\newcommand\MinPl{\mathsf{Min}}
\newcommand\MaxPl{\mathsf{Max}}

\newcommand\weightedarena{\mathcal A}
\newcommand\vertices{V}
\newcommand\minvertices{\vertices_{\MinPl}}
\newcommand\maxvertices{\vertices_{\MaxPl}}
\newcommand\edges{E}
\newcommand\finalvertices{T}

\newcommand\edgeweights{\omega}
\newcommand\game{\mathcal G}
\newcommand\Payoff{\mathbf{P}}
\ProvideDocumentCommand{\gameEx}{o}{\IfNoValueTF{#1}{\tuple{\vertices,\edges,\edgeweights,\Payoff}}{\tuple{\vertices,\edges,\edgeweights,#1}}} 
\newcommand\Val{\textnormal{\textsf{Val}}}
\ProvideDocumentCommand{\RPayoff}{o}{\IfNoValueTF{#1}{\mathbf{RP}}{#1\text{-}\mathbf{RP}}}
 
\newcommand\arena{\weightedarena}

\newcommand\graphEx{\tuple{\vertices,\edges,\edgeweights}}

\newcommand\strategy{\sigma}
\newcommand\minstrategy{\strategy_{\MinPl}}

\newcommand\maxstrategy{\strategy_{\MaxPl}}
\newcommand\outcomes{\mathsf{Play}}

\newcommand\Value{\mathsf{Val}}
\newcommand\uppervalue{\overline{\Value}}
\newcommand\lowervalue{\underline{\Value}}
\newcommand\Attr{\mathsf{Attr}}

\newcommand\pricestrategy{\Value}

\newcommand\Win{\mathsf{Win}}
\newcommand\updatefunction{\mathsf{up}}
\newcommand\decisionfunction{\mathsf{dec}}
\newcommand\memory{\mathsf{mem}}

\newcommand\MP{\textnormal{\textbf{MP}}}
\newcommand\TP{\textnormal{\textbf{TP}}}

\ProvideDocumentCommand{\MCR}{o}{\IfNoValueTF{#1}{\mathbf{MCR}}{#1\text{-}\mathbf{MCR}}}
\ProvideDocumentCommand{\RDP}{o}{\IfNoValueTF{#1}{\mathbf{RDP}}{#1\text{-}\mathbf{RDP}}}
\ProvideDocumentCommand{\RMP}{o}{\IfNoValueTF{#1}{\mathbf{RMP}}{#1\text{-}\mathbf{RMP}}}

\newcommand\boundedWeight[1]{\MCR^{\leq #1}}
\newcommand\boundeduppervalue[1]{\overline{\Value}^{\leq #1}}
\newcommand\bupval[1]{\boundeduppervalue{#1}}
\newcommand{\vbar}{\overline{v}}
\newcommand\vleq{\preccurlyeq}
\newcommand\vgeq{\succcurlyeq}
\newcommand\operator{\mathcal F}
\newcommand\decomposition{\mathsf{c}}

\newcommand\proj{\mathsf{proj}}
\newcommand\interior{\textnormal{\texttt{in}}}
\newcommand\exterior{\textnormal{\texttt{ex}}}
\newcommand\target{\textnormal{\texttt{t}}}
\newcommand\operatorBis{\mathcal H}

\newcommand\resp{respectively}

\newcommand\she{he\xspace}

\usetikzlibrary{decorations.pathreplacing}

\title{To Reach or not to Reach? \hspace{5cm} Efficient Algorithms for
  Total-Payoff Games\footnote{The research leading to these results
    has received funding from the European Union Seventh Framework
    Programme (FP7/2007-2013) under Grant Agreement n°601148
    (CASSTING).}}

\titlerunning{To Reach or not to Reach? Efficient Algorithms for
  Total-Payoff Games}
\author[1]{Thomas Brihaye}
\author[2]{Gilles Geeraerts}
\author[1]{Axel Haddad}
\author[2]{Benjamin Monmege} 
\affil[1]{Universit\'e de Mons, Belgium, \texttt{thomas.brihaye,axel.haddad@umons.ac.be}}
\affil[2]{Universit\'e libre de Bruxelles, Belgium, \texttt{gigeerae,benjamin.monmege@ulb.ac.be}}

\authorrunning{T.\,Brihaye, G.\,Geeraerts, A.\,Haddad and B.\,Monmege}

\Copyright{Thomas Brihaye, Gilles Geeraerts, Axel Haddad and Benjamin Monmege}



\serieslogo{}
\volumeinfo
  {Billy Editor and Bill Editors}
  {2}
  {Conference title on which this volume is based on}
  {1}
  {1}
  {1}
\EventShortName{}
\DOI{10.4230/LIPIcs.xxx.yyy.p}

\newtheorem{proposition}[theorem]{Proposition}

\begin{document}

\maketitle

\begin{abstract}
  Quantitative games are two-player zero-sum games played on directed
  weighted graphs. Total-payoff games---that can be seen as a
  refinement of the well-studied mean-payoff games---are the variant
  where the payoff of a play is computed as the sum of the
  weights. Our aim is to describe the first pseudo-polynomial time
  algorithm for total-payoff games in the presence of arbitrary
  weights. It consists of a non-trivial application of the value
  iteration paradigm. Indeed, it requires to study, as a milestone, a
  refinement of these games, called min-cost reachability games, where
  we add a reachability objective to one of the players. For these
  games, we give an efficient value iteration algorithm to compute the
  values and optimal strategies (when they exist), that runs in
  pseudo-polynomial time. We also propose heuristics to speed up the
  computations.
\end{abstract}

\section{Introduction}

\emph{Games played on graphs} are nowadays a well-studied and
well-established model for the computer-aided design of computer
systems, as they enable \emph{automatic synthesis} of systems that are
\emph{correct-by-construction}. Of particular interest are
\emph{quantitative games}, that allow one to model precisely
\emph{quantitative} parameters of the system, such as energy
consumption. In this setting, the game is played by two players on a
directed weighted graph, where the edge weights model, for instance, a
cost or a reward associated to the moves of the players. Each vertex
of the graph belongs to one of the two players who compete by moving a
token along the graph edges, thereby forming an infinite path called a
\emph{play}. With each play is associated a real-valued \emph{payoff}
computed from the sequence of edge weights along the play. The
traditional payoffs that have been considered in the literature
include total-payoff~\cite{GimZie04}, mean-payoff~\cite{EhrMyc79} and
discounted-payoff~\cite{ZwiPat96}. In this quantitative setting, one
player aims at maximising the payoff while the other tries to minimise
it. So one wants to compute, for each player, the best payoff that
\she can guarantee from each vertex, and the associated optimal
strategies (i.e., that guarantee the optimal payoff no matter how the
adversary is playing).

Such quantitative games have been extensively studied in the
literature. Their associated decision problems (\textit{is the value
  of a given vertex above a given threshold?}) are known to be in
\NPcoNP. Mean-payoff games have arguably been best studied from the
algorithmic point of view. A landmark is Zwick and Paterson's
pseudo-polynomial time (i.e., polynomial in the weighted graph when
weights are encoded in unary) algorithm \cite{ZwiPat96}, using the
\emph{value iteration} paradigm that consists in computing a sequence
of vectors of values that converges towards the optimal values of the
vertices. After a fixed, pseudo-polynomial, number of steps, the
computed values are precise enough to deduce the actual values of all
vertices. Better pseudo-polynomial time algorithms have later been
proposed, e.g., in \cite{BjoVor07,BriCha11,ComRiz15}, also achieving
sub-exponential expected running time by means of randomisation.

In this paper, we focus on \emph{total-payoff games}. Given an
infinite play $\pi$, we denote by $\pi[k]$ the prefix of $\pi$ of
length $k$, and by $\TP(\pi[k])$ the (finite) sum of all edge weights
along this prefix. The \emph{total-payoff} of $\pi$, $\TP(\pi)$, is
the inferior limit of all those sums, i.e.,
$\TP(\pi)=\liminf_{k\to \infty} \TP(\pi[k])$.  Compared to mean-payoff
(and discounted-payoff) games, the literature on total-payoff games is
less extensive. Gimbert and Zielonka have shown~\cite{GimZie04} that
optimal memoryless strategies always exist for both players and the
best algorithm to compute the values runs in exponential time
\cite{GawSei09}, and consists in iteratively improving
strategies. Other related works include \emph{energy games} where one
player tries to optimise its energy consumption (computed again as a
sum), keeping the energy level always above 0 (which makes difficult
to apply techniques solving those games in the case of total-payoff);
and a probabilistic variant of total-payoff games, where the weights are
restricted to be non-negative \cite{CheFor13}. Yet, we argue that the
total-payoff objective is interesting as a \emph{refinement} of the
mean-payoff. Indeed, recall first that the total-payoff is finite if
and only if the mean-payoff is null. Then, the computation of the
total-payoff enables a finer, two-stage analysis of a game~$\game$:
\begin{inparaenum}[$(i)$]
\item compute the mean payoff $\MP(\game)$;
\item subtract $\MP(\game)$ from all edge weights, and scale the
  resulting weights if necessary to obtain integers. At that point,
  one has obtained a new game $\game'$ with null mean-payoff;
\item compute $\TP(\game')$ to \emph{quantify the amount of
    fluctuation around the mean-payoff} of the original game.
\end{inparaenum}
Unfortunately, so far, no efficient (i.e., pseudo-polynomial time)
algorithms for total-payoff games have been proposed, and
straightforward adaptations of Zwick and Paterson's value iteration
algorithm for mean-payoff do not work, as we demonstrate at the end of
Section~\ref{sec:quant-games}. In the present article, we fill in this
gap by introducing the first pseudo-polynomial time algorithm for
computing the values in total-payoff games.

Our solution is a non-trivial value iteration algorithm that proceeds
through nested fixed points (see Algorithm~\ref{algo:value-iter-TPO}).
A play of a total-payoff game is infinite by essence. We transform the
game so that one of the players (the minimiser) must ensure a
\emph{reachability objective}: we assume that the game ends once this
reachability objective has been met. The intuition behind this
transformation, that stems from the use of an inferior limit in the
definition of the total-payoff, is as follows: in any play $\pi$ whose
total-payoff is \emph{finite}, there is a position $\ell$ in the play
after which all the partial sums $\TP(\pi[i])$ (with $i\geq \ell$)
will be larger than or equal to the total-payoff $\TP(\pi)$ of $\pi$,
and infinitely often both will be equal. For example, consider the
game depicted in \figurename~\ref{fig:tp}(a), where the maximiser
player (henceforth called $\MaxPl$) plays with the round vertices and
the minimiser ($\MinPl$) with the square vertices. For both players,
the optimal value when playing from $v_1$ is $2$, and the play
$\pi=v_1 v_2 v_3\ v_4 v_5\ v_4 v_3\ (v_4 v_5)^\omega$ reaches this
value (i.e., $\TP(\pi)=2$). Moreover, for all $k\geq 7$:
$\TP(\pi[k])\geq \TP(\pi)$, and infinitely many prefixes ($\pi[8]$,
$\pi[10]$, $\pi[12]$, $\ldots$) have a total-payoff of $2$, as shown
in \figurename~\ref{fig:tp}(b).

Based on this observation, we transform a total-payoff game $\game$,
into a new game that has \emph{the same value as the original
  total-payoff game} but incorporates a reachability objective for
$\MinPl$. Intuitively, in this new game, we allow a new action for
$\MinPl$: after each play prefix $\pi[k]$, \she can ask to \emph{stop
  the game}, in which case the payoff of the play is the payoff
$\TP(\pi[k])$ of the prefix. However, allowing $\MinPl$ to stop the
game at any moment would not allow to obtain the same value as in the
original total-payoff game: for instance, in the example of
\figurename~\ref{fig:tp}(a), $\MinPl$ could secure value $1$ by asking
to stop after $\pi[2]$, which is strictly smaller that the actual
total-payoff ($2$) of the whole play $\pi$. So, we allow $\MaxPl$ to
\emph{veto} to stop the game, in which case both must go on
playing. Again, allowing $\MaxPl$ to turn down all of $\MinPl$'s
requests would be unfair, so we parametrise the game with a natural
number $K$, which is the maximal number of vetoes that $\MaxPl$ can
play (and we denote by $\game^K$ the resulting game). For the
\emph{play} depicted in \figurename~\ref{fig:tp}(b), letting $K=3$ is
sufficient: trying to obtain a better payoff than the optimal,
$\MinPl$ could request to stop after $\pi[0]$, $\pi[2]$ and $\pi[6]$,
and $\MaxPl$ can veto these three requests. After that, $\MaxPl$ can
safely accept the next request of $\MinPl$, since the total payoff of
all prefixes $\pi[k]$ with $k\geq 6$ are larger than or equal to
$\TP(\pi)=2$. Our key technical contribution is to show that \emph{for
  all total-payoff games, there exists a finite, pseudo-polynomial,
  value of $K$ such that the values in $\game^K$ and $\game$ coincide}
(assuming all values are finite in $\game$: we treat the $+\infty$ and
$-\infty$ values separately). Now, assume that, when $\MaxPl$ accepts
to stop the game (possibly because \she has exhausted the maximal
number $K$ of vetoes), the game moves to a \emph{target state}, and
stops. By doing so, we effectively reduce the computation of the
values in the total-payoff game $\game$ to the computation of the
values in the total-payoff game $\game^K$ \emph{with an additional
  reachability objective} (the target state) for $\MinPl$.

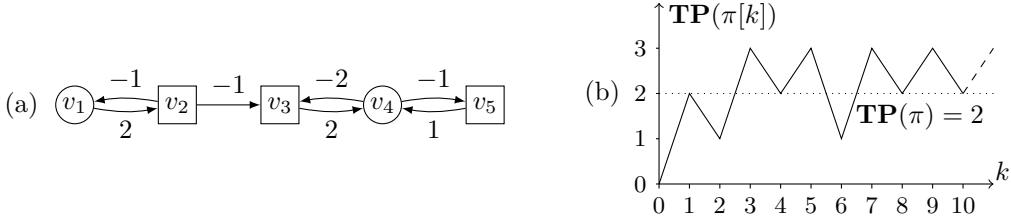
\begin{figure}[tbp]
  \centering
  \raisebox{1cm}{\begin{tikzpicture}[node distance=1.35cm,auto,->,>=latex]
  
    \node[player1](4){\makebox[0mm][c]{$v_4$}}; 
    \node[player2](3)[left of=4]{\makebox[0mm][c]{$v_3$}}; 
    \node[player2](5)[right of=4]{\makebox[0mm][c]{$v_5$}};
    \node[player2](2)[left of=3]{\makebox[0mm][c]{$v_2$}}; 
    \node[player1](1)[left of=2]{\makebox[0mm][c]{$v_1$}}; 
    
    \path 
    (4) edge[bend right=10] node[above]{$-2$} (3) 
        edge[bend left=10] node[above]{$-1$} (5)
    (3) edge[bend right=10] node[below]{$2$} (4)
    (5) edge[bend left=10] node[below]{$1$} (4);
	\path (1) edge[bend right=10] node[below]{$2$} (2);
	\path (2) edge[bend right=10] node[above]{$-1$} (1);
	\path (2) edge node[above]{$-1$} (3);
	
	\node[left of=1, node distance=.7cm]{(a)};
  \end{tikzpicture}}
  \hspace*{.7cm}
  \begin{tikzpicture}[xscale = 0.4,yscale = 0.6]
 
\draw[->] (0, -0.1) -- (0,4);
\draw[->] (-0.1,0) -- (11,0);
\draw (1,0) -- (1,-0.1);
\draw (2,0) -- (2,-0.1);
\draw (3,0) -- (3,-0.1);
\draw (4,0) -- (4,-0.1);
\draw (5,0) -- (5,-0.1);
\draw (6,0) -- (6,-0.1);
\draw (7,0) -- (7,-0.1);
\draw (8,0) -- (8,-0.1);
\draw (9,0) -- (9,-0.1);
\draw (10,0) -- (10,-0.1);

\draw (0,1) -- (-0.1,1);
\draw (0,2) -- (-0.1,2);
\draw (0,3) -- (-0.1,3);

\draw (0,0) -- (1,2) -- (2,1) --
		(3,3) -- (4,2) -- (5,3) -- 
		(6,1) -- (7,3) -- (8,2) --
		(9,3) -- (10,2);
\draw[dashed] (10,2) -- (11,3);
\draw[dotted] (0,2) -- (11,2);
\node[below left] at (11,2) {$\TP(\pi) = 2$};

\node[left] at (-0.1,0) {\small $0$};
\node[left] at (-0.1,1) {\small $1$};
\node[left](n) at (-0.1,2) {\small $2$};
\node[left] at (-0.1,3) {\small $3$};
\node[below right] at (0,4.2) {$\TP(\pi[k])$};
\node[below] at (11.3,.7) {$k$};

\node[below] at (0,-0.1) {\small $0$};
\node[below] at (1,-0.1) {\small $1$};
\node[below] at (2,-0.1) {\small $2$};
\node[below] at (3,-0.1) {\small $3$};
\node[below] at (4,-0.1) {\small $4$};
\node[below] at (5,-0.1) {\small $5$};
\node[below] at (6,-0.1) {\small $6$};
\node[below] at (7,-0.1) {\small $7$};
\node[below] at (8,-0.1) {\small $8$};
\node[below] at (9,-0.1) {\small $9$};
\node[below] at (10,-0.1) {\small $10$};

 \node[left of= n, node distance=.5cm]{(b)};
\end{tikzpicture}
\caption{(a) A total-payoff game, and (b) the evolution of the partial
  sums in $\pi$.}
  \label{fig:tp}
\end{figure}

In the following, such refined total-payoff games---where $\MinPl$
\emph{must} reach a designated target vertex---will be called
\emph{min-cost reachability games}. Failing to reach the target
vertices is the worst situation for $\MinPl$, so the payoff of all
plays that do not reach the target is $+\infty$, irrespective of the
weights along the play. Otherwise, the payoff of a play is the sum of
the weights up to the first occurrence of the target. As such, this
problem nicely generalises the classical shortest path problem in a
weighted graph. In the one-player setting (considering the point of
view of $\MinPl$ for instance), this problem can be solved in
polynomial time by Dijkstra's and Floyd-Warshall's algorithms when the
weights are non-negative and arbitrary, respectively. In
\cite{KhaBor08}, Khachiyan \textit{et al.}  propose an extension of
Dijkstra's algorithm to handle the two-player, non-negative weights
case. However, in our more general setting (two players, arbitrary
weights), this problem has, as far as we know, not been studied as
such, except that the associated decision problem is known to be in
\NPcoNP \cite{FilGen12}.  A pseudo-polynomial time algorithm to solve
a very close problem, called the \emph{longest shortest path problem}
has been introduced by Bj\"orklund and Vorobyov \cite{BjoVor07} to
eventually solve mean-payoff games. However, because of this peculiar
context of mean-payoff games, their definition of the length of a path
differs from our definition of the payoff and their algorithm cannot
be easily adapted to solve our min-cost reachability problem. Thus, as
a second contribution, we show that a value iteration algorithm
enables us to compute in pseudo-polynomial time the values of a
min-cost reachability game. We believe that min-cost reachability
games bear their own potential theoretical and practical
applications\footnote{An example of practical application would be to
  perform controller synthesis taking into account energy
  consumption. On the other hand, the problem of computing the values
  in certain classes of priced timed games has recently been reduced
  to computing the values in min-cost reachability games
  \cite{BGKMMT14a}.}.  Those games are discussed in
Section~\ref{sec:reachability-objectives}. In addition to the
pseudo-polynomial time algorithm to compute the values, we show how to
compute optimal strategies for both players and characterise them:
there is always a memoryless strategy for the maximiser player, but we
exhibit an example (see \figurename~\ref{fig:Weighted-game}$(a)$)
where the minimiser player needs (finite) memory. Those results on
min-cost reachability games are exploited in
Section~\ref{sec:solving-total-payoff} where we introduce and prove
correct our efficient algorithm for total-payoff games.

Finally, we briefly present our implementation in
Section~\ref{sec:experiments}, using as a core the numerical
model-checker PRISM. This allows us to describe some heuristics able
to improve the practical performances of our algorithms for
total-payoff games and min-cost reachability games on certain
subclasses of graphs. More technical
explanations and full proofs may be found in an extended version of
this article \cite{BGHM14}.

\section{Quantitative games with arbitrary weights}
\label{sec:quant-games}

We denote by $\Z$ the set of integers, and
$\Zbar=\Z\cup\{-\infty,+\infty\}$. The set of vectors indexed by
$\vertices$ with values in $S$ is denoted by $S^\vertices$.  We let
$\vleq$ be the pointwise order over $\Zbar^\vertices$, where
$x\vleq y$ if and only if $x(v)\leq y(v)$ for all $v\in \vertices$.

We consider two-player turn-based games on weighted graphs and denote
the two \emph{players} by $\MaxPl$ and $\MinPl$.  A \emph{weighted
  graph} is a tuple $\graphEx$ where
$\vertices=\maxvertices\uplus \minvertices$ is a finite set of
vertices partitioned into the sets $\maxvertices$ and $\minvertices$
of $\MaxPl$ and $\MinPl$ respectively,
$\edges\subseteq \vertices\times \vertices$ is a set of \emph{directed
  edges}, $\edgeweights\colon \edges \to \Z$ is the \emph{weight
  function}, associating an integer weight with each edge. In our
drawings, $\MaxPl$ vertices are depicted by circles; $\MinPl$ vertices
by boxes. For every vertex $v\in\vertices$, the set of successors of
$v$ by $\edges$ is denoted by
$\edges(v) = \{v'\in\vertices\mid (v,v')\in \edges\}$. Without loss of
generality, we assume that every graph is deadlock-free, i.e., for all
vertices~$v$, $\edges(v)\neq\emptyset$. Finally, throughout this
article, we let $W=\max_{(v,v')\in\edges}|\edgeweights(v,v')|$ be the
greatest edge weight (in absolute value) in the game graph. A
\emph{finite play} is a finite sequence of vertices
$\pi=v_0v_1\cdots v_k$ such that for all $0\leq i<k$,
$(v_i,v_{i+1})\in \edges$. A \emph{play} is an infinite sequence of
vertices $\pi = v_0v_1\cdots$ such that every finite prefix
$v_0\cdots v_k$, denoted by $\pi[k]$, is a finite play.

The total-payoff of a finite play $\pi=v_0 v_1 \cdots v_k$ is obtained
by summing up the weights along $\pi$, i.e.,
$\TP(\pi) = \sum_{i=0}^{k-1} \edgeweights(v_i,v_{i+1})$. In the
following, we sometimes rely on the mean-payoff to obtain information
about total-payoff objectives. The \emph{mean-payoff} computes the
average weight of~$\pi$, i.e., if $k\geq 1$,
$\MP(\pi) = \frac{1}{k}\sum_{i=0}^{k-1} \edgeweights(v_i,v_{i+1})$,
and $\MP(\pi)=0$ when $k=0$. These definitions are lifted
to infinite plays as follows. The total-payoff of a play $\pi$ is
given by $\TP(\pi) = \liminf_{k\to \infty} \TP(\pi[k])$.\footnote{Our
  results can easily be extended by substituting a $\limsup$ for the
  $\liminf$. The $\liminf$ is more natural since we adopt the point of
  view of the maximiser $\MaxPl$, hence the $\liminf$ is the
  \emph{worst} partial sum seen infinitely often.}  Similarly, the
mean-payoff of a play $\pi$ is given by
$\MP(\pi) = \liminf_{k\to \infty} \MP(\pi[k])$. A weighted graph
equipped with these payoffs is called a \emph{total-payoff game} or a
\emph{mean-payoff game}, respectively.

A \emph{strategy} for $\MaxPl$ (\resp, $\MinPl$) in a game
$\game=\gameEx$ (with $\Payoff$ one of the previous payoffs), is a
mapping $\strategy\colon \vertices^* \maxvertices \to \vertices$
(\resp, $\strategy\colon \vertices^* \minvertices \to \vertices$) such
that for all sequences $\pi= v_0\cdots v_k$ with $v_k\in \maxvertices$
(\resp, $v_k\in \minvertices$), $(v_k,\strategy(\pi))\in \edges$. A
play or finite play $\pi = v_0v_1\cdots$ conforms to a strategy
$\strategy$ of $\MaxPl$ (\resp, $\MinPl$) if for all $k$ such that
$v_k\in \maxvertices$ (\resp, $v_k\in\minvertices$),
$v_{k+1} = \strategy(\pi[k])$.  A strategy $\strategy$ is
\emph{memoryless} if for all finite plays $\pi, \pi'$, we have
$\strategy(\pi v)=\strategy(\pi' v)$ for all $v$. A strategy
$\strategy$ is said to be \emph{finite-memory} if it can be encoded in
a deterministic Moore machine,
$\tuple{M,m_0,\updatefunction,\decisionfunction}$, where $M$ is a
finite set representing the memory of the strategy, with an initial
memory content $m_0\in M$,
$\updatefunction\colon M\times\vertices \to M$ is a memory-update
function, and
$\decisionfunction\colon M\times \vertices \to \vertices$ a decision
function such that for every finite play $\pi$ and vertex $v$,
$\strategy(\pi v)=\decisionfunction(\memory(\pi v),v)$ where
$\memory(\pi)$ is defined by induction on the length of the finite
play $\pi$ as follows: $\memory(v_0)=m_0$, and
$\memory(\pi v)=\updatefunction(\memory(\pi),v)$.  We say that $|M|$
is the \emph{size} of the strategy.

For all strategies $\maxstrategy$ and $\minstrategy$, for all vertices
$v$, we let $\outcomes(v,\maxstrategy,\minstrategy)$ be the outcome of
$\maxstrategy$ and $\minstrategy$, defined as the unique play
conforming to $\maxstrategy$ and $\minstrategy$ and starting in~$v$.
Naturally, the objective of $\MaxPl$ is to maximise its payoff. In
this model of zero-sum game, $\MinPl$ then wants to minimise the
payoff of $\MaxPl$. Formally, we let $\Val_\game(v,\maxstrategy)$ and
$\Val_\game(v,\minstrategy)$ be the respective values of the
strategies, defined as (recall that $\Payoff$ is either $\TP$ or
$\MP$):
$\Val_\game(v,\maxstrategy) = \inf_{\minstrategy}
\Payoff(\outcomes(v,\maxstrategy,\minstrategy))$
and
$\Val_\game(v,\minstrategy) = \sup_{\maxstrategy}
\Payoff(\outcomes(v,\maxstrategy,\minstrategy))$.
Finally, for all vertices $v$, we let
$\lowervalue_\game(v) = \sup_{\maxstrategy}
\Val_\game(v,\maxstrategy)$
and
$\uppervalue_\game(v) = \inf_{\minstrategy}
\Val_\game(v,\minstrategy)$
be the \emph{lower} and \emph{upper values} of $v$ respectively.  We
may easily show that $\lowervalue_\game\vleq \uppervalue_\game$. We
say that strategies $\maxstrategy^*$ of $\MaxPl$ and $\minstrategy^*$
of $\MinPl$ are optimal if, for all vertices $v$:
$\Val_\game(v,\maxstrategy^*)=\lowervalue_\game(v)$ and
$\Val_\game(v,\minstrategy^*)=\uppervalue_\game(v)$ respectively.
We say that a game $\game$ is \emph{determined} if for all
vertices~$v$, its lower and upper values are equal. In that case, we
write $\Val_\game(v)=\lowervalue_\game(v)=\uppervalue_\game(v)$, and
refer to it as the \emph{value} of~$v$. If the game is clear from the
context, we may drop the index $\game$ of all previous values.
Mean-payoff and total-payoff games are known to be determined, with
the existence of optimal memoryless strategies
\cite{ZwiPat96,GimZie04}.

Total-payoff games have been mainly considered as a refinement of
mean-payoff games~\cite{GimZie04}. Indeed, if the mean-payoff value of
a game is positive (\resp, negative), its total-payoff value is
necessarily $+\infty$ (\resp, $-\infty$). When the mean-payoff value
is $0$ however, the total-payoff is necessarily different from
$+\infty$ and $-\infty$, hence total-payoff games are particularly
useful in this case. Deciding whether the total-payoff value of a
vertex is positive can be achieved in \NPcoNP. In \cite{GawSei09}, the
complexity is refined to UP~$\cap$~co-UP, and values are shown to be
effectively computable solving nested fixed point equations with a
strategy iteration algorithm working in exponential time in the worst
case.

Our aim is to give a pseudo-polynomial algorithm solving total-payoff
games. In many cases, (e.g., mean-payoff games), a successful way to
obtain such an efficient algorithm is the \emph{value iteration
  paradigm}.  Intuitively, value iteration algorithms compute
successive approximations $x_0, x_1, \ldots, x_i, \ldots$ of the game
value by restricting the number of turns that the players are allowed
to play: $x_i$ is the vector of optimal values achievable when the
players play at most $i$ turns. The sequence of values is computed by
means of an operator $\operator$, letting $v_{i+1}=\operator(v_i)$ for
all $i$. Good properties (Scott-continuity and monotonicity) of
$\operator$ ensure convergence towards its smallest or greatest fixed
point (depending on the value of $x_0$), which, in some cases, happens
to be the value of the game. Let us briefly explain why such a simple
approach fails with total-payoff games. In our case, the operator
$\operator$ is such that
$\operator(x)(v)=\max_{v'\in E(v)} \edgeweights (v,v') + x(v')$ for
all $v\in \maxvertices$ and
$\operator(x)(v)=\min_{v'\in E(v)} \edgeweights (v,v') + x(v')$ for
all $v\in \minvertices$. This definition matches the intuition that
$x_i$ are optimal values after $i$ turns.

Then, consider the example of \figurename~\ref{fig:tp}(a), limited to
vertices $\{v_3,v_4,v_5\}$ for simplicity. Observe that there are two
simple cycles with weight~$0$, hence the total-payoff value of this
game is finite. $\MaxPl$ has the choice between cycling into one of
these two cycles. It is easy to check that $\MaxPl$'s optimal choice
is to enforce the cycle between $v_4$ and $v_5$, securing a payoff of
$-1$ from $v_4$ (because of the $\liminf$ definition of $\TP$). Hence,
the values of $x_3$, $x_4$ and $x_5$ are respectively $1$, $-1$
and~$0$. In this game, we have
$\operator(x_3,x_4,x_5) = \big(2+x_4,\max (-2+x_3,-1+x_5),1+x_4\big)$,
and the vector $(1,-1,0)$ is indeed a fixed point
of~$\operator$. However, it is neither the greatest nor the smallest
fixed point of $\operator$, since \emph{if} $x$ is a fixed point of
$\operator$, \emph{then} $x+(a,a,a)$ is also a fixed point, for all
constant $a\in \Z$.  If we try to initialise the value iteration
algorithm with value $(0,0,0)$, which could seem a reasonable choice,
the sequence of computed vectors is: $(0,0,0)$, $(2,-1,1)$, $(1,0,0)$,
$(2,-1,1)$, $(1,0,0)$, $\ldots$ that is not stationary, and does not
even contain $(1,-1,0)$. Thus, it seems difficult to compute the
actual game values with an iterative algorithm relying on the
$\operator$ operator, as in the case of mean-payoff games.\footnote{In
  the context of stochastic models like Markov decision processes,
  Strauch \cite{Str66} already noticed that in the presence of
  arbitrary weights, the value iteration algorithm does not
  necessarily converge towards the accurate value: see
  \cite[Ex.~7.3.3]{Put94} for a detailed explanation.} Notice that, in
the previous example, the Zwick and Paterson's algorithm
\cite{ZwiPat96} to solve mean-payoff games would easily conclude from
the sequence above, since the vectors of interest are then the one
divided by the length of the current sequence, i.e., $(0,0,0)$,
$(1,-0.5,0.5)$, $(0.33,0,0)$, $(0.5,-0.25,0.25)$, $(0.2,0,0)$,
$\ldots$ indeed converging towards $(0,0,0)$, the mean-payoff values
of this game.

Instead, as explained in the introduction, we propose a different
approach that consists in reducing total-payoff games to min-cost
reachability games where $\MinPl$ must enforce a reachability
objective on top of his optimisation objective. The aim of the next
section is to study these games, and we reduce total-payoff games to
them in Section~\ref{sec:solving-total-payoff}.

\section{Min-cost reachability games}
\label{sec:reachability-objectives}

In this section, we consider \emph{min-cost reachability games} (MCR
games for short), a variant of total-payoff games where one player has
a reachability objective that \she must fulfil first, before
optimising his quantitative objective. Without loss of generality, we
assign the reachability objective to player $\MinPl$, as this will
make our reduction from total-payoff games easier to explain. Hence,
when the target is not reached along a path, its payoff shall be the
worst possible for $\MinPl$, i.e., $+\infty$.  Formally, an MCR game
is played on a weighted graph $\graphEx$ equipped with a target set of
vertices $\finalvertices\subseteq \vertices$. The payoff
$\MCR[\finalvertices](\pi)$ of a play $\pi=v_0v_1\ldots$ is given by
$\MCR[\finalvertices](\pi)=+\infty$ if the play avoids
$\finalvertices$, i.e., if for all $k\geq 0$,
$v_k\notin\finalvertices$, and $\MCR[\finalvertices](\pi)=\TP(\pi[k])$
if $k$ is the least position in $\pi$ such that
$v_k\in\finalvertices$. Lower and upper values are then defined as in
Section~\ref{sec:quant-games}. By an indirect consequence of Martin's
theorem~\cite{Mar75}, we can show that MCR games are also determined.
Optimal strategies may however not exist, as we will see later.

As an example, consider the MCR game played on the weighted graph of
\figurename~\ref{fig:Weighted-game}$(a)$, where $W$ is a positive
integer and $v_3$ is the target.
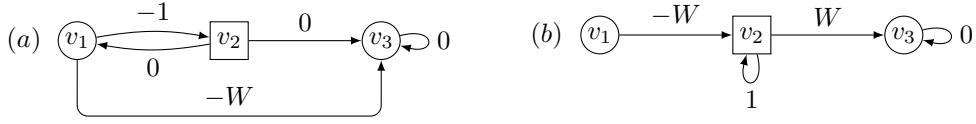
\begin{figure}[tbp]
\begin{center}
  \begin{tikzpicture}[node distance=2cm,auto,->,>=latex]
      \node[player1](1){\makebox[0mm][c]{$v_1$}}; 
      \node[player2](2)[right of=1]{\makebox[0mm][c]{$v_2$}}; 
      \node[player1](3)[right of=2]{\makebox[0mm][c]{$v_3$}};
      
      \draw[rounded corners]
      (1) -- (0, -1) -- node[above]{$-W$} (4, -1) -- (3);

      \path 
      (1) edge[bend left=10] node[above]{$-1$} (2) 
      (2) edge[bend left=10] node[below]{$0$} (1)
          edge node[above]{$0$} (3)
      (3) edge[loop right] node[right]{$0$} (3);

      \node()[left of=1,node distance=.7cm] {$(a)$};
    \end{tikzpicture}
    \qquad 
    \begin{tikzpicture}[node distance=2cm,auto,->,>=latex]
    \node[player1](1){\makebox[0mm][c]{$v_1$}}; 
    \node[player2](2)[right of=1]{\makebox[0mm][c]{$v_2$}}; 
    \node[player1](3)[right of=2]{\makebox[0mm][c]{$v_3$}};
    
    \path 
    (1) edge node[above]{$-W$} (2) 
    (2) edge[loop below] node[below]{$1$} (2)
        edge node[above]{$W$} (3)
    (3) edge[loop right] node[right]{$0$} (3);

    \node()[left of=1,node distance=.7cm] {$(b)$};
  \end{tikzpicture}
\end{center}
\caption{Two weighted graphs}
\label{fig:Weighted-game}
\end{figure}
We claim that the values of vertices $v_1$ and $v_2$ are both
$-W$. Indeed, consider the following strategy for $\MinPl$: during
each of the first $W$ visits to $v_2$ (if any), go to $v_1$; else, go
to $v_3$. Clearly, this strategy ensures that the target will
eventually be reached, and that either
\begin{inparaenum}[$(i)$]
\item edge $(v_1,v_3)$ (with weight $-W$) will eventually be
  traversed; or
\item edge $(v_1,v_2)$ (with weight $-1$) will be traversed at least
  $W$ times.
\end{inparaenum}
Hence, in all plays following this strategy, the payoff will be at
most $-W$. This strategy allows $\MinPl$ to secure $-W$, but \she
cannot ensure a lower payoff, since $\MaxPl$ always has the
opportunity to take the edge $(v_1,v_3)$ (with weight $-W$) instead of
cycling between $v_1$ and $v_2$. Hence, $\MaxPl$'s optimal choice is
to follow the edge $(v_1,v_3)$ as soon as $v_1$ is reached, securing a
payoff of $-W$. The $\MinPl$ strategy we have just given is optimal,
and there is \emph{no optimal memoryless strategy} for $\MinPl$.
Indeed, always playing $(v_2,v_3)$ does not ensure a payoff $\leq -W$;
and, always playing $(v_2,v_1)$ does not guarantee to reach the
target, and this strategy has thus value $+\infty$.

Let us note that Bj\"orklund and Vorobyov introduce in \cite{BjoVor07}
the \emph{longest shortest path problem} (LSP for short) and propose a
pseudo-polynomial time algorithm to solve it. However, their
definition has several subtle but important differences to ours, such
as definition of the payoff of a play (equivalently, the length of a
path). As an example, in the game of
\figurename~\ref{fig:Weighted-game}(a), the play
$\pi=(v_1 v_2)^\omega$ (that never reaches the target) has length
$-\infty$ in their setting, while, in our setting,
$\MCR[\{v_3\}](\pi)=+\infty$. Moreover, even if a pre-treatment would
hypothetically allow one to use the LSP algorithm to solve MCR games,
our solution is simpler to implement with the same worst-case
complexity and heuristics only applicable to our value iteration
solution. We now present our contributions for MCR games:

\begin{theorem}\label{thm:optimal-strategy}
  Let $\game = \gameEx[\MCR[T]]$ be an MCR game.
  \begin{enumerate}
  \item For $v\in V$, deciding whether $\Value(v)=+\infty$ can be
    done in polynomial time.
  \item For $v\in V$, deciding whether $\Value(v)=-\infty$ is as hard
    as mean-payoff, in \NPcoNP and can be achieved in
    pseudo-polynomial time.
  \item If $\Value(v)\neq -\infty$ for all vertices $v\in V$, then
    both players have optimal strategies. Moreover, $\MaxPl$ always
    has a memoryless optimal strategy, while $\MinPl$ may require
    finite (pseudo-polynomial) memory in his optimal strategy.
  \item Computing all values $\Value(v)$ (for $v\in V$), as well as
    optimal strategies (if they exist) for both players, can be done
    in (pseudo-polynomial) time $O(|\vertices|^2 |\edges| W)$.
  \end{enumerate}
\end{theorem}

To prove the first item it suffices to notice that vertices with value
$+\infty$ are exactly those from which $\MinPl$ cannot reach the
target.  Therefore the problem reduces to deciding the winner in a
classical reachability game, that can be solved in polynomial
time~\cite{Tho95}, using the classical \emph{attractor} construction:
in vertices of value $+\infty$, $\MinPl$ may play indifferently, while
$\MaxPl$ has an optimal memoryless strategy consisting in avoiding the
attractor. 

To prove the second item, it suffices first to notice that vertices
with value $-\infty$ are exactly those with a value $<0$ in the
mean-payoff game played on the same graph. On the other hand, we can
show that any mean-payoff game can be transformed (in polynomial time)
into an MCR game such that a vertex has value $<0$ in the mean-payoff
game if and only if the value of its corresponding vertex in the MCR
game is $-\infty$. The rest of this section focuses on the proof of
the third and fourth items. We start by explaining how to compute the
values in pseudo-polynomial, and we discuss optimal strategies
afterward. 

\begin{algorithm}[tbp]

  \DontPrintSemicolon%
  \KwIn{MCR game
    $\gameEx[\MCR]$, $W$ largest weight in absolute
    value}%
  \SetKw{value}{\ensuremath{\mathsf{X}}}
  \SetKw{prevvalue}{\ensuremath{\mathsf{X}_{pre}}}
  
  \BlankLine

  $\value(\target) := 0$\;
  \lForEach{\label{line-init}$v\in\vertices\setminus\{\target\}$}{$\value(v):=+\infty$}

  \Repeat{$\value = \prevvalue$}{%
    $\prevvalue := \value$\;%
    \lForEach{$v\in\maxvertices\setminus\{\target\}$}{$\value(v) :=
      \max_{v'\in\edges(v)}
      \big(\edgeweights(v,v')+\prevvalue(v')\big)$} %
    \lForEach{$v\in\minvertices\setminus\{\target\}$}{$\value(v) :=
      \min_{v'\in\edges(v)}
      \big(\edgeweights(v,v')+\prevvalue(v')\big)$}%
    \lForEach{$v\in\vertices\setminus\{\target\}$ \emph{such
        that} $\value(v) < -(|\vertices|-1) 
      W$\label{line-infty-RT}\label{line-infty}}%
    {$\value(v) := -\infty$\label{line-update}}%
  } %
  \Return{$\value$}
  \caption{Value iteration for min-cost reachability
    games}\label{algo:value-iteration-RT}
\end{algorithm}

\subparagraph*{Computing the values.} From now on, we assume, without
loss of generality, that there is exactly one target vertex denoted by
$\target$, and the only outgoing edge from $\target$ is a self loop
with weight $0$: this is reflected by denoting $\MCR$ the payoff
mapping $\MCR[\{\target\}]$. Our value iteration algorithm for MCR
games is given in \algorithmcfname~\ref{algo:value-iteration-RT}. To
establish its correctness, we rely mainly on the operator~$\operator$,
which denotes the function $\Zbar^\vertices \to \Zbar^\vertices$
mapping every vector $x\in\Zbar^\vertices$ to $\operator(x)$ defined
by $\operator(x)(\target)=0$ and
\[\operator(x)(v) =
    \begin{cases}
      \displaystyle{\max_{v'\in \edges(v)}}
      \big(\edgeweights(v,v')+x(v')\big)
      &\textrm{if } v\in \maxvertices\setminus\{\target\}\\
      \displaystyle{\min_{v'\in \edges(v)}}
      \big(\edgeweights(v,v')+x(v')\big) &\textrm{if } v\in
      \minvertices\setminus\{\target\}
    \end{cases}
\]
More precisely, we are interested in the sequence of iterates
$x_i=\operator(x_{i-1})$\label{fi} of $\operator$ from the initial
vector $x_0$ defined by $x_0(v)=+\infty$ for all $v\neq \target$, and
$x_0(\target)=0$. The intuition behind the sequence
$\left(x_i\right)_{i\geq 0}$ is that \emph{$x_i$ is the value of the
  game if we impose that $\MinPl$ must reach the target within $i$
  steps} (and get a payoff of $+\infty$ if \she fails to do
so). Formally, for a play $\pi=v_0 v_1 \cdots \allowbreak v_i \cdots$,
we let $\boundedWeight{i}(\pi)=\MCR(\pi)$ if $v_k=\target$ for some
$k\leq i$, and $\boundedWeight{i}(\pi)=+\infty$ otherwise. We further
let
$\boundeduppervalue{i}(v)=\inf_{\minstrategy} \sup_{\maxstrategy}
\boundedWeight{i}(\outcomes(v,\maxstrategy,\minstrategy))$
(where $\maxstrategy$ and $\minstrategy$ are respectively strategies
of $\MaxPl$ and $\MinPl$). We can show that the operator $\operator$
allows one to compute the sequence
$(\boundeduppervalue{i})_{i\geq 0}$, i.e., for all $i\geq 0$:
$x_i=\boundeduppervalue{i}$.

Let us first show that the algorithm is \emph{correct} when the values
of all nodes are \emph{finite}. Thanks to this characterisation, and
by definition of $\boundeduppervalue{i}$, it is easy to see that, for
all $i\geq 0$: $x_i=\boundeduppervalue{i}\vgeq \uppervalue = \Value$.
Moreover, $\operator$ is a monotonic operator over the complete
lattice $\Zbar^\vertices$. By Knaster-Tarski's theorem, the fixed
points of $\operator$ form a complete lattice and $\operator$ admits a
greatest fixed point. By Kleene's fixed point theorem, using the
Scott-continuity of $\operator$, this greatest fixed point can be
obtained as the limit of the non-increasing sequence of iterates
$(\operator^i(\overline{x}))_{i\geq 0}$ starting in the maximal vector
$\overline{x}$ defined by $\overline{x}(v)=+\infty$ for all
$v\in \vertices$.  As $x_0=\operator(\overline{x})$, the sequence
$(x_i)_{i\geq 0}$ is also non-increasing (i.e., $x_i\vgeq x_{i+1}$,
for all $i\geq 0$) and converges towards the greatest fixed point
of~$\operator$.  We can further show that the value of the game $\Val$
is actually the greatest fixed point of $\operator$. Moreover, we can
bound the number of steps needed to reach that fixed point (when all
values are finite---this is the point where this hypothesis is
crucial), by carefully observing the possible vectors that can be
computed by the algorithm: the sequence $(x_i)_{i\geq 0}$ is
non-increasing, and stabilises after at most
$(2|\vertices|-1) W |\vertices|+|\vertices|$ steps on $\Val$.

Thus, computing the sequence $(x_i)_{i\geq 0}$ up to stabilisation
yields the values of all vertices in an MCR game \emph{if all values
  are finite}. Were it not for line~\ref{line-infty-RT},
\algorithmcfname~\ref{algo:value-iteration-RT} would compute exactly
this sequence. We claim that
\algorithmcfname~\ref{algo:value-iteration-RT} is correct even when
vertices have values in $\{-\infty, +\infty\}$.
Line~\ref{line-infty-RT} allows to cope with vertices whose value is
$-\infty$: when the algorithm detects that $\MinPl$ can secure a value
small enough from a vertex $v$, it sets $v$'s value to
$-\infty$. Intuitively, this is correct because if $\MinPl$ can
guarantee a payoff smaller than $-(|\vertices|-1)\times W$, \she can
force a negative cycle from which \she can reach $\target$ with an
arbitrarily small value. Hence, one can ensure that, after $i$
iterations of the loop, $x_{i-1} \vgeq {\sf X} \vgeq \Val $, and the
sequence still converges to $\Val$, the greatest fixed point of
$\operator$. Finally, if some vertex $v$ has value $+\infty$, one can
check that ${\sf X}(v)=+\infty$ is an invariant of the loop. From that
point, one can prove the correctness of the algorithm.  Thus, the
algorithm executes $O(|V|^2W)$ iterations. Since each iteration can be
performed in $O(|E|)$, the algorithm has a complexity of
$O(|\vertices|^2 |\edges| W)$, as announced in
Theorem~\ref{thm:optimal-strategy}. As an example, consider the
min-cost reachability game of
\figurename~\ref{fig:Weighted-game}$(a)$. The successive values for
vertices $(v_1,v_2)$ (value of the target $v_3$ is always 0) computed
by the value iteration algorithm are the following:
$(+\infty,+\infty)$, $(+\infty,0)$, $(-1,0)$, $(-1,-1)$, $(-2,-1)$,
$(-2,-2), \ldots, (-W,-W+1)$, $(-W, -W)$. This requires $2W$ steps to
converge (hence a pseudo-polynomial time).

\subparagraph*{Computing optimal strategies for both players.} We now
turn to the proof of the third item of
Theorem~\ref{thm:optimal-strategy}, supposing that every vertex $v$ of
the game has a finite value $\Value(v)\in \Z$ (the case where
$\Value(v)=+\infty$ is delt with the attractor construction).

Observe first that, $\MinPl$ may \emph{need memory} to play optimally,
as already shown by the example in
\figurename~\ref{fig:Weighted-game}$(a)$, where the target is $v_3$.
Nevertheless, let us briefly explain why optimal strategies for
$\MinPl$ always exist, with a memory of pseudo-polynomial size. We
extract from the sequence $(x_i)_{i\geq 0}$ defined above (or
equivalently, from the sequence of vectors $\mathsf X$ of
\algorithmcfname~\ref{algo:value-iteration-RT}) the optimal strategy
$\minstrategy^*$ as follows. Let $k$ be the first index such that
$x_{k+1}=x_{k}$. Then, for every play $\pi$ ending in vertex
$v\in\minvertices$, we let
$\minstrategy^*(\pi) = \argmin_{v'\in\edges(v)}
\big(\edgeweights(v,v')+x_{k-|\pi|-1}(v')\big)$,
if $|\pi|<k$, and
$\minstrategy^*(\pi) = \argmin_{v'\in\edges(v)}
\big(\edgeweights(v,v')+x_0(v')\big)$
otherwise (those $\argmin$ may not be unique, but we can indifferently
pick any of them). Since $\minstrategy$ only requires to know the last
vertex and the length of the prefix up to $k$, and since
$k\leq (2|\vertices|-1) W |\vertices|+|\vertices|$ as explained above,
$\minstrategy^*$ needs a memory of pseudo-polynomial size
only. Moreover, it can be computed with the sequence of vectors
$\mathsf X$ in \algorithmcfname~\ref{algo:value-iteration-RT}. It is
not difficult to verify by induction that this strategy is optimal for
$\MinPl$. While optimal, this strategy might not be practical, for
instance, in the framework of controller synthesis. Implementing it
would require to store the full sequence $(x_i)_{i\geq 0}$ up to
convergence step $k$ (possibly pseudo-polynomial) in a table, and to
query this large table each time the strategy is called. Instead, an
alternative optimal strategy $\minstrategy'$ can be construct, that
consists in playing successively two \emph{memoryless strategies}
$\minstrategy^1$ and $\minstrategy^2$ ($\minstrategy^2$ being given by
the attractor construction). To determine when to switch from
$\minstrategy^1$ to $\minstrategy^2$, $\minstrategy'$ maintains a
counter that is stored in a \emph{polynomial} number of bits, thus the
memory footprints of $\minstrategy'$ and $\minstrategy^*$ are
comparable. However, $\minstrategy'$ is easier to implement, because
$\minstrategy^1$ and $\minstrategy^2$ can be described by a pair of
tables of linear size, and, apart from querying those tables,
$\minstrategy'$ consists only in incrementing and testing the counter
to determine when to switch. Moreover, this succession of two
memoryless strategies allows us to also get some interesting strategy
in case of vertices with values $-\infty$: indeed, we can still
compute this pair of strategies, and simply modify the switching
policy to run for a sufficiently long time to guarantee a value less
than a given threshold. In the following, we call such a strategy a
\emph{switching strategy}.

Finally, we can show that, contrary to $\MinPl$, $\MaxPl$ always has a
\emph{memoryless optimal strategy} $\maxstrategy^*$ defined by
$\maxstrategy^*(\pi)=\argmax_{v'\in
  E(v)}\left(\edgeweights(v,v')+\Val(v')\right)$
for all finite plays $\pi$ ending in $v\in \maxvertices$.  For
example, in the game of \figurename~\ref{fig:Weighted-game}$(a)$,
$\maxstrategy^*(\pi v_2)=v_3$ for all $\pi$, since $\Val(v_3)=0$ and
$\Val(v_1)=-W$.  Moreover, the previously described optimal strategies
can be computed along the execution of
\algorithmcfname~\ref{algo:value-iteration-RT}. Finally, we can show
that, for all vertices $v$, the pair of optimal strategies we have
just defined yields a play
$\outcomes(v,\maxstrategy^*,\minstrategy^*)$ which is
\emph{non-looping}, i.e., never visits the same vertex twice before
reaching the target. For instance, still in the game of
\figurename~\ref{fig:Weighted-game}$(a)$,
$\outcomes(v_1,\maxstrategy^*,\minstrategy^*)= v_1 v_2 v_3^\omega$.

\section{An efficient algorithm to solve total-payoff games}
\label{sec:solving-total-payoff}

We now turn our attention back to total-payoff games (without
reachability objective), and discuss our main contribution. Building
on the results of the previous section, we introduce the \emph{first}
(as far as we know) \emph{pseudo-polynomial time algorithm} for
solving those games in the presence of arbitrary weights, thanks to a
reduction from total-payoff games to min-cost reachability games.  The
MCR game produced by the reduction has size pseudo-polynomial in the
size of the original total-payoff game. Then, we show how to compute
the values of the total-payoff game without building the entire MCR
game, and explain how to deduce memoryless optimal strategies from the
computation of our algorithm.

\subparagraph*{Reduction to min-cost reachability games.}  We provide
a transformation from a total-payoff game $\game=\gameEx[\TP]$ to a
min-cost reachability game $\game^K$ such that the values of $\game$
can be extracted from the values in $\game^K$ (as formalised
below). Intuitively, $\game^K$ simulates the game where players play
in $\game$; $\MinPl$ may propose to stop playing and reach a fresh
vertex $\target$ acting as the target; $\MaxPl$ can then accept, in
which case we reach the target, or refuse at most $K$ times, in which
case the game continues. Structurally, $\game^K$ consists of a
sequence of copies of $\game$ along with some new states that we now
describe formally. We let $\target$ be a fresh vertex, and, for all
$n\geq 1$, we define the min-cost reachability game
$\game^n=\tuple{\vertices^n,\edges^n,\edgeweights^n,\MCR[\{\target\}]}$
where $\maxvertices^n$ (\resp, $\minvertices^n$) consists of $n$
copies $(v,j)$, with $1\leq j\leq n$, of each vertex
$v\in \maxvertices$ (\resp, $v\in \minvertices$) and some
\emph{exterior vertices} $(\exterior,v,j)$ for all $v\in \vertices$
and $1\leq j\leq n$ (\resp, \emph{interior vertices} $(\interior,v,j)$
for all $v\in \vertices$ and $1\leq j\leq n$).  Moreover,
$\maxvertices^n$ contains the fresh target vertex $\target$. Edges are
given by
\begin{align*}
  E^n &= \left \{ ( \target, \target ) \right \} \uplus \left \{ \big
    ( (v,j) ,
    (\interior,v',j) \big ) \mid (v,v')\in E, 1\leq j\leq n \right \} \\
  &\uplus \left \{ \big ( (\interior,v,j) , (v,j)\big ) \mid v\in \vertices,
    1\leq j\leq n \right \} \uplus \left \{ \big ( (\exterior,v,j),
    \target \big ) \mid v\in \vertices,
    1\leq j\leq n \right \} \\
  &\uplus \left \{ \big ( (\interior,v,j) , (\exterior,v,j)\big )
    \mid v\in \vertices, 1\leq j\leq n \right \}\\
  &\uplus \left \{ \big (
    (\exterior,v,j), (v,j-1) \big ) \mid v\in \vertices, 1<j\leq n\right \}.
\end{align*}
All edge weights are zero, except edges $\big ( (v,j) ,
(\interior,v',j) \big)$ that have weight $\edgeweights(v,v')$.

\begin{figure}[tbp]
\newcommand{\TPMCRexampleCopyNumber}[1]{
\fill[gray!20,rounded corners] (2.6,1.8) -- (2.6,-3.8) -- (-.8,-3.8) --
(-.8,1.8) -- cycle;  
\node[player1] (A#1) at (0.2,0.2) {$v_1,#1$};
\node[player2](intA#1) at (1,1){$\interior,v_1,#1$};
\node[player2] (B#1) at (1.7,0) {$v_2,#1$};
\node[player2](intB#1) at (1,-1){$\interior,v_2,#1$};
\node[player2] (C#1) at (1,-2) {$v_3,#1$};
\node[player2](intC#1) at (1,-3){$\interior,v_3,#1$};
\draw[->] (A#1) to node[right,yshift=-.2mm,xshift=-.5mm]{$-1$}  (intB#1);  
\draw[->] (B#1) to (intA#1);  
\draw[->] (A#1) to[bend right] node[midway,left]{$-W$} (intC#1.north west);  
\draw[->] (B#1) to[bend left]  (intC#1.north east);
\draw [->] (C#1) to[bend right]  (intC#1) ;
\draw [->] (intA#1) to (A#1);
\draw [->] (intB#1) to (B#1);
\draw [->] (intC#1) to[bend right] (C#1)
}
\newcommand{\TPMCRexampleExteriorNumber}[1]{
\node[player1](extA#1) at (0,1){$\exterior,v_1,#1$};
\node[player1](extB#1) at (0,-1){$\exterior,v_2,#1$};
\node[player1](extC#1) at (0,-3){$\exterior,v_3,#1$};

\draw[rounded corners] (extA#1) -- (0,0) -- (-1,0) -- (-1,-4) ; 
\draw[rounded corners] (extB#1) -- (0,-2) -- (-1,-2) -- (-1,-4);
\draw[rounded corners] (extC#1) -- (0,-3.9) -- (-1,-3.9) -- (-1,-4);
}
\centering
\scalebox{.68}{
  \begin{tikzpicture}[>=latex]

    \begin{scope}
      \TPMCRexampleCopyNumber{3};
    \end{scope}

    \begin{scope}[xshift=4cm]
      \TPMCRexampleExteriorNumber{3};
    \end{scope}
    
    \begin{scope}[xshift=6cm]
      \TPMCRexampleCopyNumber{2};
    \end{scope}
    
    \begin{scope}[xshift=10cm]
      \TPMCRexampleExteriorNumber{2};
    \end{scope}
    
    \begin{scope}[xshift=12cm]
      \TPMCRexampleCopyNumber{1};
    \end{scope}
    
    \begin{scope}[xshift=16cm]
      \TPMCRexampleExteriorNumber{1};
    \end{scope}
    
    \node[player1](TARGET) at (9,-5) {\makebox[0mm][c]{$\target$}};
    \path[->] (TARGET) edge[out=-30,in=30,loop] (TARGET);
    
    \draw[->] (intA3) to (extA3); 
    \draw[->] (intB3) to (extB3);
    \draw[->] (intC3) to (extC3); 
    
    \draw[->] (intA2) to (extA2); 
    \draw[->] (intB2) to (extB2);
    \draw[->] (intC2) to (extC2); 
    
    \draw[->] (intA1) to (extA1); 
    \draw[->] (intB1) to (extB1);
    \draw[->] (intC1) to (extC1); 
    
    \draw[->] (extA3) to (A2); 
    \draw[->,bend right=-5] (extB3) to (B2);
    \draw[->] (extC3) to (C2); 
    
    \draw[->] (extA2) to (A1); 
    \draw[->,bend right=-5] (extB2) to (B1);
    \draw[->] (extC2) to (C1); 
    
    \draw[rounded corners] (3,-4) -- (3,-4.25) -- (9,-4.25) -- (TARGET);
    
    \draw[->,rounded corners] (9,-4)  -- (TARGET);
    
    \draw[rounded corners] (15,-4) -- (15,-4.25) -- (9,-4.25) -- (TARGET);

\end{tikzpicture}}
\caption{MCR game $\game^3$ associated with the total-payoff game of
  \figurename~\ref{fig:Weighted-game}$(a)$}
\label{FigureExTPMCR}
\end{figure}
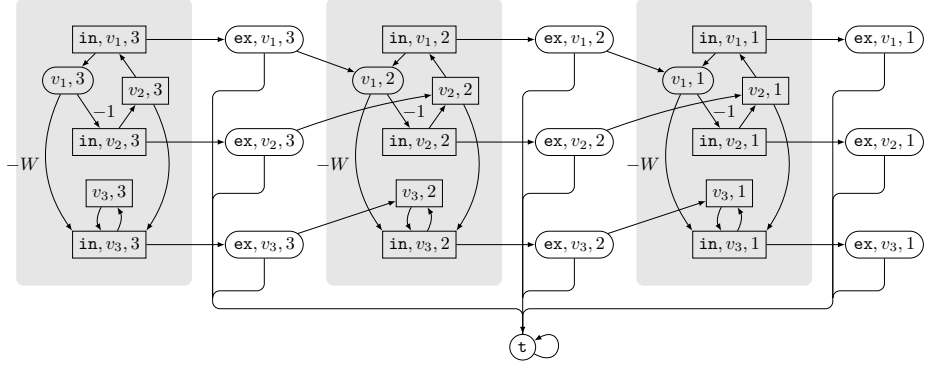
For example, considering the weighted graph of
\figurename~\ref{fig:Weighted-game}$(a)$, the corresponding
reachability total-payoff game $\game^3$ is depicted in
\figurename~\ref{FigureExTPMCR} (where weights $0$ have been
removed). The next proposition formalises the relationship between the
two games.

\begin{proposition}\label{TrueTP2MCR}
  Let $K=|\vertices| (2 (|\vertices|-1) W +1)$. For all $v\in\vertices$
  and $k\geq K$,
  \begin{itemize}
  \item $\Val_\game(v)\neq +\infty$ if and only if
    $\Val_\game(v)=\Val_{\game^k}((v,k))$;
  \item $\Val_\game(v)=+\infty$ if and only if
    $\Val_{\game^k}((v,k))\geq (|\vertices|-1) W+1$.
  \end{itemize}
\end{proposition}

The bound $K$ is found by using the fact (informally described in the
previous section) that if not infinite, the value of a min-cost
reachability game belongs in
$[-(|\vertices|-1)\times W+1, |\vertices|\times W]$, and that after
enough visits of the same vertex, an adequate loop ensures that
$\game^k$ verifies the above properties.

\subparagraph*{Value iteration algorithm for total-payoff games.}  By
Proposition~\ref{TrueTP2MCR}, an immediate way to obtain a value
iteration algorithm for total-payoff games is to build game $\game^K$,
run \algorithmcfname~\ref{algo:value-iteration-RT} on it, and map the
computed values back to $\game$. We take advantage of the structure of
$\game^K$ to provide a better algorithm that avoids building
$\game^K$. We first compute the values of the vertices in the
\emph{last copy of the game} (vertices of the form $(v,1)$,
$(\interior,v,1)$ and $(\exterior,v,1)$), then of those in the
penultimate (vertices of the form $(v,2)$, $(\interior,v,2)$ and
$(\exterior,v,2)$), and so on.

\begin{figure}[tbp]
\centering
\begin{tikzpicture}[>=latex,scale = 0.8, transform shape]
\node[player1] (A) at (0,0) {\makebox[0mm][c]{$v_1$}};
\node[player2](intA) at (1,1){$\interior,v_1$};
\node[player2] (B) at (2,0) {\makebox[0mm][c]{$v_2$}};
\node[player2](intB) at (1,-1){$\interior,v_2$};
\node[player1](tg) at (4,-1) {\makebox[0mm][c]{$\target$}};
\node[player2] (C) at (1,-2) {\makebox[0mm][c]{$v_3$}};
\node[player2](intC) at (1,-3){$\interior,v_3$};
\draw[->] (A) to node[above right]{$-1$}  (intB);  
\draw[->] (B) to (intA);  
\draw[->] (A) to[bend right] node[midway,left]{$-W$} (intC.north west);  
\draw[->] (B) to[bend left]  (intC.north east);
\draw [->] (C) to[bend right]  (intC) ;
\draw [->] (intA) to (A);
\draw [->] (intB) to (B);
\draw [->] (intC) to[bend right] (C);
\draw [->] (intB) -- node[midway,below]{\footnotesize $\qquad \begin{array}{c}\max(0,\\ \;\;Y(v_1))\end{array}$}(tg); 
\draw [rounded corners,->] (intA) -- node[midway,below ]{\footnotesize $\ \ \ \max(0,Y(v_2))$} (4,1) -- (tg);
\draw [rounded corners,->] (intC) -- node[midway,above]{\footnotesize $\ \ \ \max(0,Y(v_3))$} (4,-3) -- (tg); 
\end{tikzpicture}
\caption{MCR game $\game_Y$ associated with the total-payoff game of
  \figurename~\ref{fig:Weighted-game}$(a)$}
\label{fig:gameX}
\end{figure}
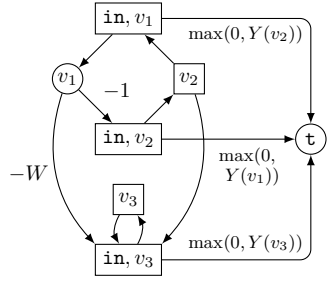
We formalise this idea as follows. Let $Z^j$ be a vector mapping each
vertex $v$ of $\game$ to the value $Z^j(v)$ of vertex $(v,j)$ in
$\game^K$. Then, let us define an operator $\operatorBis$ such that
$Z^{j+1}=\operatorBis(Z^j)$.  The intuition behind the definition of
$\operatorBis(Y)$ for some vector $Y$, is to extract from $\game^K$
one copy of the game, and make $Y$ appear in the weights of some edges
as illustrated in \figurename~\ref{fig:gameX}.  This game,
$\game_Y$, \label{GY} simulates a play in $\game$ in which $\MinPl$
can opt for `leaving the game' at each round (by moving to the
target), obtaining $\max(0,Y(v))$, if $v$ is the current vertex. Then
$\operatorBis(Y)(v)$ is defined as the value of $v$ in $\game_Y$. By
construction, it is easy to see that $Z^{j+1}=\operatorBis(Z^j)$ holds
for all $j\geq 1$. Furthermore, we define $Z^0(v)=-\infty$ for all
$v$, and have $Z^1= \operatorBis(Z^0)$. One can prove the following
properties of $\operatorBis$:
\begin{inparaenum}[$(i)$]
\item $\operatorBis$ is monotonic, but may not be Scott-continuous;
\item the sequence $(Z^j)_{j\geq 0}$ converges towards $\Value_\game$.
\end{inparaenum}

We are now ready to introduce
\algorithmcfname~\ref{algo:value-iter-TPO} to solve total-payoff
games. Intuitively, the outer loop computes, in variable
${\mathsf Y}$, a non-decreasing sequence of vectors whose limit is
$\Value_\game$, and that is stationary (this is not necessarily the
case for the sequence $(Z^j)_{j\geq 0}$).  Line~\ref{line:tpo-init}
initialises ${\mathsf Y}$ to $Z^0$. Each iteration of the outer loop
amounts to running \algorithmcfname~\ref{algo:value-iteration-RT} to
compute $\operatorBis({\sf Y}_{pre})$ (lines~\ref{line:begin}
to~\ref{line:end}), then detecting if some vertices have value
$+\infty$, updating ${\mathsf Y}$ accordingly
(line~\ref{line:3-line-infty}, following the second item of
Proposition~\ref{TrueTP2MCR}). One can show that, for all $j> 0$, if
we let $Y^j$ be the value of ${\mathsf Y}$ after the $j$-th iteration
of the main loop, $Z^j\vleq Y^j \vleq \Value_\game$, which ensures the
correctness of the algorithm.

\begin{algorithm}[tbp]
  \DontPrintSemicolon %
  \KwIn{Total-payoff game $\game=\gameEx[\TP]$, $W$ largest weight in absolute
    value}
  \SetKw{value}{\ensuremath{\mathsf{Y}}}%
  \SetKw{prevvalue}{\ensuremath{\mathsf{Y}_{pre}}}%
  \SetKw{valueint}{\ensuremath{\mathsf{X}}}%
  \SetKw{prevvalueint}{\ensuremath{\mathsf{X}_{pre}}}%
  \BlankLine
  
  \lForEach{$v\in\vertices$}{$\value(v) := -\infty$}%
  \label{line:tpo-init}%
  \Repeat{$\value=\prevvalue$}{ %
    \lForEach{$v\in\vertices$}{$\prevvalue(v):=\value(v)$; $\value(v):= \max
      (0,\value(v))$; $\valueint(v):=+\infty$\label{line:begin}\label{line:tpo-init-y-v}}
    \Repeat{$\valueint=\prevvalueint$}{%
      $\prevvalueint:=\valueint$\;%
      \lForEach{$v\in\maxvertices$}{$\valueint(v) := \max_{v'\in \edges(v)}
        \big[\edgeweights(v,v')+\min(\prevvalueint(v'),\value(v'))\big]$\label{line:3-7}\label{line:minvertex}} %
      \lForEach{$v\in\minvertices$}{$\valueint(v) := \min_{v'\in\edges(v)}
        \big[\edgeweights(v,v')+\min(\prevvalueint(v'),\value(v'))\big]$\label{line:maxvertex}} %
      \lForEach{$v\in\vertices$ \emph{such that} $\valueint(v) <
        -(|\vertices|-1) W$}%
      {$\valueint(v) := -\infty$} %
    } %
    $\value:=\valueint$\label{line:end}\;%
    \lForEach{$v\in\vertices$ \emph{such that} $\value(v) >
      (|\vertices|-1) W$\label{line:3-line-infty}}%
    {$\value(v) := +\infty$}%
  }%
  
  \Return{$\value$} \;

  \caption{A value iteration algorithm for total-payoff
    games\label{algo:value-iter-TPO}}
    
\end{algorithm}

\begin{theorem}\label{thm:VI-TP}
  If a total-payoff game $\game=\gameEx[\TP]$ is given as input,
  \algorithmcfname~\ref{algo:value-iter-TPO} outputs the vector
  $\Value_\game$ of optimal values, after at most
  $K=|\vertices| (2 (|\vertices|-1) W+1)$ iterations of the external
  loop. The complexity of the algorithm is
  $O(|\vertices|^4 |\edges| W^2)$.
\end{theorem}

The number of iterations in each internal loop is controlled by
Theorem~\ref{thm:optimal-strategy}. On the example of
\figurename~\ref{fig:Weighted-game}$(a)$, only 2 external iterations
are necessary, but the number of iterations of each internal loop
would be $2W$. By contrast, for the total-payoff game depicted in
\figurename~\ref{fig:Weighted-game}$(b)$, each internal loop requires
2 iterations to converge, but the external loop takes $W$ iterations
to stabilise. A combination of both examples would experience a
pseudo-polynomial number of iterations to converge in both the
internal and external loops, matching the $W^2$ term of the above
complexity.

\subparagraph*{Optimal strategies.} In
Section~\ref{sec:reachability-objectives}, we have shown, for any
min-cost reachability game, the existence of a pair of memoryless
strategies permitting to reconstruct a \emph{switching} optimal
strategy for $\MinPl$ (if every vertex has value different from
$-\infty$, or a strategy ensuring any possible threshold for vertices
with value $-\infty$). If we apply this construction to the game
$\game_{\Value_\game}$, we obtain a pair
$(\minstrategy^1,\minstrategy^2)$ of strategies (remember that
$\minstrategy^2$ is a strategy obtained by the attractor construction,
hence it will not be useful for us for total-payoff games). Consider
the strategy $\bar\minstrategy$, obtained by projecting
$\minstrategy^1$ on $\vertices$ as follows: for all finite plays $\pi$
and vertex $v\in\minvertices$, let $\bar\minstrategy(\pi v) = v'$ if
$\minstrategy^1(v)=(\interior,v')$. We can show that
$\bar\minstrategy$ is optimal for $\MinPl$ in $\game$. Notice that
$\minstrategy^1$, and hence $\bar\minstrategy$, can be computed during
the last iteration of the value iteration algorithm, as explained in
the case of min-cost reachability. A similar construction can be done
to compute an optimal strategy of $\MaxPl$.

\section{Implementation and heuristics}
\label{sec:experiments}

In this section, we report on a prototype implementation of our
algorithms.\footnote{Source and binary files, as well as some
  examples, can be downloaded from
  \url{http://www.ulb.ac.be/di/verif/monmege/tool/TP-MCR/}.} For
convenience reasons, we have implemented them as an add-on to
PRISM-games \cite{CheFor13}, although we could have chosen to extend
another model-checker as we do not rely on the probabilistic features
of PRISM models (i.e., we use the PRISM syntax of \emph{stochastic
  multi-player games}, allowing arbitrary rewards, and forbidding
probability distributions different of Dirac ones). We then use rPATL
specifications of the form
$\langle\!\langle C \rangle\!\rangle \mathsf R^{\min/\max=?}[\mathsf
F^\infty \varphi]$
and
$\langle\!\langle C \rangle\!\rangle \mathsf R^{\min/\max=?}[\mathsf
F^c \bot]$
to model respectively min-cost reachability games and total-payoff
games, where $C$ represents a coalition of players that want to
minimise/maximise the payoff, and $\varphi$ is another rPATL formula
describing the target set of vertices (for total-payoff games, such a
formula is not necessary). We have tested our implementation on toy
examples. On the parametric one studied after Theorem~\ref{thm:VI-TP},
obtained by mixing the graphs of \figurename~\ref{fig:Weighted-game}
and repeating them for $n$ layers, results obtained by applying our
algorithm for total-payoff games are summarised in
Table~\ref{table:res}, where for each pair $(W,n)$, we give the time
$t$ in seconds, the number $k_e$ of iterations in the external loop,
and the total number $k_i$ of iterations in the internal loop.

\begin{table}[tbp]
  \caption{Results of value iteration on a parametric
    example\label{table:res}}
  \centering
  {\setlength{\tabcolsep}{.5em}
  \begin{tabular}{|c|c||c|c|c||c|c|c|}
    \hline
    \multicolumn{2}{|c||}{}& \multicolumn{3}{c||}{without heuristics} 
             & \multicolumn{3}{c|}{with heuristics} \\ %
    $W$ & $n$ & $t$ & $k_e$ & $k_i$ & $t$ & $k_e$ & $k_i$ \\\hhline{|=|=#=|=|=#=|=|=|} %
    50 & 100 & 0.52s & 151 & 12,603 & 0.01s & 402 & 1,404 \\ \hline %
    50 & 500 & 9.83s & 551 & 53,003 & 0.42s & 2,002 & 7,004 \\\hline %
    200 & 100 & 2.96s & 301 & 80,103 & 0.02s & 402 & 1,404 \\\hline %
    200 & 500 & 45.64s & 701 & 240,503 & 0.47s & 2,002 & 7,004 \\\hline %
    500 & 1,000 & 536s & 1,501 & 1,251,003 & 2.37s & 4,002 & 14,004 \\\hline %
  \end{tabular}}
\end{table}

We close this section by sketching two techniques that can be used to
speed up the computation of the fixed point in
\algorithmcfname{s}~\ref{algo:value-iteration-RT} and
\ref{algo:value-iter-TPO}. We fix a weighted graph $\graphEx$. Both
accelerations rely on a topological order of the strongly connected
components (SCC for short) of the graph, given as a function
$\decomposition\colon \vertices\to \N$, mapping each vertex to its
\emph{component}, verifying that
\begin{inparaenum}[$(i)$]
  \item $\decomposition(\vertices)=\{0,\ldots,p\}$ for some $p\geq 0$,
  \item $\decomposition^{-1}(q)$ is a maximal SCC for all $q$, 
  \item and $\decomposition(v)\geq \decomposition(v')$ for all
    $(v,v')\in\edges$.\footnote{Such a mapping is computable in
      linear time, e.g., by Tarjan's algorithm \cite{Tar72}.}
\end{inparaenum}
In case of an MRC game with $\target$ the unique target,
$\decomposition^{-1}(0) = \{\target\}$.  Intuitively, $\decomposition$
induces a directed acyclic graph whose vertices are the sets
$\decomposition^{-1}(q)$ for all $q\in\decomposition(\vertices)$, and
with an edge $(S_1,S_2)$ if and only if there are
$v_1\in S_1, v_2\in S_2$ such that $(v_1,v_2)\in\edges$.

The \emph{first acceleration heuristic} is a divide-and-conquer
technique that consists in applying
\algorithmcfname~\ref{algo:value-iteration-RT} (or the inner loop of
\algorithmcfname~\ref{algo:value-iter-TPO}) iteratively on each
$\decomposition^{-1}(q)$ for $q=0,1,2,\ldots,p$, using at each step
the information computed during steps $j<q$ (since the value of a
vertex $v$ depends only on the values of the vertices $v'$ such that
$\decomposition(v')\leq\decomposition(v)$).  The \emph{second
  acceleration heuristic} consists in studying more precisely each
component $\decomposition^{-1}(q)$. Having already computed the
optimal values $\Value(v)$ of vertices
$v\in \decomposition^{-1}(\{0,\ldots,q-1\})$, we ask an oracle to
precompute a finite set $S_v\subseteq\Zbar$ of possible optimal values
for each vertex $v\in\decomposition^{-1}(q)$. For MCR games and the
inner iteration of the algorithm for total-payoff games, one way to
construct such a set $S_v$ is to consider that possible optimal values
are the one of non-looping paths inside the component exiting it,
since, in MCR games, there exist optimal strategies for both players
whose outcome is a non-looping path (see
Section~\ref{sec:reachability-objectives}).

We can identify classes of weighted graphs for which there exists an
oracle that runs in polynomial time and returns, for all vertices~$v$,
a set $S_v$ of polynomial size. On such classes,
\algorithmcfname{s}~\ref{algo:value-iteration-RT} and
\ref{algo:value-iter-TPO}, enhanced with our two acceleration
techniques, \emph{run in polynomial time}. For instance, for all fixed
positive integers $L$, the class of weighted graphs where every
component $\decomposition^{-1}(q)$ uses at most $L$ distinct weights
(that can be arbitrarily large in absolute value) satisfies this
criterion. Table~\ref{table:res} contains the results obtained with
the heuristics on the parametric example presented before. Observe
that the acceleration technique permits here to decrease drastically
the execution time, the number of iterations in both loops depending
not even anymore on $W$.  Even though the number of iterations in the
external loop increases with heuristics, due to the decomposition,
less computation is required in each internal loop since we only apply
the computation for the active component.

\label{EndOfArticle}

\clearpage
\appendix

\changepage{4cm}{3cm}{-1.5cm}{-1.5cm}{}{-1cm}{}{}{}

\section{Determinacy of total-payoff and min-cost reachability
  games\label{sec:determinacy}}

Consider a quantitative game $\game=\gameEx$ and a vertex
$v\in\vertices$. We will prove the determinacy result by using the
Borel determinacy result of \cite{Mar75}. Indeed, for an integer $M$,
consider $\Win_M$ to be the set of infinite plays with a payoff less
than or equal to $M$. Payoff mapping $\MCR[\finalvertices]$ is easily
shown to be Borel measurable since the set of plays with finite payoff
is then a countable union of open sets of plays. For $\TP$, it is
shown by considering a pointwise limit of Borel measurable
functions. Therefore, for payoff $\Payoff$ representing total-payoff
and min-cost reachability, we know that $\Win_M$ is a Borel set,
so that the qualitative game defined over the graph $\graphEx$ with
winning condition $\Win_M$ is determined. We now use this preliminary
result to show our determinacy result.

We first consider cases where lower or upper values are
infinite. Suppose first that $\lowervalue(v)=-\infty$. We have to show
that $\uppervalue(v)=-\infty$ too. Let $M$ be an integer. From
$\lowervalue(v)<M$, we know that for all strategy $\maxstrategy$ of
$\MaxPl$, there exists a strategy $\minstrategy$ for $\MinPl$, such that
$\Payoff(\outcomes(v,\maxstrategy,\minstrategy))\leq M$. In
particular, $\MaxPl$ has no winning strategy in the qualitative game
equipped with $\Win_M$ as a winning condition. By the previous
determinacy result, we know that $\MinPl$ has a winning strategy in
that case, i.e., a strategy $\minstrategy$ such that every strategy
$\maxstrategy$ of $\MaxPl$ verifies
$\Payoff(\outcomes(v,\maxstrategy,\minstrategy))\leq M$. This exactly
means that $\uppervalue(v)\leq M$. Since this holds for every value
$M$, we get that $\uppervalue(v)=-\infty$. The proof goes exactly in a
symetrical way to show that $\uppervalue(v)=+\infty$ implies
$\lowervalue(v)=+\infty$.

Consider then the case where both $\uppervalue(v)$ and
$\lowervalue(v)$ are finite values. Suppose that
$\lowervalue(v)<\uppervalue(v)$ and consider a real number $r$
strictly in-between those two values. From $r<\uppervalue(v)$, we
deduce that $\MinPl$ has no winning strategy from $v$ in the
qualitative game with winning condition $\Win_r$. Identically, from
$\lowervalue(v)<r$, we deduce that $\MaxPl$ has no winning strategy
from $v$ in the same game. This contradicts the determinacy of this
qualitative game. Hence, $\lowervalue(v)=\uppervalue(v)$.

\section{Comparison with the \emph{longest shortest path problem} of
  Bj\"orklund and Vorobyov\label{sec:comparisonLSP}}

Bj\"orklund and Vorobyov have studied games related to our min-cost
reachability games, as the longest shortest path problem (LSP for
short) in \cite{BjoVor07}. To clarify the following discussion, let us
recall the definition of LSP, adapted to our syntax. 

The LSP problem considers a weighted graph $\game$, whose vertices are
partitioned into two players, and equipped with a single target vertex
$\target$. The problem asks to find a \emph{memoryless} strategy
$\maxstrategy$ of $\MaxPl$ such that in the graph
$\game_{\maxstrategy}$, obtained from $\game$ by deleting all outgoing
edges from vertices of $\maxvertices$ except those selected in
$\maxstrategy$, the length of the shortest path from every vertex to
the target $\target$ is as large as possible (over all memoryless
strategies). The definition of paths from a vertex $v$ to the target
is, however, very different from ours: such a path may indeed be a
finite play, without cycle, from $v$ to the target $\target$ (in which
case its length is the sum of the weights of edges). However, any
cycle containing $v$ (even if it cannot reach the target) is also
considered as a path to the target. The \emph{length} of such cycle is
defined as follows:
\begin{enumerate}
\item if the cycle has a negative weight, its length is $-\infty$;
\item else, if it has a positive weight, its length is $+\infty$;
\item else, if it has zero weight, its length is $0$.
\end{enumerate}

Consider, as an example, the weighted graph of Fig.~\ref{fig:LSP},
where, once again, $\MaxPl$ vertices are depicted with
circles. Clearly, vertices $\{v_1,v_2,v_3\}$ and $\{v_4\}$ form
independant subgames. We first consider the status of $v_1$, $v_2$ and
$v_3$. In this subgame, $\MaxPl$ has two possible strategies:
$\sigma_1$ that selects the $(v_1,v_2)$ edge, or $\sigma_2$ that
selects the $(v_1,v_3)$ edge instead.
\begin{enumerate}
\item If we select $\sigma_1$, we obtain the following values for
  $v_1$, $v_2$ and $v_3$ respectively: $0$, $0$ and $1$. Indeed, from
  vertex $v_2$, we can loop in the zero cycle in-between $v_1$ and
  $v_2$, which is better than the simple path to the target (having
  weight 3), so that the distance from $v_1$ and $v_2$ is $0$. From
  vertex $v_3$, only a single path leads to the target with weight
  $1$.\todo{B: I am now wondering if their semantics would consider
    the path from $v_3$ to $v_1$ then looping in the cycle in-between
    $v_1$ and $v_2$... This would not change the value in that case
    anyway but this could be another difference to point out.}
\item If we select $\sigma_2$, we obtain the following values for
  $v_1$, $v_2$ and $v_3$ respectively: $1$, $3$ and $1$. Indeed, from
  vertex $v_3$, the distance is the length of the simple path to the
  target, which is better than looping in-between $v_1$ and $v_3$ with
  a positive weight cycle. Hence, from vertex $v_1$, the distance is
  $1$ too, whereas from $v_2$, the distance is $3$.
\end{enumerate}
Now, let us turn our attention to $v_4$. Again, we must consider two
strategies: strategy $\sigma_3$ selects the $(v_4,v_4)$ edge (the
self-loop on $v_4$) and strategy $\sigma_4$ selects the
$(v_4,\target)$ edge.
\begin{enumerate}
\item If we select $\sigma_3$, we obtain value $-\infty$ for $v_4$,
  because all plays starting in $v_4$ loop in this negatively priced
  cycle. Recall that with our definition of min-cost reachability,
  such a cycle that never reaches the target would yield a
  total-payoff of $+\infty$, regardless of the weights. This is
  coherent with our intuition that the first goal of $\MinPl$ is to
  reach the target.
\item If we select $\sigma_4$, the value of $v_4$ is now $0$, which is
  thus better for $\MaxPl$.
\end{enumerate}
Finally, the optimal strategy for $\MaxPl$ is to choose edges
$(v_1,v_3)$ and $(v_4,\target)$, leading to the LSP distances
$(1,3,1,0)$ for vertices $(v_1,v_2,v_3,v_4)$. With our definition of
min-cost reachability game, the situation is \emph{completely
  different}: the value vector is then $(2,3,1,+\infty)$, associated
with the optimal strategy $\maxstrategy$ selecting edges $(v_1,v_2)$
and $(v_4,v_4)$.

Indeed, two main differences separate our two definitions. The first
one is the treatment of negative weight cycles. Actually, in LSP, the
fact that after selecting strategy $\maxstrategy$ a vertex cannot
reach the target anymore in $\game_{\maxstrategy}$ does not prevent
from mapping the distance $-\infty$ to a vertex contained in a
negative cycle. This is in contrast with our definition, that would
benefit to $\MaxPl$ since in such a situation, the target is not
reachable leading to the value $+\infty$.\footnote{We believe that
  this difference would certainly be eliminated by our precomputation
  of vertices of value $+\infty$ presented in the first item of
  Theorem~\ref{thm:optimal-strategy}.} The second difference consists
in the treatment of zero weight cycle. In LSP, a distance zero is then
computed, which is highly related with the fact that the authors want
to apply the resolution of LSP to mean-payoff games.

\begin{figure}[tbp]
  \centering
  \begin{tikzpicture}[node distance=2.5cm,auto,->,>=latex]
    \node[player1](1){\makebox[0mm][c]{$v_1$}}; 
    \node[player2](2)[above right of=1,yshift=-.7cm]{\makebox[0mm][c]{$v_2$}}; 
    \node[player2](3)[below right of=1,yshift=.7cm]{\makebox[0mm][c]{$v_3$}};
    \node[player1](4)[right of=2]{\makebox[0mm][c]{$v_4$}};
    \node[player1](tg)[right of=3]{\makebox[0mm][c]{$\target$}};
    
    \path 
    (1) edge[bend left=20] node[above left]{$-1$} (2) 
        edge[bend right=20] node[below left]{$0$} (3)
    (2) edge node[below right]{$1$} (1)
        edge node[above right]{$3$} (tg)
    (3) edge node[above right]{$1$} (1)
        edge node[above]{$1$} (tg)
    (4) edge[loop right] node[right]{$-1$} (4)
        edge node[right]{$0$} (tg)
    (tg) edge[loop right] node[right]{$0$} (tg);
  \end{tikzpicture}
  \caption{An instance of the longest shortest path problem}
  \label{fig:LSP}
\end{figure}
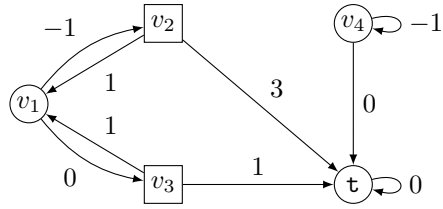

Nevertheless, in Section 9 of \cite{BjoVor07}, the authors
\emph{briefly} study another more natural definition, mapping zero
weight cycle to the distance $+\infty$ (closer to our definition). In
that case, the LSP distances of vertices $v_1$, $v_2$ and $v_3$ then
match our value vector in min-cost reachability games. Unfortunately,
it is not clear how the algorithm that is presented by Bj\"orklund and
Vorobyov can be adapted to accomodate this new definition of the cost
and compute the values of the nodes. Indeed, Proposition~9.1 of
\cite{BjoVor07} proves that the decision version of the LSP (i.e.,
deciding if the LSP of a vertex is positive or not in a given graph)
is in \NPcoNP, and they claim (without formal proof) that their
pseudo-polynomial time algorithm permits to decide this problem. This
requires a first transformation of the problem, so that no zero weight
cycles remain in the weighted graph. This transformation, explained
above Proposition~9.1, is correct but\emph{ does not preserve the LSP
  of the vertices}, so that their algorithm do not compute the LSP
distance of vertices, but only study their positivity.

Hence, the comparison between our definition of min-cost reachability
and the definition of LSP can be summarised as follows:
\begin{enumerate}
\item Either one relies on the first definition of LSP given in the
  paper. In this case, Bj\"orklund and Vorobyov propose a
  pseudo-polynomial time algorithm to compute the LSP, but these
  values \emph{do not match our definition of the min-cost
    reachability} payoff, as demonstrated by the above example.
\item Or one relies on the second definition. In this case, we believe
  that the LSP values may match our definition of min-cost
  reachability but the paper \emph{does not explain formally how to
    compute those new values}.
\end{enumerate}

It is plausible that the pseudo-polynomial time algorithm of
Bj\"orklund and Vorobyov could be adapted to compute the LSP value
according to the second definition. Yet, even in this case, there are
several points of comparison worth mentioning:
\begin{enumerate}
\item The worst-case complexity of the algorithm proposed by
  Bj\"orklund and Vorobyov is $O(|\vertices|^2 |E| W)$, which matches
  ours (see Proposition~\ref{the:rtpo-ca-marche-youpi}). Nevertheless,
  our solution is much easier to describe and to implement: while they
  must make repeated calls to a modified version of the Bellman-Ford
  algorithm (the modification is crucial to obtain their complexity),
  we have a simple fixed point algorithm.
\item Our algorithm exploits the value iteration paradigm (that can be
  seen as a \emph{backward induction}), while theirs is a
  \emph{strategy iteration algorithm}. Because of that, there are
  examples on which our algorithm is more efficient that theirs
  (although the worst-case complexity is the same). As an example, the
  LSP instance presented in Fig.~2 of \cite{BjoVor07}, with $2n+1$
  vertices but a biggest weight $W$ exponential in $n$, requires an
  exponential number $2^n$ of iterations for the strategy iteration
  algorithm, but our value iteration algorithm (based on backward
  induction) would solve it in linear time with respect to $n$. More
  precisely, with our tool, we are able to compute the values of this
  game in less than a millisecond for constant $n=15$ (i.e., a game
  with $32$ vertices and largest weight $W\approx 2.68\cdot 10^8$):
  the value of the initial state obtained is $4.74 \cdot 10^9$, and as
  aforementioned, the number of iterations that our value iteration
  requires is $16$, linear in $n$. Our tool being an add-on of the
  PRISM model-checker, relying on the use of integer rewards, we have
  not been able to build the game for a value $n$ greater than $15$.
\item We propose several acceleration heuristics that perform well on
  several examples we have tried. These accelerations cannot easily be
  incorporated in their algorithm (because it is a strategy iteration
  algorithm).
\item Finally, from the theoretical point of view, Bj\"orklund and Vorobyov do
  not study the strategies of the opponent player $\MinPl$. In
  contrast, our study allows us to produce optimal strategies for both
  players, and in particular, show that memoryless strategies may not
  be sufficient for $\MinPl$, whereas they are enough for $\MaxPl$.
\end{enumerate}

\section{\texorpdfstring{Finding vertices of value $+\infty$ in
    min-cost reachability games}{Finding vertices of positive infinity
    value in min-cost reachability games}\label{sec:+infty}}

Let $\game=\gameEx[\MCR[\finalvertices]]$ be a min-cost reachability
game. Notice that for all plays $\pi=v_0v_1\cdots$,
$\MCR[\finalvertices](\pi)=+\infty$ if and only if
$v_k\notin \finalvertices$ for all $k\geq 0$, i.e., $\pi$ avoids the
target. Then, let us show that the classical attractor technique
\cite{Tho95} allows us to compute the set
$\vertices_{+\infty}=\{v\in\vertices\mid \Value(v)=+\infty\}$. Recall
that the attractor of a set $\finalvertices$ of vertices is obtained
thanks to the sequence
$\Attr_0(\finalvertices), \ldots, \Attr_i(\finalvertices),\ldots$
where: $\Attr_0(\finalvertices) = \finalvertices$; and
for all $i\geq 0$: 
\[\Attr_{i+1}(\finalvertices) =
\Attr_i(\finalvertices) \cup \{v\in\minvertices\mid
\edges(v)\cap\Attr_i(\finalvertices)\neq\emptyset\} \cup
\{v\in\maxvertices\mid \edges(v)\subseteq\Attr_i(\finalvertices)
\}\,.\]
It is well-known that this sequence converges after at most
$|\vertices|$ steps to the set $\Attr(\finalvertices)$ of all vertices
from which $\MinPl$ has a memoryless strategy to ensure reaching
$\finalvertices$. Hence, under our hypothesis,
$\vertices_{+\infty}=\vertices \setminus \Attr(\finalvertices)$. This
proves the first item of Theorem~\ref{thm:optimal-strategy}.

\section{\texorpdfstring{Finding vertices of value $-\infty$ in
    min-cost reachability games}{Finding vertices of negative infinity
    value in min-cost reachability games}\label{sec:-infty}}

We give here the proof of the following
 aiming at computing the
set $\vertices_{-\infty}$ of vertices with a value $-\infty$ in a
min-cost reachability game $\game = \gameEx[\MCR[T]]$, from which we
suppose that there is no more vertices with value $+\infty$.

\begin{proposition}\label{prop:from-wg-to-mp-and-vice-versa}
  For all MCR game $\game = \gameEx[\MCR]$, for all vertices $v$ of
  $\game$, $\Value_\game(v)=-\infty$ if and only if
  $\Value_{\game'}(v)<0$, where $\game'$ is the mean-payoff game
  $\gameEx[\MP]$. Conversely, given a mean-payoff game
  $\game=\gameEx[\MP]$, we can build, in polynomial time, an MCR game
  $\game'$ such that for all vertices $v$ of $\game$:
  $\Value_{\game}(v)<0$ if and only if $\Value_{\game'}(v)=-\infty$.
\end{proposition}
\begin{proof}
  Consider first a min-cost reachability game
  $\game = \gameEx[\MCR[T]]$ such that $\Value_\game(v)\neq +\infty$
  for all $v\in\vertices$, and $\game'=\gameEx[\MP]$ the same weighted
  graph equipped of a mean-payoff objective.

  If $\Value_{\game'}(v)<0$, the outcome starting in $v$ and following
  a profile $(\maxstrategy^*,\minstrategy^*)$ of optimal memoryless
  strategies necessarily starts with a finite prefix and then loops in
  a cycle with a total weight less than $0$. For every $M>0$, we
  construct a strategy $\minstrategy$ that ensures in $\game$ a cost
  less than or equal to $-M$: this will prove that
  $\Value_\game(v)=-\infty$. Since in every vertex of $\game'$,
  $\MinPl$ has a strategy in $\game$ to reach the final vertices
  (otherwise there would exist a vertex with value $+\infty$), there
  exists $w''$ such that $\MinPl$ can reach from any vertex of
  $\game'$ the target with a cost at most $w''$. The strategy
  $\minstrategy$ of $\MinPl$ is then to follow $\minstrategy^*$ until
  the accumulated cost is less than $-M-w''$ (which we can prove will
  happen), at which point it follows its strategy to reach the target.

  Reciprocally, if $\Value_\game(v)=-\infty$, consider
  $M=|\vertices| W$ and a strategy $\minstrategy^M$ of $\MinPl$
  ensuring a cost less than $-M$, i.e., such that
  $\pricestrategy(v,\minstrategy^M)<-M$. Consider the
  finitely-branching tree built from $\game$ by unfolding the game
  from vertex $v$ and resolving the choices of $\MinPl$ with strategy
  $\minstrategy^M$. Each branch of this tree corresponds to a possible
  strategy of $\MaxPl$. Since this strategy generates a finite cost,
  we are certain that every such branch leads to a vertex of
  $\finalvertices$. If we trim the tree at those vertices, we finally
  obtain a finite tree. Now, for a contradiction, suppose the optimal
  memoryless strategy $\maxstrategy^*$ of $\MaxPl$ ensures a
  non-negative mean-payoff, that is,
  $\Value_{\game'}(v,\maxstrategy^*)\geq 0$.  Consider the branch of
  the previous tree where $\MaxPl$ follows strategy
  $\maxstrategy^*$. Since this finite branch has cost less than
  $-M=-|\vertices| W$, we know for sure that there is two occurrences
  of the same vertex $v'$ with an in-between weight less than $0$:
  otherwise, by removing all nonnegative cycles, we obtain a play
  without repetition of vertices, henceforth of length bounded by
  $|\vertices|$, and therefore of cost at least $-M$.  Suppose that
  $v'\in\maxvertices$. Then, $\MinPl$ has a strategy $\minstrategy$ to
  ensure a negative mean-payoff $\Value_{\game'}(v,\minstrategy)<0$:
  indeed, it simply modifies its strategy so that it always follows
  his choices made in the negative cycle starting in $v'$, ensuring
  that, against the optimal strategy $\maxstrategy^*$ of $\MaxPl$, he
  gets a mean-payoff being the cost of the cycle. This is a
  contradiction since $\MaxPl$ is supposed to have a strategy ensuring
  a non-negative mean-payoff from $v$. Hence, $v'\in\minvertices$. But
  the same contradiction appears in that case since $\MinPl$ can force
  that it always stays in the negative cycle by modifying his
  strategy. Finally, we have proved that $\MaxPl$ cannot have a
  memoryless strategy ensuring a non-negative mean-payoff from $v$. By
  memoryless determinacy of the mean-payoff games, this ensures that
  $\MinPl$ has a memoryless strategy ensuring a negative mean-payoff
  from~$v$.

  Hence, we have shown that $\Value_\game(v)=-\infty$ if and only if
  $\Value_{\game'}(v)<0$, which concludes the first claim of
  Proposition~\ref{prop:from-wg-to-mp-and-vice-versa}.

  \medskip Conversely, we reduce mean-payoff games to min-cost
  reachability games as follows. Let $\game = \gameEx[\MP]$ be a
  mean-payoff game. Without loss of generality, we may suppose that
  the graph of the game is bipartite, in the sense that
  $\edges\subseteq \maxvertices\times \minvertices \cup
  \minvertices\times \maxvertices$.
  The problem we are interested in is to decide whether
  $\Value_\game(v)<0$ for a given vertex~$v$. We now construct a
  min-cost reachability game
  $\game' = \tuple{\vertices',\edges',
    \edgeweights',\MCR[\finalvertices']}$
  from $\game$. The only difference is the presence of a fresh target
  vertex $\target$ on top of vertices of $\vertices$:
  $\vertices'=\vertices\uplus\{\target\}$ with
  $\finalvertices'=\{\target\}$. Edges of $\weightedarena'$ are given
  by
  $\edges' = \edges \cup\{(v,\target)\mid v\in\minvertices\}\cup
  \{(\target,\target)\}$.
  Weights of edges are given by:
  $\edgeweights'(v,v')=\edgeweights(v,v')$ if $(v,v')\in\edges$, and
  $\edgeweights'(v,\target)=\edgeweights'(\target,\target)=0$.  We
  show that $\Value_{\game}(v)<0$ if and only if
  $\Value_{\game'}(v)=-\infty$.

  In $\game'$, all values are different from $+\infty$, since $\MinPl$
  plays at least every two steps, and has the capability to go to the
  target vertex with weight 0. Hence, letting
  $\game''=\tuple{\vertices',\edges', \edgeweights',\MP}$ the
  mean-payoff game on the weighted graph of $\game'$, by the previous
  direction, we have that for every vertex $v\in\vertices'$,
  $\Value_{\game'}(v)=-\infty$ if and only if $\Value_{\game''}(v)<0$.

  To conclude, we prove that for all vertices $v\in\vertices$,
  $\Value_{\game''}(v)<0$ if and only if $\Value_{\game}(v)<0$. If
  $\Value_{\game}(v)<0$, by mapping the memoryless optimal strategies
  of $\game$ into $\game''$, we directly obtain that
  $\Value_{\game''}(v)\leq \Value_{\game}(v)<0$. Reciprocally, if
  $\Value_{\game''}(v)<0$, we can project a profile of memoryless
  optimal strategies over vertices of $\game$: the play obtained from
  $v$ in $\game$ is then the projection of the play obtained from $v$
  in $\game''$, with the same cost. Hence,
  $\Value_{\game}(v)\leq \Value_{\game''}(v)<0$.
\end{proof}

\section{Computing finite values in min-cost reachability
  games\label{sec:finite-values}}

We now prove the correction of
Algorithm~\ref{algo:value-iteration-RT}, as stated below.

\begin{proposition}\label{the:rtpo-ca-marche-youpi}
  If an MCR game $\game=\gameEx[\MCR]$ is given as input (possibly
  with values $+\infty$ or $-\infty$),
  Algorithm~\ref{algo:value-iteration-RT} outputs $\Value_\game$,
  after at most $(2|\vertices|-1) W |\vertices|+2|\vertices|$
  iterations.
\end{proposition}

Observe first that for all $v\in \vertices$, for all $i\geq 0$, and for all
strategies $\maxstrategy$ and $\minstrategy$:
\[ \boundedWeight{i}(\outcomes(v,\maxstrategy,\minstrategy))\geq
\MCR(\outcomes(v,\maxstrategy,\minstrategy))\,.\]%
Indeed, if the target vertex $\target$ is reached within $i$ steps,
then both plays will have the same payoff. Otherwise,
$\boundedWeight{i}(\outcomes(v,\minstrategy,\maxstrategy))
=+\infty$. Thus, for all $i\geq 1$ and $v\in\vertices$:
\[\boundeduppervalue{i}(v)\geq \uppervalue(v)=\Value(v)\]
which can be rewritten as 
\[\boundeduppervalue{i}\vgeq \uppervalue = \Value\,.\]

Let us now consider the sequence $(\bupval{i})_{i\geq 0}$. We first
give an alternative definition of this sequence permitting to show its
convergence.
\begin{lemma}\label{lemma:min-max-charact-of-bupval}
  For all $i\geq 1$, for all $v\in \vertices$:
  \[
    \bupval{i}(v) =
    \begin{cases}
      \displaystyle{\max_{v'\in\edges(v)}}
      \big(\edgeweights(v,v')+\bupval{i-1}(v')\big)
      &\textrm{if } v\in \maxvertices\setminus\{\target\}\\
      \displaystyle{\min_{v'\in\edges(v)}}
      \big(\edgeweights(v,v')+\bupval{i-1}(v')\big)
      &\textrm{if } v\in \minvertices\setminus\{\target\}\\
      0 &\textrm{if } v=\target
    \end{cases}
  \]
\end{lemma}
\begin{proof}
  The lemma can be established by showing that $\bupval{i}(v)$ is the
  value in a game played on a finite tree of depth $i$ (i.e., by
  applying a backward induction). We adopt the following notation for
  labeled unordered trees. A leaf is denoted by $(v)$, where
  $v\in \vertices$ is the label of the leaf. A tree with root labeled
  by $v$ and subtrees $A_1,\ldots,A_n$ is denoted by
  $(v, \{A_1,\ldots, A_n\})$. Then, for each $v\in \vertices$ and
  $i\geq 0$, we define $A^i(v)$ as follows:
  \begin{align*}
    A^0(v) &= (v)\\
    \textrm{for all }i\geq 1: A^i(v) &= (v,\{A^{i-1}(v')\mid (v,v')\in
    E\})
  \end{align*}
  Now, let us further label those trees by a value in
  $\Z\cup\{+\infty\}$ thanks to the function $\lambda$. For all tree
  of the form $A^0(v)=(v)$, we let
  \begin{align*}
    \lambda(A^0(v)) &= \left\{
      \begin{array}{ll}
        0&\textrm{if } v=\target\\
        +\infty&\textrm{if } v\neg\target\\
      \end{array}
    \right.
  \end{align*}
  For all tree of the form $A^{i}(v)=(v,\{A^{i-1}(v_1),\ldots,
  A^{i-1}(v_m)\})$ (for some $i\geq 1$), we let
  \begin{equation}
    \lambda(A^{i}(v)) = 
    \begin{cases}
      \displaystyle{\max_{1\leq j \leq m}}
      \left(\edgeweights(v,v_j)+\lambda(A^{i-1}(v_j))\right)
      &\textrm{if }v\in \maxvertices\setminus\{\target\}\\
      \displaystyle{\min_{1\leq j \leq m}}
      \left(\edgeweights(v,v_j)+\lambda(A^{i-1}(v_j))\right)
      &\textrm{if }v\in \minvertices\setminus\{\target\}\\
      0 &\textrm{if } v=\target
    \end{cases}\label{eq:a1}
  \end{equation}

  Clearly, for all $v\in \vertices$, for all $i\geq 0$, the branches of
  $A^i(v)$ correspond to all the possible finite plays
  $\outcomes(v,\maxstrategy,\minstrategy)[i]$, i.e., there is a branch
  for each possible strategy profile $(\maxstrategy$,
  $\minstrategy)$. Thus, $\lambda(A^i(v))=\bupval{i}(v)$ for all
  $i\geq 0$, which permits us to conclude from \eqref{eq:a1}.
\end{proof}

We have just shown that for all $i\geq 1$, $\bupval{i}
=\operator(\bupval{i-1})$, so that $x_i=\bupval{i}$ for all $i\geq 0$
(with $x_i$ defined in page~\pageref{fi}). 

\begin{remark}\label{rem:Kleene-sequence}
  Notice that, at this point, it would not be too difficult to show
  that $\Value$ is a fixed point of operator $\operator$. However, it
  is more difficult to show that it is the greatest fixed point of
  $\operator$, and to deduce directly properties over optimal
  strategies in the min-cost reachability game. Instead, we use
  the sequence $(\bupval{i})_{i\geq 0}$ to obtain more interesting
  results on $\Value$. 
\end{remark}

We now study refined properties of the sequence
$(\bupval{i})_{i\geq 0}$, namely its stationarity and the speed of its
convergence.  We start by characterizing how $\bupval{i}$ evolves over
the first $|\vertices|+1$ steps. The next lemma states that, for each
node $v$, the sequence
$\bupval{0}(v),\bupval{1}(v),\ldots,\allowbreak\bupval{i}(v),\ldots,
\bupval{|\vertices|}(v)$ is of the form
\[
\underbrace{+\infty,+\infty,\ldots,+\infty}_{k\textrm{
    times}},a_{k}, a_{k+1},\ldots,a_{|\vertices|}
\]
where $k$ is the step at which $v$ has been added to the attractor,
and each value $a_i$ is finite and bounded:
\begin{lemma}\label{lem:link-attractor-values}
  Let $v\in \vertices$ be a vertex and let $0\leq k\leq |\vertices|$
  be such that
  $v\in\Attr_k(\{\target\})\setminus \Attr_{k-1}(\{\target\})$
  (assuming $\Attr_{-1}(\{\target\})=\emptyset$). Then, for all
  $0\leq j\leq |\vertices|$:
  \begin{inparaenum}[(i)]
  \item $j<k$ implies $\bupval{j}(v)=+\infty$ and
  \item $j\geq k$ implies $\bupval{j}(v)\leq j  W$.
  \end{inparaenum}
\end{lemma}
\begin{proof}
  We prove the property for all vertices $v$, by induction on
  $j$.\smallskip

  \noindent\textbf{Base case}: 
  $j=0$. We consider two cases. Either $v=\target$. In this case,
  $k=0$, and we must show that $\bupval{0}(v)\leq 0\times W=0$, which
  is true by definition of $\bupval{0}$. Or $v\neq \target$.  In this
  case, $k>0$, and we must show that $\bupval{0}(v)=+\infty$, which is
  true again by definition of $\bupval{0}$.  \smallskip

  \noindent\textbf{Inductive case}: 
  $j=\ell\geq 1$. Let us assume that the lemma holds for all $v$, for
  all values of $j$ up to $\ell-1$, and let us show that it holds for
  all $v$, and for $j=\ell$. Let us fix a vertex $v$, and its
  associated value $k$. We consider two cases.
  \begin{enumerate}
  \item First, assume $k>\ell$. In this case, we must show that
    $\bupval{\ell}(v)=+\infty$. We consider again two cases:
    \begin{enumerate}
    \item If $v\in \minvertices$, then none of its successors belong to
      $\Attr_{\ell-1}(\{\target\})$, otherwise, $v$ would be in
      $\Attr_{\ell}(\{\target\})$, by definition of the attractor,
      and we would have $k\leq \ell$. Hence, by induction hypothesis,
      $\bupval{\ell-1}(v')=+\infty$ for all $v'$ such that $(v,v')\in
      E$. Thus:
      \begin{align*}
      \bupval{\ell}(v) &= \min_{(v,v')\in E} \left(\edgeweights(v,v')
        +\bupval{\ell-1}(v')\right)
      &\textrm{(Lemma~\ref{lemma:min-max-charact-of-bupval})}\\
      &=+\infty
    \end{align*}
  \item If $v\in \maxvertices$, then at least one successor of $v$ does not
    belong to $\Attr_{\ell-1}(\{\target\})$, otherwise, $v$ would
    be in $\Attr_{\ell}(\{\target\})$, by definition of the
    attractor, and we would have $k\leq \ell$. Hence, by induction
    hypothesis, there exists $v'$ such that $(v,v')\in E$ and
    $\bupval{\ell-1}(v')=+\infty$. Thus:
    \begin{align*}
      \bupval{\ell}(v) &= \max_{(v,v')\in E} \left(\edgeweights(v,v')+
        \bupval{\ell-1}(v')\right)
      &\textrm{(Lemma~\ref{lemma:min-max-charact-of-bupval})}\\
      &=+\infty
    \end{align*}
    \end{enumerate}
  \item Second, assume $k\leq \ell$. In this case, we must show that
    $\bupval{\ell}(v)\leq \ell  W$. As in the previous item, we
    consider two cases:
    \begin{enumerate}
    \item In the case where $v\in \minvertices$, we let $\vbar$ be a vertex
      such that $\vbar\in \Attr_{k-1}(\{\target\})$ and
      $(v,\vbar)\in E$. Such a vertex exists by definition of the
      attractor. By induction hypothesis, $\bupval{\ell-1}(\vbar)\leq
      \ell  W$. Then:
      \begin{align*}
        \bupval{\ell}(v) &= \min_{(v,v')\in E} \left(\edgeweights(v,v')
          +\bupval{\ell-1}(v')\right)
        &\textrm{(Lemma~\ref{lemma:min-max-charact-of-bupval})}\\
        &\leq \edgeweights(v,\vbar)+\bupval{\ell-1}(\vbar)
        &((v,\vbar)\in E)\\ 
        &\leq\edgeweights(v,\vbar)+(\ell-1)  W
        &\textrm{(Ind. Hyp.)}\\ 
        &\leq W+(\ell-1)  W\\
        &=\ell  W
      \end{align*}
    \item In the case where $v\in \maxvertices$, we know that all successors
      $v'$ of $v$ belong to $\Attr_{k-1}(\{\target\})$ by
      definition of the attractor. By induction hypothesis, for all
      successors $v'$ of $v$: $\bupval{\ell-1}(v')\leq \ell 
      W$.Hence:
      \begin{align*}
        \bupval{\ell}(v) &= \max_{(v,v')\in E}
        \left(\edgeweights(v,v')+ \bupval{\ell-1}(v')\right)
        &\textrm{(Lemma~\ref{lemma:min-max-charact-of-bupval})}\\
        &\leq \max_{(v,v')\in E} \left(W+(\ell-1)  W\right)
        &\textrm{(Ind. Hyp.)}\\ 
        &=\ell  W\,.
      \end{align*}
    \end{enumerate}
  \end{enumerate}
\end{proof}

In particular, this allows us to conclude that, after $|\vertices|$
steps, all values are bounded by $|\vertices|  W$:
\begin{corollary}\label{cor:after-n-steps-no-infty}
   For all $v\in \vertices$: $\bupval{|\vertices|}(v)\leq |\vertices|  W\,$.
\end{corollary}

The next step is to show that the sequence stabilises after a bounded
number of steps:
\begin{lemma}\label{lem:value-stab}
  The sequence $\bupval{0},\ldots,\bupval{i},\ldots$ stabilises after
  at most $(2|\vertices|-1)  W   |\vertices|+|\vertices|$
  steps.
\end{lemma}
\begin{proof}
  We first show that if $\MinPl$ can secure, from some vertex $v$, a
  payoff less than $-(|\vertices|-1) W$, i.e.,
  $\Value(v)<-(|\vertices|-1) W$, then it can secure an arbitrarily
  small payoff from that vertex, i.e., $\Value(v)=-\infty$, which
  contradicts our hypothesis that the value is finite. Hence, let us
  suppose that there exists a strategy $\minstrategy$ for $\MinPl$
  such that $\pricestrategy(v,\minstrategy)<-(|\vertices|-1) W$. Let
  $\game'$ be the mean-payoff game studied in
  Proposition~\ref{prop:from-wg-to-mp-and-vice-versa}. We will show
  that $\Value_{\game'}(v)<0$, which permits to conclude that
  $\Value_\game(v)=-\infty$. Let $\maxstrategy$ be a memoryless
  strategy of $\MaxPl$. By hypothesis, we know that
  $\MCR(\outcomes(v,\maxstrategy,\minstrategy)) < -(|\vertices|-1) W$.
  This ensures the existence of a cycle with negative cost in the play
  $\outcomes(v,\maxstrategy,\minstrategy)$: otherwise, we could
  iteratively remove every possible nonnegative cycle of the finite play
  before reaching $\target$ (hence reducing the cost of the play) and
  obtain a play without cycles before reaching $\target$ with a cost
  less than $-(|\vertices|-1) W$, which is impossible (since it should
  be of length at most $|\vertices|-1$ to cross at most one occurrence
  of each vertex). Consider the first negative cycle in the
  play. After the first occurrence of the cycle, we let $\MinPl$
  choose its actions like in the cycle. By this way, we can construct
  another strategy $\minstrategy'$ for $\MinPl$, verifying that for
  every memoryless strategy $\maxstrategy$ of $\MaxPl$, we have
  $\MP(\outcomes(v,\maxstrategy,\minstrategy'))$ being the weight of
  the negative cycle in which the play finishes. Since for mean-payoff
  games, memoryless strategies are sufficient for $\MaxPl$, we deduce
  that $\Value_{\game'}(v)<0$.

  This reasoning permits to prove that at every step $i$,
  $\bupval{i}(v)\geq \Value(v)\geq -(|\vertices|-1)  W+1$ for all
  vertices~$v$. Recall from Corollary~\ref{cor:after-n-steps-no-infty}
  that, after $|\vertices|$ steps in the sequence, all vertices are
  assigned a value smaller that $|\vertices|  W$. Moreover, we
  know that the sequence is non-increasing. In summary, for all
  $k\geq 0$ and for all vertices $v$:
  \[-(|\vertices|-1)  W+1 \leq \bupval{|\vertices|+k}(v)\leq
  |\vertices|  W\]%
  Hence, in the worst case a strictly decreasing sequence will need
  $(2|\vertices|-1)  W   |\vertices|$ steps to reach the
  lowest possible value where all vertices are assigned
  $-(|\vertices|-1)  W+1$ from the highest possible value where
  all vertices are assigned $|\vertices|  W$. Thus, taking into
  account the $|\vertices|$ steps to reach a finite value on all
  vertices, the sequence stabilises in at most $(2|\vertices|-1) 
  W   |\vertices|+|\vertices|$ steps.
\end{proof}

Let us thus denote by $\bupval{}$ the value obtained when the sequence
$(\bupval{i})_{i\geq 0}$ stabilizes. We know from a previous
discussion that $\bupval{}$ is the greatest fixed point of operator
$\operator$. We are now ready to prove that this value is the actual
value of the game:
\begin{lemma}\label{lem:convergence}
  For all min-cost reachability game: $\bupval{}=\Value\,.$
\end{lemma}
\begin{proof}
  We already know that $\bupval{}\vgeq \Value$. Let us show that
  $\bupval{}\vleq \Value$. Let $v\in\vertices$ be a vertex. Since
  $\Value(v)$ is finite integer, there exists a strategy
  $\minstrategy$ for $\MinPl$ that realises this value, i.e.,
  \[\Value(v)=\sup_{\maxstrategy}
  \MCR(\outcomes(v,\maxstrategy,\minstrategy))\,.\] Notice that
  this holds because the values are integers, enducing that the
  infimum in the definition of $\uppervalue(v)=\Value(v)$ is indeed
  reached.

  Let us build a tree $A_{\minstrategy}$ unfolding all possible plays
  from $v$ against $\minstrategy$. $A_{\minstrategy}$ has a root
  labeled by $v$. If a tree node is labeled by a vertex $v$ of
  $\MinPl$, this tree node has a unique child labeled by
  $\minstrategy(v)$. If a tree node is labeled by a vertex $v$ of
  $\MaxPl$, this tree node has one child per successor $v'$ of $v$ in
  the graph, labeled by $v'$. We proceed this way until we encounter a
  node labeled by a vertex from $\target$ in which case this node is a
  leaf. $A_{\minstrategy}$ is necessarily finite. Otherwise, by
  K\"onig's Lemma, it has one infinite branch that never reaches
  $\target$. From that infinite branch, one can extract a strategy
  $\maxstrategy$ for $\MaxPl$ such that
  $\MCR(\outcomes(v,\maxstrategy,\minstrategy))=+\infty$, hence
  $\Value(v)=+\infty$, which contradicts the hypothesis. Assume the
  tree has depth $m$. Then, $A_{\minstrategy}$ is a subtree of the
  tree $A$ obtained by unfolding all possible plays up to length $m$
  (as in the proof of Lemma~\ref{lemma:min-max-charact-of-bupval}). In
  this case, it is easy to check that the value labeling the root of
  $A_{\minstrategy}$ after applying backward induction is larger than
  or equal to the value labeling the root of $A$ after applying
  backward induction. The latter is $\Value(v)$ while the former is
  $\bupval{m}(v)$, by Lemma~\ref{lemma:min-max-charact-of-bupval}, so
  that $\Value(v)\geq\bupval{m}(v)$. Since the sequence is
  non-increasing, we finally obtain $\Value(v)\geq \bupval{}(v)$. 
\end{proof}

As a corollary of this lemma, we obtain:
\begin{corollary}\label{lem:min-max-charact-of-value}
  $\Value$ is the greatest fixed point of $\operator$.
\end{corollary}

This permits to obtain a value iteration algorithm, described in
Algorithm~\ref{algo:value-iteration-RT}, that computes optimal
values. Notice that we do not suppose that every vertex has a finite
value, which is justified in the proof of
Proposition~\ref{the:rtpo-ca-marche-youpi} below proving the
correctness of the algorithm, as well as its complexity. A crucial
argument is given in the following lemma, following from the fact that
$\Value$ is the greatest fixed point of $\operator$:
\begin{lemma}\label{lem:acceleration}
  If the Kleene sequence $(\operator^i(x_0))_{i\geq 0}$ is initiated
  with a vector of values $x_0$ that is greater or equal to the
  optimal value vector $\Value$, then the sequence converges at least
  as fast as before towards the optimal value vector.
\end{lemma}

\begin{proof}[Proof of Proposition~\ref{the:rtpo-ca-marche-youpi}]
  Let us first suppose that values of every vertices are finite. Then,
  we can easily prove by induction that at the beginning of the $j$th
  step of the loop, $\mathsf X$ is equal to the vector $\bupval j$,
  and that the condition of line~\ref{line-infty} has never been
  fulfilled. Hence, by Lemma~\ref{lem:value-stab}, after at most
  $(2|\vertices|-1)  W  |\vertices|+|\vertices|$ iterations,
  all values are found correctly in that case.

  Suppose now that there exist vertices with value $+\infty$. Those
  vertices will remain at their initial value $+\infty$ during the
  whole computation, and hence do not interfere with the rest of the
  computation. 

  Finally, consider that the game contains vertices with value
  $-\infty$. We know that optimal values of vertices of values
  different from $-\infty$ are at least $-(|\vertices|-1)  W +1$ so that,
  if the value of a vertex reaches an integer below $-(|\vertices|-1)
    W$, we are sure that its value is indeed $-\infty$, which
  proves correct the line~\ref{line-infty} of the algorithm. This
  update may cost at most one step per vertices, which in total adds
  at most $|\vertices|$ iterations. Moreover, by
  Lemma~\ref{lem:acceleration}, dropping the value to $-\infty$ does
  not harm the correction for the other vertices (it may only speed
  the convergence of their values). 
\end{proof}

\section{Computing optimal strategies in min-cost reachability
  games\label{sec:comp-optim-strat}\label{sec:optimal-strategies}}

This section is devoted to the formal construction of optimal
strategies in MCR games, as explained briefly at the end of
Section~\ref{sec:reachability-objectives}.

\subsection{\texorpdfstring{Strategies of $\MinPl$}{Strategies of Min}}
 
We have already seen an example in Fig.~\ref{fig:Weighted-game} of a
game where $\MinPl$ may need memory in an optimal strategy, i.e., where
$\MinPl$ has an optimal strategy, but no memoryless optimal
strategies. Reciprocally, as a consequence of the previous work, we
first show that, for vertices with finite value, $\MinPl$ has always a
finite-memory optimal strategy.

\begin{proposition}\label{prop-rtp-finite-memory}
  In all min-cost reachability game with only vertices of finite
  values, $\MinPl$ has a finite-memory optimal strategy.
\end{proposition}
\begin{proof}
  We explain how to reconstruct from the fixpoint computation an
  optimal strategy $\minstrategy^*$ for $\MinPl$. Let $k$ be the
  integer such that $\bupval{k+1}=\bupval{k}$. For every step
  $\ell\leq k$ of the fixpoint computation, we define a strategy
  $\minstrategy^\ell$. For every finite play $\pi$ ending in a vertex
  $v$ of $\MinPl$, if the length of $\pi$ is $i<\ell$, we let
  \[\minstrategy^\ell(\pi) = \argmin_{v'\in\edges(v)}
  \big(\edgeweights(v,v')+\bupval{\ell-i-1}(v')\big)\,, \]
  and otherwise, we let 
  \[\minstrategy^\ell(\pi) = \argmin_{v'\in\edges(v)}
  \big(\edgeweights(v,v')+\bupval{0}(v')\big)\,. \] Notice that the
  $\argmin$ operator may select any possible vertex as long as it
  selects one with minimum value. We also let
  $\minstrategy^{\ell}=\minstrategy^k$ for every $\ell>k$, as well as
  $\minstrategy^*=\minstrategy^k$.  Notice that $\minstrategy^*$ is a
  finite-memory strategy, since it only requires to know the last
  vertex and the length of the prefix up to $k$.

  We now prove that
  $\Val(v,\minstrategy^*)=\sup_{\maxstrategy}
  \MCR(\outcomes(v,\maxstrategy,\minstrategy^*)) \leq \bupval{}(v)$
  for all vertices~$v$, which proves that $\minstrategy^*$ is an
  optimal strategy since $\bupval{}(v)=\Value(v)$ by
  Lemma~\ref{lem:convergence}. To do so, we first show by induction on
  $\ell$ that
  \begin{equation}
    \boundedWeight{\ell}(\outcomes(v,\maxstrategy,\minstrategy^\ell)) \leq
    \bupval{\ell}(v)\label{eq:inequality-ind}
  \end{equation}
  holds for every strategy $\maxstrategy$ of $\MaxPl$ and
  $v\in\vertices$. This permits to conclude since, from
  $\minstrategy^*=\minstrategy^\ell$ for every $\ell\geq k$, we can
  deduce:
  \begin{align*}
    \lim_{\ell\to\infty} \sup_{\maxstrategy}
    \boundedWeight\ell(\outcomes(v,\maxstrategy,\minstrategy^\ell))
    &=\lim_{\ell\to\infty} \sup_{\maxstrategy}
    \boundedWeight\ell(\outcomes(v,\maxstrategy,\minstrategy^*))\\
    &=\sup_{\maxstrategy}
    \MCR(\outcomes(v,\maxstrategy,\minstrategy^*))
  \end{align*}
  and $\bupval{i}$ is a stationary sequence converging towards
  $\bupval{}$. The proof by induction goes as follows. In case
  $\ell=0$, either $v=\target$ and both terms of
  \eqref{eq:inequality-ind} are equal to $0$, or $v\neq\target$ and
  both terms of \eqref{eq:inequality-ind} are equal to
  $+\infty$. Supposing now that the property holds for an index
  $\ell$, let us prove it for $\ell+1$. For that, we consider a
  strategy $\maxstrategy$ of $\MaxPl$. In case $v=\target$, we have
  \[\boundedWeight{\ell}(\outcomes(v,\maxstrategy,\minstrategy^\ell))=0=
  \bupval{\ell}(v)\,.\]
  We now consider the case $v\neq\target$. Let $v'$ be the second
  vertex in $\outcomes(v,\maxstrategy,\minstrategy^{\ell+1})$. From
  the definition of $\minstrategy^{\ell+1}$,
  $\outcomes(v,\maxstrategy,\minstrategy^{\ell+1})[\ell+1]$ is the
  concatenation of $v$ and
  $\outcomes(v',\maxstrategy,\minstrategy^\ell)[\ell]$. Hence,
  \[\boundedWeight{\ell+1}(\outcomes(v,\maxstrategy,\minstrategy^{\ell+1}))
  = \edgeweights(v,v') +
  \boundedWeight\ell(\outcomes(v',\maxstrategy,\minstrategy^\ell))\,.\]
  By induction hypothesis, we obtain that
  \begin{equation}
    \boundedWeight{\ell+1}(\outcomes(v,\maxstrategy,\minstrategy^{\ell+1}))
    \leq \edgeweights(v,v') + \bupval{\ell}(v')\,.\label{eq:IH}
  \end{equation}
  Now, consider the two following cases.
  \begin{itemize}
  \item If $v\in\minvertices\setminus \{\target\}$, we have
    $v'=\minstrategy^{\ell+1}(v)$, so that, in case $\ell+1\leq k$:
    \[\edgeweights(v,v') +\bupval{\ell}(v')\leq \min_{v''\in
      \vertices\mid (v,v'')\in\edges}
    \big(\edgeweights(v,v'')+\bupval{\ell}(v'')\big)\,.\]
    Using \eqref{eq:IH} and
    Lemma~\ref{lemma:min-max-charact-of-bupval}, we obtain 
    \[\boundedWeight{\ell+1}(\outcomes(v,\maxstrategy,\minstrategy^{\ell+1}))
    \leq \bupval{\ell+1}(v)\,.\] In case $\ell+1>k$, we indeed have
    $\bupval{\ell}=\bupval{\ell+1}= \bupval{k}$, so that we conclude
    similarly.
  \item If $v\in\maxvertices\setminus \{\target\}$, we have
    $v'=\maxstrategy(v)$ and
    \[\edgeweights(v,v') +\bupval{\ell}(v')\leq \max_{v''\in
      \vertices\mid (v,v'')\in\edges}
    \big(\edgeweights(v,v'')+\bupval{\ell}(v'')\big)\,.\] Once again
    using \eqref{eq:IH} and
    Lemma~\ref{lemma:min-max-charact-of-bupval}, we obtain
    \[\boundedWeight{\ell+1}(\outcomes(v,\maxstrategy,\minstrategy^{\ell+1}))
    \leq \bupval{\ell+1}(v)\,.\]
  \end{itemize}
  This concludes the induction proof.
\end{proof}

Notice that the proof of Proposition~\ref{prop-rtp-finite-memory}
together with the statement of Lemma~\ref{lem:value-stab} imply that a
memory of size pseudo-polynomial for the strategy of $\MinPl$ is
sufficient. Before stating the result for $\MaxPl$, we informally
refine this result in order to find a strategy of $\MinPl$ having more
\emph{structural properties}. It will be composed of two memoryless
strategies $\minstrategy^1$ and $\minstrategy^2$: the game will start
with $\MinPl$ following $\minstrategy^1$, and at some point,
determined by the weight of the current finite play, $\MinPl$ will switch
to strategy $\minstrategy^2$ which is an attractor strategy, i.e., a
strategy that reaches the target in less than $|\vertices|$ steps,
regardless of the weights along this path.  Intuitively the strategy
$\minstrategy^1$ ensures either to reach the target with optimal
value, or to go in cycles of negative weights.  The only chance for
$\MinPl$ of having a greater value than the optimal is to go
infinitely through these cycles without reaching the target.  But if
it does so, the total-payoff will decrease and at some point the value
will be so low, that the cost of calling the attractor strategy will
leave the total-payoff smaller than the optimal value. Let us
formalise this construction.

For the sake of exposure, we present the construction when all values
are finite, but such construction can be applied with few changes when
some vertices have value $-\infty$ or $+\infty$.  We start by defining
a memoryless strategy $\minstrategy^1$ that has some good properties
(stated in Proposition~\ref{prop:almostPerfect}).  Let $X^i$ denote
the value of variable $\mathsf{X}$ after $i$ iteration of the loop of
Algorithm~\ref{algo:value-iteration-RT}, and let $X^0(v)=+\infty$ for
all $v\in \vertices$.  We have seen that the sequence
$X^0\vgeq X^1 \vgeq X^2\vgeq\cdots$ is stationary at some point, equal
to $\Value$.  For all vertices
$v\in \maxvertices\setminus\{\target\}$, let $i_v>0$ be the first
index such that $X^i(v)=\Value(v)$.  Fix a vertex $v'\neq \target$
such that $X^i(v) = \edgeweights(v,v') + X^{i-1}(v')$ (that exists by
definition) and define $\minstrategy^1(v)=v'$. The following lemma
states that the vertex $v'$ already reached its final value at step
$i-1$.

\begin{lemma}
  For all vertices $v\in \minvertices\setminus\{\target\}$,
  $ X^{i-1}(\minstrategy^1(v))= \Value(\minstrategy^1(v))$.
\end{lemma}
\begin{proof}
  Let $v'=\minstrategy^1(v)$. By contradiction assume that
  $X^{i-1}(v')> \Value(v')$, Note that there exists $j>i$ such that
  $X^{j-1}(v')= \Value(v')$. By definition,
  \begin{multline*}
    \Value(v) \leq X^j(v) \leq  \edgeweights(v,v') + X^{j-1}(v') =
    \edgeweights(v,v') + \Value(v') \\
    < \edgeweights(v,v') + X^{i-1}(v') = X^i(v) = \Value(v),
  \end{multline*}
  which raises a contradiction.
\end{proof}

We can state the properties of $\minstrategy^1$: intuitively one
can see it as an \emph{almost perfect strategy}, in the sense that it
is memoryless, if it reaches the target, then the value obtained is
optimal, and if it does not reach the target then the total-payoff of
the finite play will decrease as the game goes on. The only problem is one
cannot ensure that we reach the target.

\begin{proposition}\label{prop:almostPerfect}
  For all vertices $v$, and for all plays $\pi=v_1 v_2 \cdots$ starting
  in $v$ and conforming to $\minstrategy^1$,%
  \begin{enumerate}
  \item[(1)] if there exists $i<j$ such that $v_i=v_j$, then
    $\TP(v_i \cdots v_j)<0$,
  \item[(2)] if $\pi$ reaches $\target$ then $\TP(\pi)\leq \Value(v)$.
  \end{enumerate}
\end{proposition}
\begin{proof}
  Let us prove $(1)$, take a cycle $v_i\cdots v_j$ with $v_j=v_i$.
  Notice that at least one vertex of this cycle belongs to $\MinPl$,
  since, otherwise, $\MaxPl$ would have a strategy to obtain a value
  $+\infty$ for vertex $v_1$, which contradicts the hypothesis on the
  game.  Hence, for the sake of the explanation, we suppose that
  $v_i\in\minvertices$.  Let us also suppose that $i_{v_0}$ is maximal
  among $\{i_{v_\ell}\mid i\leq \ell <j, v_{i'}\in\maxvertices\}$.  The
  following extends straightforwardly to the case where this maximal
  vertex of $\MinPl$ is not $v_i$.  We prove by induction over
  $i< \ell \leq j$ that
  \[X^{i_{v_i}}(v_i)\geq \TP(v_i\cdots v_{\ell}) +
  X^{i_{v_i}-1}(v_\ell)\,.\] The base case comes from the fact that
  since $v_i\in \minvertices$ we have $\minstrategy^1(v_i)=v_{i+1}$
  thus $X^{i_{v_0}}=
  \edgeweights(v_i,v_{i+1})+X^{i_{v_0}-1}(v_{i+1})$. For the inductive
  case, let us consider $i< \ell < j$ such that $X^{i_{v_i}}(v_i)\geq
  \TP(v_i\cdots v_{\ell}) + X^{i_{v_i}-1}(v_\ell)$ and let us prove it
  for $\ell+1$.
 
  If $v_\ell\in\minvertices$, by definition of $X^{i_{v_i}}$, we have
  \[X^{i_{v_0}}(v_\ell)=\max_{(v_\ell,v')\in
    \edges}\edgeweights(v_\ell,v')+ X^{i_{v_0}-1}(v')\geq
  \edgeweights(v_\ell,v_{\ell+1})+ X^{i_{v_0}-1}(v_{\ell+1})\,.\] %
  In case $v_j\in\minvertices$, by maximality of $i_{v_0}$, we have
  \begin{align*}
    X^{i_{v_0}}(v_\ell)&=X^{i_{v_\ell}}(v_\ell)=
    \edgeweights(v_\ell,v_{\ell+1}) +  X^{i_{v_\ell}-1}(v_{\ell+1})\\
    &\geq \edgeweights(v_\ell,v_{\ell+1}) + X^{i_{v_0}-1}(v_{\ell+1})
  \end{align*}
  using that the sequence $X^0,X^1,X^2,\ldots$ is non-increasing. 

  Hence, in all cases, we have
  \[X^{i_{v_0}}(v_\ell)\geq
  \edgeweights(v_\ell,v_{\ell+1})+X^{i_{v_0}-1}(v_{\ell+1})\] %
  Using again that $X^0,X^1,X^2,\ldots$ is non-decreasing, we obtain
  \[X^{i_{v_0}-1}(v_\ell)\geq X^{i_{v_0}}(v_\ell)\geq
  \edgeweights(v_\ell,v_{\ell+1}) + X^{i_{v_0}-1}(v_{\ell+1})\,.\] %
  Injecting this into the induction hypothesis, we have
  \begin{align*}
    X^{i_{v_0}}(v_i) &\geq \TP(v_i\cdots v_\ell)
    +\edgeweights(v_\ell,v_{\ell+1}) +X^{i_{v_0}-1}(v_{\ell+1})\\
    & = \TP(v_i\cdots v_{\ell+1})+X^{i_{v_0}-1}(v_{\ell+1})
  \end{align*}
  which concludes the proof by induction. In particular, for
  $\ell=j$, as $v_i=v_j$ we obtain that
  \[X^{i_{v_0}}(v_i)\geq X^{i_{v_0}-1}(v_i)+\TP(v_i\cdots v_j)\,.\]
  and as, by definition of $i_{v_0}$, we have $X^{i_{v_0}}(v_i) <
  X^{i_{v_0}-1}(v_i)$, we necessarily have $\TP(v_1\cdots v_j)<0$.

  To prove $(2)$ we decompose $\pi$ as $\pi = v_1 \cdots v_k
  \target^\omega$ with for all $i$, $v_i\neq \target$.  We prove by
  decreasing induction on $i$ that $\TP(v_i \cdots v_k
  \target^\omega)\leq \Value(v)$.  If $i=k+1$,
  $\TP(\target^\omega)=0=\Value(\target)$.  If $i\leq k$, by induction
  we have $\TP(v_{i+1} \cdots v_k \target^\omega)=\Value(v_{i+1})$
  thus $\TP(v_i \cdots v_k \target^\omega)= \edgeweights(v_i,v_{i+1})+
  \Value(v_{i+1})$.  If $v_i\in \minvertices$ then $v_{i+1} =
  \minstrategy^\star(v)$ and $\Value(v_i)= \edgeweights(v_i,v_{i+1})+
  \Value(v_{i+1})= \TP(v_i \cdots v_k \target^\omega)$.  If $v_i\in
  \maxvertices$, then $\TP(v_i \cdots v_k \target^\omega)=
  \edgeweights(v_i,v_{i+1})+ \Value(v_{i+1}) \leq \max_{(v_i,v')\in E}
  \edgeweights(v_i,v_{i+1})+ \Value(v_{i+1})= \Value(v_i)$.
\end{proof}

Next, let $\minstrategy^2$ be the memoryless strategy induced by
the computation of the attractor: notice that it is possible to
construct it directly from the value iteration computation by mapping
a vertex $v$ to one vertex from which $v$ is first discovered (i.e.,
its value is first set to a real value different from $+\infty$). This
strategy ensures to reach the target after at most $|\vertices|$ steps, thus
for all $v$, $\Value(v,\minstrategy^2)\leq W (|\vertices|-1)$.

Before defining the strategy $\widetilde\minstrategy$, we introduce the
notion of \emph{switchable finite play} as follows. A finite play
$v_1\cdots v_k$ is switchable if $v_k\in \maxvertices$ and
$\TP(v_1\cdots v_k) \leq \Value(v_1)
-\Value(v_k,\minstrategy^2)$.
Intuitively a strategy is switchable, if by switching to the attractor
strategy, we ensures to get an optimal value.

We define the strategy $\minstrategy'$ as follows: for all
finite play $v_1\cdots v_k$ with $v_k\in \maxvertices$:
\[ \minstrategy'(v_1\cdots v_k) =
\begin{cases}
  \minstrategy^2(v_k) &\text{if }v_1 \cdots v_i \text{ is switchable
    for some } i \\ 
  \minstrategy^1(v_k) & \text{otherwise}
\end{cases}\]

One can easily show that for all play $v_1 v_2\cdots $ that conforms
to $\minstrategy'$ either $v_1 v_2 \cdots$ conforms to
$\minstrategy^1$ if no finite play is switchable, or there exists $k$ such
that $v_1\cdots v_k$ conforms to $\minstrategy^1$, $v_1\cdots v_k$ is
a switchable finite play, and $v_k v_{k+1} \cdots$ conforms to
$\minstrategy^2$.  The following proposition states that the strategy
$\minstrategy'$ is optimal.
\begin{proposition}
  For all $v$, $\Value(v,\minstrategy')= \Value(v)$.
\end{proposition}
\begin{proof}
  Let $v\in\vertices$ and
  $\pi=v_1 v_2 \cdots \in \outcomes(v,\minstrategy')$, let us
  show that $\MCR(\pi)\leq\Value(v)$.  Assume first that there exists
  $k$ such that $v_1\cdots v_k$ conforms to $\minstrategy^1$,
  $v_1\cdots v_k$ is a switchable finite play, and $v_k v_{k+1} \cdots$
  conforms to $\minstrategy^2$. Thus $\pi$ reaches the target as
  $v_k v_{k+1} \cdots$ does, and
  $\MCR(\pi) = \TP(v_1\cdots v_k)+\MCR(v_k v_{k+1} \cdots)$. As
  $v_1 \cdots v_k$ is a switchable finite play, we have
  $\TP(v_1\cdots v_k) \leq \Value(v_1)
  -\Value(v_k,\minstrategy^2)$,
  thus
  $\MCR(\pi) \leq \Value(v_1) -\Value(v_k,\minstrategy^2) +
  \MCR(v_k v_{k+1} \cdots) \leq \Value(v_1)$.

  Assume now that $v_1 v_2 \cdots$ does not contains a switchable
  prefix and thus $\pi$ conforms to $\minstrategy^1$. If $\pi$
  reaches $\target$ then by Proposition~\ref{prop:almostPerfect},
  $\MCR(\pi)\leq\Value(v)$.  To conclude, let us prove that $\pi$
  reaches $\target$, and by contradiction assume that this is not the
  case.

  First we prove by induction on $k$, that $(\star)$ for all play
  $v'_1 \cdots v'_k$ that conforms to $\minstrategy^1$ and does
  not reach $\target$,
  $\TP( v_1 \cdots v_k) \leq (|\vertices|-1)W - \left \lfloor
    \frac{k}{|\vertices|}\right \rfloor$.
 
  If $k< |\vertices|$ the results is straightforward. For the
  inductive case, let $k\geq |\vertices|$ and $v'_1 \cdots v'_k$ that
  conforms to $\minstrategy^1$ and does not reach $\target$.  Then as
  $k \geq |\vertices|$, there exists $i<j$ such that $v_i = v_j$.
  Thus as $v'_1 \cdots v'_i v'_{j+1} \cdots v'_k$ is a play of size
  $k-|\vertices|$ that conforms to $\minstrategy^1$ and does not reach
  $\target$, we know by induction hypothesis that
  $\TP(v'_1 \cdots v'_i v'_{j+1} \cdots v'_k) \leq (|\vertices|-1)W -
  \left \lfloor \frac{k-|\vertices|}{|\vertices|} \right \rfloor =
  (|\vertices|-1)W - \left \lfloor \frac{k}{|\vertices|} \right
  \rfloor + 1 $.
  Furthermore by Proposition~\ref{prop:almostPerfect}, we have that
  $\TP(v'_i\cdots v'_j)\leq -1$.  Thus
  \begin{multline*}
    \TP(v'_1\cdots v'_k)
    = \TP (v'_1 \cdots v'_i) + \TP(v'_i \cdots v'_j)+\TP(v'_j \cdots v'_k) \\
    =\TP(v'_1 \cdots v'_i v'_{j+1} \cdots v'_k) +\TP(v'_1\cdots v'_j)
    \leq (|\vertices|-1)W - \left \lfloor \frac{k}{|\vertices|}
    \right\rfloor,
  \end{multline*}
  which concludes the proof of $(\star)$.
 
  Now we can get back to raising a contradiction, and for that we show
  that there exists a switchable finite play in $v_1 v_2 \cdots$.  Take
  $k$ be the least index greater that $3|\vertices|(|\vertices|-1)W$
  such that $v_k\in \maxvertices$ (we know that there exists one,
  otherwise the vertices $v_j$ with
  $j\geq 3|\vertices|(|\vertices|-1)W$ would not be in the attractor
  of $\target$).  We have
  \[\TP(v_1 \cdots v_k) \leq (|\vertices|-1)W - \left \lfloor
    \frac{k}{|\vertices|}\right \rfloor \leq -2(|\vertices|-1)W \leq
  \Value(v_1) + \Value(v_k,\minstrategy^2).\]
 
  Thus $v_1\cdots v_k$ is a switchable prefix, which raises a
  contradiction.  
\end{proof}

Notice that $\minstrategy'$ may be more easily implementable than a
more general finite-memory strategy, in particular, we may encode the
current total-payoff in binary, hence saving some space. We give in
Algorithm~\ref{algo:value-iteration-RT-strategy} a way to compute
strategies $\minstrategy^1$ and $\minstrategy^2$.

\subsection{\texorpdfstring{Strategies of $\MaxPl$}{Strategies of Max}}

While we have already shown that optimal strategies for $\MinPl$ might
require memory, let us show that $\MaxPl$ always has a
\emph{memoryless optimal strategy}. This asymmetry stems directly from
the asymmetric definition of the game -- while $\MinPl$ has the double
objective of reaching $\target$ and minimising its cost, $\MaxPl$ aims
at avoiding $\target$, and if not possible, maximising the cost.

\begin{proposition}\label{prop:optimal-player2}
  In all min-cost reachability game, $\MaxPl$ has a memoryless
  optimal strategy.
\end{proposition}
\begin{proof}
  For vertices with value $+\infty$, we already know a memoryless
  optimal strategy for $\MaxPl$, namely any strategy that remains
  outside the attractor of the target vertices. For vertices with
  value $-\infty$, all strategies are equally bad for $\MaxPl$.

  We now explain how to define a memoryless optimal strategy
  $\maxstrategy^*$ for $\MaxPl$ in case of a graph containing only
  finite values. For every finite play $\pi$ ending in a vertex
  $v\in\maxvertices$ of $\MaxPl$, we let
  \[\maxstrategy^*(\pi)=\argmax_{v'\in \edges(v)}
  \big(\edgeweights(v,v')+ \Value(v')\big)\,.\] This is clearly a
  memoryless strategy. Let us prove that it is optimal for $\MaxPl$,
  that is, for every vertex $v\in \vertices$, and every strategy
  $\minstrategy$ of $\MinPl$
  \[\MCR(\outcomes(v,\maxstrategy^*,\minstrategy)) \geq
  \Value(v)\,.\]
  In case $\outcomes(v,\maxstrategy^*,\minstrategy)$ does not reach
  the target set of vertices, the inequality holds
  trivially. Otherwise, we let
  $\outcomes(v,\maxstrategy^*,\minstrategy)= v_0v_1\cdots
  v_\ell\cdots$
  with $\ell$ the least position such that $v_\ell=\target$. If
  $\ell=0$, i.e., $v=v_0=\target$, we have
  \[\MCR(\outcomes(v,\maxstrategy^*,\minstrategy)) =0=
  \Value(v)\,.\]
  Otherwise, let us prove by induction on $0\leq i\leq \ell$ that 
  \[\MCR(v_{\ell-i}\cdots v_\ell) \geq \Value(v_{\ell-i})\,.\]
  This will permit to conclude since
  \[\MCR(\outcomes(v,\maxstrategy^*,\minstrategy)) =
  \MCR(v_0v_1\cdots v_\ell) \geq \Value(v_0)=\Value(v)\,.\]
  The base case $i=0$ corresponds to the previous case where the
  starting vertex is $\target$. Supposing that the property holds for
  index $i$, let us prove it for $i+1$. We
  have
  \[\MCR(v_{\ell-i-1}\cdots v_\ell) =
  \edgeweights(v_{\ell-i-1},v_{\ell-i}) + \MCR(v_{\ell-i}\cdots
  v_\ell)\,.\] By induction hypothesis, we have
  \begin{equation}
    \MCR(v_{\ell-i-1}\cdots v_\ell) \geq 
    \edgeweights(v_{\ell-i-1},v_{\ell-i}) +
    \Value(v_{\ell-i})\,.\label{eq:IH2}
  \end{equation}
  We now consider two cases:
  \begin{itemize}
  \item If $v_{\ell-i-1}\in\maxvertices\setminus\{\target\}$, then
    $v_{\ell-i}=\maxstrategy^*(v_0v_1\cdots v_{\ell-i-1})$, so that
    by definition of~$\maxstrategy^*$:
    \[\edgeweights(v_{\ell-i-1},v_{\ell-i}) +
    \Value(v_{\ell-i}) = \max_{v'\in\vertices\mid (v_{\ell-i-1},v')\in
      \edges} \big(\edgeweights(v_{\ell-i-1},v')+ \Value(v')\big)\,.\]
    Using Corollary~\ref{lem:min-max-charact-of-value} and
    \eqref{eq:IH2}, we obtain
    \[\MCR(v_{\ell-i-1}\cdots v_\ell) \geq
    \Value(v_{\ell-i-1})\,.\]
  \item If $v_{\ell-i-1}\in\minvertices\setminus\{\target\}$, then 
    \[\edgeweights(v_{\ell-i-1},v_{\ell-i}) +
    \Value(v_{\ell-i}) \geq \min_{v'\in\vertices\mid
      (v_{\ell-i-1},v')\in \edges} \big(\edgeweights(v_{\ell-i-1},v')+
    \Value(v')\big)\,.\] Once again using
    Corollary~\ref{lem:min-max-charact-of-value} and \eqref{eq:IH2},
    we obtain
    \[\MCR(v_{\ell-i-1}\cdots v_\ell) \geq
    \Value(v_{\ell-i-1})\,.\]
  \end{itemize}
  This concludes the proof.
\end{proof}

This strategy $\maxstrategy^\star$ can directly be computed along the
execution of the value iteration algorithm. This is done in
Algorithm~\ref{algo:value-iteration-RT-strategy}.

\begin{algorithm}[H]
  \DontPrintSemicolon%
  \KwIn{min-cost reachability game
    $\gameEx[\MCR]$, $W$ greatest weight in absolute
    value in the arena}%
  \SetKw{value}{\ensuremath{\mathsf{X}}}
  \SetKw{prevvalue}{\ensuremath{\mathsf{X}_{pre}}}
  
  \BlankLine

  $\value(\target) := 0$\;
  \lForEach{$v\in\vertices\setminus\{\target\}$}{$\value(v):=+\infty$}

  \Repeat{$\value = \prevvalue$}{%
    $\prevvalue := \value$\;%
    \ForEach{$v\in\maxvertices\setminus\{\target\}$}%
    {$\value(v) := \max_{v'\in\edges(v)}
      \big(\edgeweights(v,v')+\prevvalue(v')\big)$\;%
      \lIf{$\value(v)\neq \prevvalue(v)$}%
      { $\maxstrategy^*(v) := \argmax_{v'\in\edges(v)}
        \big(\edgeweights(v,v')+\prevvalue(v')\big)$ }%
    } %
    \ForEach{$v\in\minvertices\setminus\{\target\}$}%
    {$\value(v) := \min_{v'\in\edges(v)}
      \big(\edgeweights(v,v')+\prevvalue(v')\big)$\; %
      \If{$\value(v)\neq \prevvalue(v)$}%
      { $\minstrategy^1(v) := \argmax_{v'\in\edges(v)}
        \big(\edgeweights(v,v')+\prevvalue(v')\big)$\; %
        \lIf{$\prevvalue(v)=+\infty$}{$\minstrategy^2(v)=\minstrategy^1(v)$}%
      }%
    }%
    \ForEach{$v\in\vertices\setminus\{\target\}$}{%
      \lIf{$\value(v) < -(|\vertices|-1)   W$}%
      {$\value(v) := -\infty$}%
    } %
  }%
  \Return{$\value$}
  \caption{Computation of optimal strategy for both players in value
    iteration algorithm for min-cost reachability
    games}\label{algo:value-iteration-RT-strategy}
\end{algorithm}

\section{Reduction of total-payoff games to min-cost reachability
  games\label{sec:app-total-payoff}}

This section is devoted to the proof of Proposition~\ref{TrueTP2MCR}. 

For that purpose, we must relate paths in games $\game$ and $\game^n$:
with each finite path in $\game^n$, we associate a finite path in
$\game$, obtained by looking at the sequence of vertices of
$\vertices$ appearing inside the vertices of the finite play. Formally,
the \emph{projection} of a finite path $\pi$ is the sequence
$\proj(\pi)$ of vertices of $\game$ inductively defined by
$\proj(\varepsilon)=\varepsilon$ and for all finite path $\pi$,
$v\in\vertices$ and $1\leq j\leq n$:%
\begin{align*}
  &\proj((\interior,v,j)\pi)=\proj(\pi)\,, \qquad
  \proj((\exterior,v,j)\target\pi)= v\,, \qquad
  \proj((\exterior,v,j))=\varepsilon\,,\\
  & \proj((\exterior,v,j+1)(v,j)\pi)=\proj((v,j)\pi)=v\,\proj(\pi)\,.
\end{align*}
In particular, notice that in the case of a play with prefix
$(\exterior,v,j)\target$, the rest of the play is entirely composed of
target vertices $\target$, since $\target$ is a sink state. For
instance, the projection of the finite play 
$$(v_1,3)(\interior,v_2,3)
(\exterior,v_2,3)(v_2,2)(\interior,v_3,2)(v_3,2)
(\interior,\allowbreak v_3,2)\allowbreak (\exterior,v_3,2)\target$$
 of
the game $\game^3$ of Fig.~\ref{FigureExTPMCR} is given by
$v_1v_2v_3v_3$.

The following lemma relates plays of $\game^n$ with their projection
in $\game$, comparing their total-payoff.
\begin{lemma}\label{ClaimProjection} The projection mapping satisfies
  the following properties.
  \begin{enumerate}
  \item\label{item:partialPlays} If $\pi$ is a finite play in
    $\game^n$ then $\proj(\pi)$ is a finite play in $\game$.
  \item\label{item:plays} If $\pi$ is a play in $\game^n$ that
    does not reach the target, then $\proj(\pi)$ is a play in
    $\game$.
  \item\label{item:payoff} For all finite play $\pi$,
    $\TP(\pi)=\TP(\proj(\pi))$.
  \end{enumerate}
\end{lemma}
\begin{proof} 
  The proof of \ref{item:plays} is a direct consequence of
  \ref{item:partialPlays}. To each vertex
  $w\in\vertices^n\setminus\{\target\}$, we associate a vertex $f(w)$
  as follows:
  \[f(v,j)=f(\interior,v,j)=f(\exterior,v,j)=v\,.\]%
  Then notice that if $(w,w')\in \edges^n$ with $w,w'\neq \target$,
  then either $f(w)=f(w')$ or $(f(w),f(w'))\in \edges$. We now prove
  \ref{item:partialPlays} and \ref{item:payoff} inductively on the
  size of the finite play $\pi=w_1\cdots w_k$ of $\game^n$, along
  with the fact that
  \begin{center}
    $4.$ if $\proj(\pi)\neq \emptyset$ and $w_1\neq \target$ then the
    first vertex of $\proj(\pi)$ is $f(w_1)$.
  \end{center}
  If $k=0$, then $\pi=\proj(\pi)=\varepsilon$ are finite plays with
  the same total-payoff. If $k=1$, either $\proj(\pi)=\varepsilon$ or
  $\pi=(v,j)$ and $\proj(\pi)=v$: in both cases, the properties hold
  trivially. Otherwise, $k\geq 2$ and we distinguish several possible
  prefixes:
  \begin{itemize}
  \item If $\pi=(\interior,v,j)\pi'$, then
    $\proj(\pi)=\proj(\pi')$. Hence, \ref{item:partialPlays} holds by
    induction hypothesis. If $\proj(\pi)$ is non-empty, so is
    $\proj(\pi')$. Moreover, the first vertex of $\pi'$ is either
    $(v,j)$ or $(\exterior,v,j)$, so that we can show $4$ by induction
    hypothesis. Finally, the previous remark shows that the first edge
    of $\pi$ has necessarily weight 0, so that, $\TP(\pi)=\TP(\pi')$,
    and \ref{item:payoff} also holds by induction hypothesis.
  \item If $\pi=(v,j)\pi'$, then $\proj(\pi)=v\,\proj(\pi')$ so that
    $4$ holds directly. Moreover, $\pi'$ is a non-empty finite play so
    that $\pi'=(\interior,v',j)\pi''$ with $(v,v')\in\edges$, and
    $\proj(\pi')=\proj(\pi'')$. By induction, $\proj(\pi')$ is a
    finite play in $\arena$, and it starts with $v'$ (by $4$). Since
    $(v,v')\in\edges$, this shows that $\proj(\pi)$ is a
    finite play. Moreover,
    $\TP(\pi)=\edgeweights^n((v,j),(\interior,v',j))+\TP(\pi') =
    \edgeweights(v,v')+\TP(\pi')$.
    By induction hypothesis, we have
    $\TP(\pi')=\MCR(\proj(\pi'))$. Moreover,
    $\MCR(\proj(\pi)) = \edgeweights(v,v')+\MCR(\proj(\pi'))$ which
    concludes the proof of~\ref{item:payoff}.
  \item If $\pi=(\exterior,v,j)(v,j-1)\pi'$ then
    $\proj(\pi)=v\,\proj(\pi')=\proj((v,j-1)\pi')$: this allows us to
    conclude directly by using the previous case.
  \item Otherwise, $\pi=(\exterior,v,j)\target\pi'$, and then
    $\proj(\pi)=v$ is a finite play with total-payoff $0$, like
    $\pi$, and $4$ holds trivially. 
  \end{itemize}
\end{proof}

The next lemma states that when playing memoryless strategies, one can
bound the total-payoff of all finite plays.

\begin{lemma}\label{LemmaBoundedValue}
  Let $v\in \vertices$, and $\maxstrategy$ (respectively,
  $\minstrategy$) be a memoryless strategy for $\MaxPl$ (respectively,
  $\MinPl$) in the total-payoff game $\game$, such that
  $\Val(v,\maxstrategy)\neq -\infty$ (respectively,
  $\Val(v,\minstrategy)\neq +\infty$). Then for all finite play $\pi$
  conforming to $\maxstrategy$ (respectively, to $\minstrategy$),
  $\TP(\pi) \geq -(|\vertices|-1)   W$ (respectively,
  $\TP(\pi) \leq (|\vertices|-1)   W$).
\end{lemma}
\begin{proof} 
  We prove the part for $\MinPl$, the other case is similar.  The
  proof proceeds by induction on the size of a partial play
  $\pi=v_1\cdots v_k$ with $v_1=v$. If $k\leq|\vertices|$ then $\TP(\pi) =
  \sum_{i=1}^{k-1} \edgeweights(v_i,v_{i+1}) \leq (k-1)   W \leq
  (|\vertices|-1)  W$.  If $k\geq |\vertices|+1$ then there exists $i<j$ such
  that $v_i=v_j$.  Assume by contradiction that $\TP(v_i \cdots v_j) >
  0$.  Then the play $\pi' = v_1 \cdots v_i \cdots v_j (v_{i+1} \cdots
  v_j)^\omega$ conforms to $\minstrategy$ and $\TP(\pi') =+\infty$
  which contradicts $\Val(v,\minstrategy)\neq +\infty$.  Therefore
  $\TP(v_i \cdots v_j) \leq 0$.  We have $\TP(\pi) = \TP(v_1 \cdots
  v_i)+ \TP(v_i \cdots v_j) + \TP(v_{j+1} \cdots v_{k})$, and since
  $v_i=v_j$, $v_1 \cdots v_i v_{j+1} \cdots v_k$ is a finite play starting
  from $v$ that conforms to $\minstrategy$, and by induction
  hypothesis $\TP(v_1 \cdots v_iv_{j+1} \cdots v_{k} ) \leq
  (|\vertices|-1)  W$.  Then $\TP(\pi) = \TP(v_1 \cdots v_iv_{j+1} \cdots
  v_{k} )+ \TP(v_i \cdots v_j)\leq \TP(v_1 \cdots v_iv_{j+1} \cdots
  v_{k} )\leq (|\vertices|-1)  W$.  
\end{proof}

This permits to bound the finite values $\Val(v)$ of vertices $v$ of
the game:
\begin{corollary}\label{cor:bornesValeursTP}
  For all $v\in \vertices$, $\Val(v)\in \big[-(|\vertices|-1) 
  W,(|\vertices|-1)  W\big] \uplus\{-\infty,+\infty\}$.
\end{corollary}
\begin{proof}
  From \cite{GimZie04}, we know that total-payoff games are
  positionally determined, i.e., there exists two memoryless
  strategies $\maxstrategy,\minstrategy$ such that for all $v$,
  $\Val(v)= \Val(v,\maxstrategy) = \Val(v,\minstrategy)$. Assume that
  $\Val(v)\notin\{-\infty,+\infty\}$.  Then since
  $\Val(v,\minstrategy)=\Val(v)\neq -\infty$,
  Lemma~\ref{LemmaBoundedValue} shows that all finite play $\pi$ that
  conforms to $\minstrategy$ verifies
  $\TP(\pi) \geq -(|\vertices|-1) W$, therefore
  $\Val(v) \geq -(|\vertices|-1) W$. One can similarly prove that
  $\TP(v) \leq (|\vertices|-1) W$.  
\end{proof}

We now compare values in both games. A first lemma shows, in
particular, that $\Val_{\game^n}(v,n)\leq \Val_\game(v)$, in case
$\Val_\game(v)\neq +\infty$.

\begin{lemma}\label{lem:upperbound}
  For all $m\in \Z$, $v\in\vertices$, and $n\geq 1$, if $\Val_\game(v)
  \leq m$ then $\Val_{\game^n}(v,n)\leq m$.
\end{lemma}
\begin{proof} 
  By hypothesis and using the memoryless determinacy of
  \cite{GimZie04}, there exists a memoryless strategy $\minstrategy$
  for $\MinPl$ in $\game$ such that
  $\Val_\game(v,\minstrategy) \leq m$. Let $\minstrategy^m$ be the
  strategy in $\game^n$ defined, for all finite play $\pi$, vertex
  $v'$ and $1\leq j\leq n$, by
  \begin{align*}
    \minstrategy^m(\pi (v',j)) &= (\interior,\minstrategy(\proj(\pi)v'),j)\,,\\
    \minstrategy^m(\pi(\interior,v',j))&=
    \begin{cases}
      (v',j) & \text{if }\TP(\pi(\interior,v',j))\geq m+1\,,\\
      (\exterior,v',j) & \text{if } \TP(\pi(\interior,v',j)) \leq m \,.
    \end{cases}
  \end{align*} 
  Intuitively $\minstrategy^m$ simulates $\minstrategy$, and asks to
  leave the copy when the current total-payoff is less than or equal
  to $m$. Notice that, by construction of $\minstrategy^m$,
  $\proj(\pi)$ conforms to $\minstrategy$, if $\pi$ conforms to
  $\minstrategy^m$.

  As a first step, if a play $\pi$ starting in $(v,n)$ and conforming
  to $\minstrategy^m$ encounters the target then its value is at most
  $m$. Indeed, it is of the form
  $\pi = \pi' (\interior,v',j) (\exterior,v',j) \target^\omega$, and
  since it conforms to $\minstrategy^m$ we have
  \[\MCR(\pi) = \TP(\pi' (\interior,v',j) (\exterior,v',j)
  \target) = \TP(\pi'(\interior,v',j)) \leq m.\]
  
  Then, assume, by contradiction, that there exists a play $\pi$
  starting in $(v,n)$ and conforming to $\minstrategy^m$, that does
  not encounter the target. Then, this means that $\MinPl$ does not ask
  $n+1$ times the ability to exit in $\pi$ (since on the $(n+1)$th
  time that we jump in an exterior vertex, $\MaxPl$ is forced to go to
  the target). In particular, there exists $0\leq j\leq n$ such that
  $\pi$ is of the form
  $\pi' (v_1,j) (\interior,v_2,j) (v_2,j) (\interior,v_3,j) \cdots
  (v_k,j)(\interior,v_k,\allowbreak j)\cdots$.
  Since for all $i$,
  $\minstrategy^m(\pi' (v_1,j) (\interior,v_2,j) \cdots \allowbreak
  (\interior, v_i,j))=(v_i,j)$,
  we have that
  $\TP(\pi'\allowbreak (v_1,j) (\interior,v_2,j) \cdots
  (\interior,v_i,j))\geq m+1$.
  Therefore, since any prefix of $\proj(\pi)$ is the projection of a
  prefix of $\pi$, Lemma~\ref{ClaimProjection} shows that
  $\TP(\proj(\pi))\geq m+1>m$, which raises a contradiction since
  $\proj(\pi)$ conforms to $\minstrategy$ and
  $\Val(v,\minstrategy)\leq m$. Hence every play that conforms to
  $\minstrategy^m$ encounters the target, and, hence, has value at
  most~$m$. This implies that $\Val_{\game^n}(v,n)\leq m$. 
\end{proof}

We now turn to the other comparison between $\Val_{\game^n}(v,n)$
and $\Val_\game(v)$. Since $\Val_\game(v)$ can be infinite in case the
target is not reachable, we have to be more careful. In particular, we
show that $\Val_{\game^n}(v,n)\geq \min\big(\Val_\game(v),(|\vertices|-1)  
W +1\big)$ holds for large values of $n$. In the following, we let $K=|\vertices|
  (2 (|\vertices|-1)  W +1)$.

\begin{lemma}\label{lem:lowerbound}
  For all $m\leq (|\vertices|-1)W+1$, $k\geq K$, and vertex $v$, if $\Val_\game(v) \geq m$
  then $\Val_{\game^k}(v,k)\geq m$.
\end{lemma}
\begin{proof} 
  By hypothesis and using the memoryless determinacy of
  \cite{GimZie04}, there exists a memoryless strategy $\maxstrategy$
  for $\MaxPl$ in $\game$ such that
  $\Val_\game(v,\maxstrategy) \geq m$. Let $\maxstrategy^m$ be the
  strategy in $\game^K$ defined, for all finite play $\pi$, vertex
  $v'$ and $1\leq j\leq n$, by:
  \begin{align*}
    \maxstrategy^m(\pi (v,j)) &= (\interior,\maxstrategy(\proj(\pi)v),j)\,,\\
    \maxstrategy^m(\pi(\exterior,v,j))&=
    \begin{cases}
      (v,j-1)  &\text{if }\TP(\pi)\leq m-1 \text{ and } j>1\,,\\
      \target &\text{otherwise}\,.
    \end{cases}
  \end{align*}

  Intuitively $\maxstrategy^m$ simulates $\maxstrategy$, and accepts
  to go to the target when the current total-payoff is greater than or
  equal to $m$.

  By construction of $\maxstrategy^m$, if $\pi$ conforms to
  $\maxstrategy^m$, then $\proj(\pi)$ conforms to $\maxstrategy$.
  From the structure of the weighted graph, we know that for every
  play $\pi$ of $\game^k$, there exists $1\leq j\leq k$ such that
  $\pi$ is of the form
  $\pi_k (\exterior,v_k,k)\pi_{k-1} (\exterior,v_{k-1},k-1) \cdots
  \pi_j (\exterior,v_j,j) \pi'$
  verifying that: there are no occurrences of exterior vertices in
  $\pi_k,\ldots,\pi_j,\pi'$; for all $\ell\leq j$, all vertices in
  $\pi_\ell$ belong to the $\ell$-th copy of $\game$; either
  $\pi'=\target^\omega$ or all vertices of $\pi'$ belong to the
  $(j+1)$th copy of $\game$ (in which case, $j<k$).

  We now show that, in $\game^k$, $\MCR(\pi) \geq m$ for all play
  $\pi$ starting in $(v,k)$ and conforming to $\maxstrategy^m$. There
  are three cases to consider.
  \begin{enumerate}
  \item If $\pi$ does not reach the target, then $\MCR(\pi)=+\infty
    \geq m$.
  \item If $\pi = \pi_k (\exterior,v_k,k) \cdots \pi_j
    (\exterior,v_j,j) \target^\omega$ and $j>1$ then,
    \[\maxstrategy^m (\pi_k (\exterior,v_0,k) \cdots \pi_j
    (\exterior,v_j,j))=\target\,.\] %
    Thus, using Lemma~\ref{ClaimProjection}, 
    \begin{align*}
      \MCR(\pi) &= \TP(\pi_k
                  (\exterior,v_k,k) \cdots \pi_j
                  (\exterior,v_j,j)\target)\\ 
                &=\TP(\proj(\pi_k (\exterior,v_k,k) \cdots
                  \pi_j (\exterior,v_j,j)\target)) \\ 
                &\geq \Value_\game(v,\maxstrategy)
                  \geq m\,.
    \end{align*}
  \item If
    $\pi = \pi_k (\exterior,v_k,k) \cdots \pi_1 (\exterior,v_1,1)
    \target^\omega$, assume by contradiction that
    \[\TP (\pi_k (\exterior,v_k,k) \cdots \pi_1)\leq m-1\,.\]
    Otherwise, we directly obtain $\MCR(\pi)\geq m$.  Let $v^\star$ be
    a vertex that occurs at least
    $N=\left\lceil K/|\vertices|\right\rceil = 2 (|\vertices|-1) W +1$
    times in the sequence $v_1,\ldots,v_k$: such a vertex exists,
    since otherwise $K\leq k \leq (N-1) |\vertices|$ which contradicts
    the fact that $(N-1) |\vertices|<K$. Let $j_1>\cdots>j_N$ be a
    sequence of indices such that $v_{j_i} = v^\star$ for all $i$.  We
    give a new decomposition of $\pi$:
    \[\pi = \pi'_1 (\exterior,v_{j_1},j_1) \cdots \pi'_N
    (\exterior,v_{j_N},j_N) \pi'_{N+1}\,.\]%
    Since $\pi$ conforms to $\maxstrategy^m$ and according to the
    assumption, we have that for all $i$,
    \[\TP(\pi'_1 (\exterior,v_{j_1},j_1) \cdots \pi'_i)\leq m-1\,.\] We
    consider two cases.
    \begin{enumerate}
    \item If there exists $\pi'_i$ such that $\TP(\pi'_i) \leq 0$
      then, let $\proj(\pi'_i) = u_1\cdots u_\ell$ with $u_1 = u_\ell
      = v^\star$, Since $\pi'_i$ conforms to $\maxstrategy^m$,
      $\proj(\pi'_i)$ conforms to $\maxstrategy$.  Therefore the play
      \[\widetilde\pi=\proj(\pi'_1 (\exterior,v_{j_1},j_1) \cdots \pi'_i
      (\exterior,v_{j_i},j_i))(u_1 \cdots u_{\ell-1})^\omega\]
      conforms to $\maxstrategy$. Furthermore, using again
      Lemma~\ref{ClaimProjection},
      \[\hspace{-5mm}\TP(\widetilde\pi) = \liminf_{n\rightarrow +\infty}
        \Big(\TP\big(\pi'_1 (\exterior,v_{j_1},i_1) \cdots \pi'_i
        (\exterior,v_{j_i},j_i)\big)  + n \TP(u_1\cdots
        u_\ell)\Big)\]
      and since $\TP(u_1\cdots u_\ell)= \TP(\pi'_i)\leq 0$, we have
      \[\TP(\widetilde\pi) \leq \TP(\pi'_1 (\exterior,v_{j_1},i_1) \cdots \pi'_i
      (\exterior,v_{j_i},j_i)) \leq m-1.\] Thus $\widetilde\pi$ is a
      play starting from $v$ that conforms to $\maxstrategy$ but whose
      total-payoff is strictly less than $m$, which raises a
      contradiction.
    \item If for all $\pi'_i$, $\TP(\pi'_i)\geq 1$ (notice that it is
      implied by $\TP(\pi'_i)>0$). From Lemma~\ref{LemmaBoundedValue},
      since $\Value_\game(v,\maxstrategy)\geq m\neq-\infty$, we know
      that $\TP(\proj(\pi'_0))\geq -(|\vertices|-1)  W$. From
      Lemma~\ref{ClaimProjection}, $\TP(\pi'_0)\geq -(|\vertices|-1) 
      W$. Therefore
      \begin{align*}
        \TP(\pi'_1 (\exterior,v_{j_1},i_1) \cdots \pi'_N) &\geq
        -(|\vertices|-1)  W + N \\ &=(|\vertices|-1)   W +1 \geq m
      \end{align*}
      which contradicts the assumption that
      \begin{align*}
        \TP(\pi'_1(\exterior,v_{j_1},i_1) \cdots \pi'_N) < m. 
      \end{align*}
    \end{enumerate}
  \end{enumerate}
  We have shown that $\MCR(\pi) \geq m$ for all play
  $\pi$ starting in $(v,k)$ and conforming to $\maxstrategy^m$, which implies
  $\Value_{\game^K}((v,k),\maxstrategy^m)\geq
  m$. 
\end{proof}

From the two previous lemmas, we are ready to relate precisely values
in $\game$ and $\game^k$.

\begin{proof}[Proof of Proposition~\ref{TrueTP2MCR}]
  Let $v\in\vertices$.
  \begin{itemize}
  \item If $\Val_\game(v)= -\infty$, then for all $m$,
    $\Val_\game(v)\leq m$. Thus, by Lemma~\ref{lem:upperbound},
    $\Val_{\game^K}((v,K)) \leq m$. Therefore $\Val_{\game^K}(v,K) =
    -\infty$.
  \item If $\Val_\game(v)= m \in [-(|\vertices|-1)  W,(|\vertices|-1) 
    W]$. Then, $m\leq \Val_\game(v)\geq m$. Thus, by
    Lemma~\ref{lem:upperbound} and \ref{lem:lowerbound}, $m\leq
    \Val_{\game^K}((v,K)) \geq m$.  Therefore $\Val_{\game^K}((v,K)) =
    m$.
  \item If $\Val_\game(v) = +\infty$, then
    $\Val_\game(v)\geq (|\vertices|-1)  W+1$.  Thus, by
    Lemma~\ref{lem:lowerbound},
    $\Val_{\game^K}((v,K)) \geq (|\vertices|-1)W+1$.  
  \end{itemize}
\end{proof}

\section{Value iteration algorithm for total-payoff
  games\label{app:VI-total-payoff}}

This section is devoted to the study of
Algorithm~\ref{algo:value-iter-TPO}, in particular in the proof of
Theorem~\ref{thm:VI-TP}.

We first define formally the game $\game_Y$ described informally on
page \pageref{GY}. To the original total-payoff game
$\game=\gameEx[\TP]$ and to every vector $Y\in\Zbar^\vertices$, we
associate the min-cost reachability game
$\game_Y=\tuple{\vertices',\edges_Y,\edgeweights_Y,\MCR[\{\target\}]}$
as follows. The sets of vertices are given by
\[\maxvertices' = \maxvertices \uplus 
\{ (\interior,v) \mid v\in \vertices\} \uplus \{ \target\} \quad
\text{and} \quad \minvertices'=\minvertices\,.\]
As in game $\game^j$, vertices of the form $(\interior,v)$ are called
\emph{interior vertices}. Edges are defined by
\begin{align*}
  \edges_Y &= \phantom{{}\uplus{}} \left\{ \big( (v , (\interior,v')
             \big) \mid (v,v')\in \edges \right\} \uplus \left\{ \big(
             (\interior,v) , v\big) \mid v\in \vertices \right\} \\
           &\phantom{{}={}}\uplus \left\{ \big( (\interior,v), \target \big)
             \mid v\in \vertices \land Y(v)\neq +\infty \right\} \uplus \left\{ ( \target, \target )
             \right\}
\end{align*}
while weights of edges are defined, for all $(v,v')\in \edges$, by
\begin{align*}
  \edgeweights_Y \big ( v , (\interior,v') \big) &=
  \edgeweights_Y(v,v')\,, & \edgeweights_Y \big ( (\interior,v),
  \target  \big) &= \max \big ( 0, Y(v) \big ) \,, \\
  \edgeweights_Y \big( (\interior,v) , v \big) &= 0. & & 
\end{align*}

It is easy to see that lines~\ref{line:begin} to~\ref{line:end} are a
rewriting of Algorithm~\ref{algo:value-iteration-RT} in the special
case of game $\game_Y$: in particular, neither the target vertex nor
interior vertices are explicit, but their behavior is taken into
account by the transformation performed in line~\ref{line:begin} and
the operators $\min$ used in the inner computation of
lines~\ref{line:minvertex} and \ref{line:maxvertex}. Hence, if we
define $\operatorBis(Y)(v)=\Value_{\game_Y}(v)$ for all
$v\in\vertices$, we can say that if inside the main loop, at
line~\ref{line:begin} the variable $\mathsf Y$ has value $Y$, then
after line~\ref{line:end}, it has value $\operatorBis(Y)$.

Notice that the game $\game_Y$ resembles a copy of $\game$ in the
game $\game^j$ of the previous section. More, precisely, from the
values $(\Value_{\game^j}(v,j))_{v\in\vertices}$ in the $j$th copy, we
can deduce the values in the $(j+1)$th copy by an application of
operator $\operatorBis$:
\[\big(\Value_{\game^{j+1}}(v,j+1)\big)_{v\in\vertices}=
\operatorBis\big((\Value_{\game^j}(v,j))_{v\in\vertices}\big)\,.\]
Although the $0$th copy is not defined, we abuse the notation and set
$\Value_{\game^0}(v,0)=-\infty$, which still conforms to the above
equality. Furthermore, due to the structure of the game $\game^j$
notice that for all $j\leq j'$, $\Value_{\game^j}(v,j)=
\Value_{\game^{j'}}(v,j)$.

Notice the absence of exterior vertices $(\exterior,v',j)$ in game
$\game_Y$, replaced by the computation of the maximum between $0$ and
$X(v')$ on the edge towards the target.  Before proving the
correctness of Algorithm~\ref{algo:value-iter-TPO}, we prove several
interesting properties of operator $\operatorBis$.

\begin{proposition}
  $\operatorBis$ is a monotonic operator.
\end{proposition}
\begin{proof}
  For every vector $Y\in\Zbar^\vertices$, let $\operator_Y$ be the
  operator associated with the min-cost reachability game as
  defined in Section~\ref{sec:reachability-objectives}, i.e., for all
  $X\in \Zbar^{\vertices'}$, and for all $v_1\in \vertices'$
  \[\operator_Y(X)(v_1)=
  \begin{cases}
    \displaystyle{\max_{v_2\in \edges_Y(v_1)}}
    \big(\edgeweights_Y(v_1,v_2)+X(v_2)\big) &
    \text{if } v\in \maxvertices'\setminus\{\target\}\\
    \displaystyle{\min_{v_2\in \edges_Y(v_1)}}
    \big(\edgeweights_Y(v_1,v_2)+X(v_2)\big) &
    \text{if } v\in \minvertices\\
    0 & \text{if } v_1=\target\,.
  \end{cases}\]%
  We know from Corollary~\ref{lem:min-max-charact-of-value} that
  $\Value_{\game_Y}$ is the greatest fixed point of
  $\operator_Y$. Consider now two vectors $Y,Y'\in\Zbar^\vertices$
  such that $Y\vleq Y'$.

  First, notice that for all $X\in\Zbar^{\vertices'}$:
  \begin{equation}\label{eqFixPointFxFxp}
    \operator_Y(X)\vleq \operator_{Y'}(X)\,.
  \end{equation}
  Indeed, to get the result it suffices to notice that for all
  $v_1,v_2\in \vertices'$, $\edgeweights_Y(v_1,v_2)\leq
  \edgeweights_{Y'}(v_1,v_2)$.

  Consider then the vector $X_0$ defined by $X_0(v_1)=+\infty$ for all
  $v_1\in \vertices'$. From~\eqref{eqFixPointFxFxp}, we have that
  $\operator_{Y}(X_0)\vleq \operator_{Y'}(X_0)$, then a simple
  induction shows that for all $i$, $\operator^i_{Y}(X_0)\vleq
  \operator^i_{Y'}(X_0)$. Thus, since $\Value_{\game_Y}$
  (respectively, $\Value_{\game_{Y'}}$) is the greatest fixed point of
  $\operator_{Y}$ (respectively, $\operator_{Y'}$), we have
  $\Value_{\game_Y}\vleq\Value_{\game_{Y'}}$.  As a consequence
  $\operatorBis(Y)\vleq \operatorBis(Y')$.
\end{proof}

Notice that $\operatorBis$ may not be Scott-continuous, as shown in
the following example.

\begin{example}
  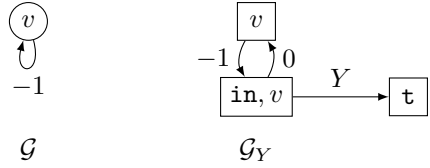
\begin{figure}[tbp]\centering
    \begin{tikzpicture}[>=latex]
      \begin{scope}
        \node[player1] (v) at (0,0) {\makebox[0mm][c]{$v$}};
        \draw[->] (v) edge[loop below] node [midway, below] {$-1$} (v);
        
        \node at (0,-1.7) {$\game$};
        
      \end{scope}
      
      \begin{scope}[xshift = 3cm]
        \node[player2] (vBis) at (0,0) {\makebox[0mm][c]{$v$}};
        \node[player2] (inv) at (0,-1) {$\interior,v$};
        \node[player2] (tg) at (2,-1) {\makebox[0mm][c]{$\target$}};
        
        \draw[->] (vBis) to[bend right] node[midway, left] {$-1$} (inv);
        \draw[->] (inv) to[bend right] node[midway, right] {$0$} (vBis);
        \draw[->] (inv) to node[above] {$Y$} (tg);
        
        \node at (0,-1.7) {$\game_Y$};
      \end{scope}
      
    \end{tikzpicture}
    \caption{The games $\game$ and $\game_Y$.}
    \label{fig:notContinuous}
  \end{figure}
  Recall that, in our setting, a Scott-continuous operator is a
  mapping $F: \Zbar^\vertices \rightarrow \Zbar^\vertices$ such that
  for any sequence of vectors $(x_i)_{i\geq 0}$ having a limit
  $x_{\omega}$, the sequence $(F(x_i))_{i\geq 0}$ has a limit equal to
  $F(x_{\omega})$.

  We present a total-payoff game whose associated operator
  $\operatorBis$ is not continuous.  Let $\game$ be the total-payoff
  game containing one vertex $v$ of $\MinPl$ and a self loop of weight
  $-1$ (as depicted in Fig.~\ref{fig:notContinuous}).  For all
  $Y\in \Z$, in the min-cost reachability game $\game_Y$, $v$ has
  value $-\infty$, indeed one can take the loop an arbitrary number of
  times before reaching the target, ensuring a value arbitrary
  low. Therefore, if we take an increasing sequence $(Y_i)_{i\geq 0}$
  of integers, $\operatorBis(Y_i)(v)=-\infty$ for all $i$, thus the
  limit of the sequence $(\operatorBis(Y_i))_{i\geq 0}$ is $-\infty$.
  However, the limit of the sequence $(Y_i)_{i\geq 0}$ is $+\infty$
  and $\operatorBis(+\infty)(v)=+\infty$, since the target is not
  reachable anymore (in case the weight of an edge would be $+\infty$,
  it is removed in the definition of $\edges_Y$). Thus, $\operatorBis$
  is not Scott-continuous. \qed
\end{example}

In particular, we may not use the Kleene sequence, as we have done for
min-cost reachability games, to conclude to the correctness of our
algorithm. Anyhow, we will show that the sequence $(Y^j)_{j\geq 0}$
indeed converges towards 
the vector of values of the total-payoff game. We first show that this
vector is a pre-fixed point of $\operatorBis$
starting with a technical lemma that is useful in the subsequent
proof.

\begin{lemma}\label{PropCoolDesTP}
  Let $\minstrategy$ be a strategy for $\MinPl$ in $\game$, and
  $\pi=v_1\cdots v_i$ a finite play that conforms to
  $\minstrategy$. Then:
  \[ \TP(v_1\cdots v_i) + \Val_\game(v_i) \leq
  \Val_\game(v_1,\minstrategy)\,.\]
\end{lemma}
\begin{proof}
  Let $\maxstrategy$ be an optimal strategy for $\MinPl$ and
  $v_i v_{i+1} v_{i+2}\cdots$ be the play
  $\outcomes(v_i,\maxstrategy,\minstrategy)$. Since $\maxstrategy$ is
  optimal, $\TP(v_i v_{i+1} v_{i+2}\cdots) \geq \Val_\game(v_i)$.
  Furthermore notice that $v_1 \cdots v_i v_{i+1} \cdots$ conforms to
  the strategy $\minstrategy$, therefore
  $\TP (v_1 v_1 \cdots v_i v_{i+1} \cdots) \leq
  \Val_\game(v_1,\minstrategy)$. Thus:
  \begin{align*}
    \TP(v_1\cdots v_i) + \Val_\game(v_i) & \leq \TP(v_1\cdots v_i) +
    \TP(v_i v_{i+1} \cdots) \\
    & \leq  \TP(v_1 v_1 \cdots v_i v_{i+1} \cdots) \leq
    \Val_\game(v_1,\minstrategy)\,.
  \end{align*}
\end{proof}

\begin{lemma}\label{lem:pre-fixed-point}
  $\Val_\game$ is a pre-fixed point of $\operatorBis$, i.e.,
  $\operatorBis(\Val_\game) \vleq \Val_\game$.
\end{lemma}
\begin{proof}
  To ease the notations, we denote $\Val_\game$ by $Y^\star$ in this
  proof.  To prove this lemma, we just have to show that for all
  $v_1\in \vertices$, the value of $v_1$ in the min-cost reachability
  game $\game_{Y^\star}$ is at most its value in the original
  total-payoff game $\game$, i.e.,
  \[\operatorBis(Y^\star)(v_1)=\Value_{\game_{Y^\star}}(v_1) \leq
  Y^\star(v_1)\,.\]

  Let $\minstrategy$ be a memoryless strategy in $\game$ such that
  $\Val_\game(v_1,\minstrategy)\leq m$ for some
  $m \in \Z \uplus\{+\infty\}$. And let $\minstrategy^m$ be a strategy
  in $\game_X$ defined for all finite play $\pi v'$ with
  $v'\in \minvertices$ by
  $\minstrategy^m(\pi v') = (\interior, \minstrategy(v'))$ and for all
  finite play $\pi (\interior,v')$,
  \[ \minstrategy^m(\interior,v') = \begin{cases}
    \target & \text{if } \TP(\pi (\interior,v') \target)\leq m \\
    v & \text{otherwise}\,.
  \end{cases}\]

  Notice that, by construction, all plays starting in $v_1$,
  conforming to $\minstrategy^m$ and that reach the target have value at
  most $m$. Assume by contradiction that there exists a play $v_1
  (\interior,v_2) v_2 (\interior,v_3) \cdots\in
  \outcomes(v_1,\minstrategy^m)$ that never reaches $\target$. In
  particular, for all~$i$, $\TP(v_1 (\interior,v_2) \cdots
  (\interior,v_i) \target)>m$.

  Again by construction, $v_1 v_2 \cdots$ is a play in $\game$ that
  conforms to $\minstrategy$ and for all~$i$,
  $\TP(v_1 (\interior,v_2) \cdots \allowbreak (\interior,v_i))=
  \TP(v_1\cdots v_i)$.
  \begin{itemize}
  \item If there exists $i\geq 2$ such that
    $Y^\star(v_i)=\Val_\game(v_i)\geq 0$, then
    \begin{align*}
      \TP(v_1 (\interior,v_2) \cdots (\interior,v_i) \target) & = \TP(
      v_1 (\interior,v_2) \cdots (\interior,v_i)) +
      \edgeweights_X((\interior,v_i), \target)\\
      & = \TP(v_1\cdots v_i) + \max(0,Y^\star(v_i))\\
      & = \TP(v_1\cdots v_i) +Y^\star(v_i) \\
      & \leq \Val(v_1,\minstrategy) \qquad\qquad \text{(from
        Lemma~\ref{PropCoolDesTP})}\\ 
      & \leq m
    \end{align*}
    which raises a contradiction.
  \item If for all $i\geq 2$, $Y^\star(v_i)=\Val_\game(v_i)<0$ then
    for all $i\geq 2$,
    \[\TP(v_1 (\interior,v_2) \cdots (\interior,\allowbreak v_i)
    \target) = \TP(v_1 (\interior,v_2) \cdots
    (\interior,v_i))=\TP(v_1\cdots v_i)>m\,.\]
    Thus $\TP(v_1 v_2 \cdots)>m$, which contradicts the fact that
    $v_1v_2\cdots$ conforms to $\minstrategy$ and
    $\Val_\game(v_1,\minstrategy)\leq m$.
  \end{itemize}
  Thus $\Value_{\game_{Y^\star}}(v_1,\minstrategy) \leq m$. As a consequence,
  $\Value_{\game_{Y^\star}}(v_1) \leq \Val_\game(v_1)$.
\end{proof}

\begin{remark}
  Even if it is not necessary for the proof of
  Theorem~\ref{thm:VI-TP}, we can show that $\Value_\game$ is the
  least pre-fixed point of $\operatorBis$. Notice that, by
  monotonicity of $\operatorBis$, this directly implies that
  $\Val_\game$ is the least fixed point of $\operatorBis$. The proof
  that $\Value_\game$ is the least pre-fixed point of $\operatorBis$
  amounts to better understand the convergence of the sequence
  $(\Value_{\game^j}(v,j))_{j\geq 0}$ (remember that we set
  $\Value_{\game^0}(v,0)=-\infty$ for all $v$). Indeed,
  Proposition~\ref{TrueTP2MCR} already shows that for vertices $v$
  such that $\Value_\game(v)<+\infty$,
  $(\Value_{\game^j}(v,j))_{j\geq 1}$ converges towards
  $\Value_\game(v)$. It is also the case for vertices $v$ such that
  $\Value_\game(v)=+\infty$ as we show in Lemma~\ref{lem:limit}
  below. Then, consider a pre-fixed point $Y$ of $\operatorBis$, i.e.,
  $\operatorBis(Y)\vleq Y$. Since
  $\Value_{\game^0}(v,0)=-\infty\leq Y(v)$ and $\operatorBis$ is
  monotonous, we prove by immediate induction that
  $\Value_{\game^j}(v,j)\leq Y(v)$ for all $v$ and $j\geq 0$: indeed,
  if $\Value_{\game^j}(v,j)\leq Y(v)$ for all $v$, we have
  \[(\Value_{\game^{j+1}}(v,j+1))_{v\in
    V}=\operatorBis((\Value_{\game^j}(v,j))_{v\in V})\vleq
  \operatorBis(Y)\vleq Y\,.\]
  This implies that $\Value_{\game}\vleq Y$, showing that
  $\Value_{\game}$ is indeed the least pre-fixed point of
  $\operatorBis$, and hence the least fixed point of $\operatorBis$,
  by the above reasoning.
\end{remark}

Before continuing the proof of Theorem~\ref{thm:VI-TP}, we show the
result used in the previous remark.

\begin{lemma}\label{lem:limit}
  Let $v\in \vertices$ such that $\Val_\game(v)=+\infty$, and
  $\maxstrategy$ a memoryless strategy for $\MaxPl$ in $\game$ such
  that $\Val_\game(v,\maxstrategy)=+\infty$. Then the following holds:
  \begin{enumerate}
  \item[$(i)$] For every finite play $v_1 \cdots v_k$ conforming to
    $\maxstrategy$ starting in $v_1=v$, if there exists $i<j$ such
    that $v_i=v_j$ then $\TP(v_i\cdots v_j) \geq 1$.
  \item[$(ii)$] For every $m\in \N$, $k\geq m|V|+1$ and
    $v_1 \cdots v_k$ a finite play conforming to $\maxstrategy$ and
    starting in $v_1=v$, $\TP( v_1 \cdots v_k) \geq m-(|V|-1)W$.
  \item[$(iii)$] For all $m\in \N$ and $k\geq (m+(|V|-1)W) |V|+1$,
    $\Val_{\game^{k}}(v,k)\geq m$.
  \item[$(iv)$] $\lim_{j\to \infty} \Val_{\game^{j}}(v,j)= +\infty$.
  \end{enumerate}
\end{lemma}
\begin{proof}
  We prove $(i)$ by contradiction. Therefore, assume that
  $\TP(v_i\cdots v_j) \leq 0$, then
  $\pi=v_1 \cdots v_{i-1} (v_i \cdots v_{j-1})^\omega$ comforms to
  $\maxstrategy$ and $\TP(\pi)\leq \TP(v_1 \cdots v_{i-1})<+\infty$,
  which contradicts the fact that
  $\Val_\game(v,\maxstrategy)=+\infty$.

  We prove $(ii)$ by induction on $m$. The base case is
  straightforward. For the inductive case, let $m>0$, $k\geq m|V|+1$
  and $v_1 \cdots v_k$ a finite play conforming to $\maxstrategy$ and
  starting in $v_1=v$. Since $k\geq m|V|+1 \geq |V|+1$ there exists
  $i<j\leq |V|+1$ such that $v_i=v_j$. Thus:
  \begin{align*}
    \TP(v_1 \cdots v_k) & = \TP(v_1\cdots v_i)+ \TP(v_i\cdots v_j)+\TP(v_j\cdots v_k)\\
                        & =\TP(v_1\cdots v_i v_{j+1}\cdots v_k)+ \TP(v_i\cdots v_j)\\
                        & \geq \TP(v_1\cdots v_i v_{j+1}\cdots v_k)+
                          1\,. && \text{from $(i)$}
  \end{align*}
  By induction hypothesis, as $v_1\cdots v_i v_{j+1}\cdots v_k$
  conforms to $\maxstrategy$ and has length at least $(m-1)|V|+1$
  (because $i<j\leq |V|+1$, implying that $j-i\leq |V|$), we have
  $\TP(v_1\cdots v_i v_{j+1}\cdots v_k)\geq m-1-(|V|-1)W$, thus
  $\TP( v_1 \cdots v_k) \geq m-(|V|-1)W$.

  To prove $(iii)$, let $\maxstrategy'$ be a strategy of $\MaxPl$ in
  $\game^k$ defined by
  $\maxstrategy'(v,j)=(\interior,\maxstrategy(v),j) $ and
  $\maxstrategy'(\exterior,v,j)= (v,j-1)$ for all $v\in \vertices$ and
  $j\geq k$. Let $\pi$ be a play starting in $(v,k)$ and conforming to
  $\maxstrategy'$. If $\pi$ does not reach $\target$, then
  $\MCR(\pi)=+\infty \geq m$. If $\pi$ reaches the target then
  $\proj(\pi)$ is of the form $v_1\cdots v_\ell\target^\omega$, with
  $\MCR(\pi)=\TP(v_1\cdots v_\ell)$. It is clear by construction of
  $\maxstrategy'$ that $v_1\cdots v_\ell$ is a finite play of $\game$
  that conforms to $\maxstrategy$. Furthermore,
  $\ell \geq k\geq (m+(|V|-1)W) |V|+1$ thus, from $(ii)$, we have that
  $\TP(v_1\cdots v_\ell)\geq m$. This implies $\MCR(\pi)\geq
  m$.
  Hence, every play in $\game^k$ conforming to $\maxstrategy'$ and
  starting in $(v,k)$ has a value at least $m$, which means that
  $\Val_{\game^{k}}(v,k)\geq m$.

  Item $(iv)$ is then a direct consequence of $(iii)$.
\end{proof}

We are now ready to state and prove the inductive invariant allowing us
to show the correctness of Algorithm~\ref{algo:value-iter-TPO}.
\begin{lemma}\label{lem:invariant}
  Before the $j$-th iteration of the external loop of
  Algorithm~\ref{algo:value-iter-TPO}, we have $\Value_{\game^j}(v,j)
  \leq Y^j(v) \leq \Value_\game(v)$ for all vertices $v\in\vertices$.
\end{lemma}
\begin{proof}
  For $j=0$, we have $Y^0(v)=-\infty=\Value_{\game^0}(v,0)$ for all
  vertex $v\in \vertices$.  Suppose then that the invariant holds for
  $j\geq 0$. We know that
  $\Value_{\game^{j+1}}(v,j+1)
  =\operatorBis((\Value_{\game^{j}}(v',j))_{v'\in\vertices})$.
  Moreover, after the assignment of line~\ref{line:end}, by definition
  of $\operatorBis$, variable $\mathsf Y$ contains
  $\operatorBis(Y^j)$. The operation performed on
  line~\ref{line:3-line-infty} only increases the values of vector
  $\mathsf Y$, so that at the end of the $j$th iteration, we have
  $\operatorBis(Y^j)\vleq Y^{j+1}$. Since $\operatorBis$ is
  monotonous, and by the invariant at step $j$, we obtain
  \[\Value_{\game^{j+1}}(v,j+1)
  =\operatorBis((\Value_{\game^{j}}(v',j))_{v'\in\vertices}) \leq
  \operatorBis(Y^j)\leq Y^{j+1}\,.\]
  Moreover, using again the monotony of $\operatorBis$ and
  Lemma~\ref{lem:pre-fixed-point}, we have
  \[\operatorBis(Y^j)\vleq \operatorBis(\Value_\game)\vleq \Value_\game\,.\]
  A closer look at line~\ref{line:3-line-infty} shows that
  $\operatorBis(Y^j)$ and $Y^{j+1}$ coincide over vertices $v$ such
  that $\operatorBis(Y^j)(v)\leq (|\vertices|-1)  W$, and
  otherwise $Y^{j+1}(v)=+\infty$. Hence, if $\operatorBis(Y^j)(v)\leq
  (|\vertices|-1)  W$, we directly obtain
  $Y^{j+1}(v)=\operatorBis(Y^j)(v)\leq \Value_\game(v)$. Otherwise, we
  know that $\Value_\game(v)>(|\vertices|-1)  W$. By
  Corollary~\ref{cor:bornesValeursTP}, we know that
  $\Value_\game(v)=+\infty$, so that $Y^{j+1}(v)=+\infty =
  \Value_\game(v)$. In the overall, we have proved
  \begin{equation}
    \Value_{\game^{j+1}}(v,j+1) \leq Y^{j+1}(v) \leq
    \Value_\game(v)
  \end{equation}
\end{proof}

We are now able to prove the correction of the algorithm.
\begin{proof}[Proof of Theorem~\ref{thm:VI-TP}]
  For $j=K$ (remember that $K$ was defined in the previous section),
  the invariant of Lemma~\ref{lem:invariant} becomes
  \[\Value_{\game^{K}}(v,K) \leq Y^{K}(v) \leq
  \Value_\game(v)\] for all vertices $v\in\vertices$. Notice that the
  iteration may have stopped before iteration $K$, in which case the
  sequence $(Y^{j})_{j\geq 0}$ may be considered as stationary. In
  case $\Value_\game(v)\neq +\infty$, Proposition~\ref{TrueTP2MCR}
  proves that $\Value_{\game^{K}}(v,K)=\Value_\game(v)$, so that we
  have $Y^{K}(v) = \Value_\game(v)$. In case $\Value_\game(v)=
  +\infty$, Proposition~\ref{TrueTP2MCR} shows that
  $\Value_{\game^{K}}(v,K)>(|\vertices|-1)  W$: by the
  operation performed at line~\ref{line:3-line-infty}, we obtain that
  $Y^K(v)=+\infty=\Value_\game(v)$.

  Hence, $K=|\vertices|  (2 (|\vertices|-1)   W +1)$ is an upper bound on
  the number of iterations before convergence of
  Algorithm~\ref{algo:value-iter-TPO}, and moreover, at the
  convergence, the algorithm outputs the vector of optimal values of
  the total-payoff game.
\end{proof}

\section{An example of parametric total-payoff
  game\label{sec:monster-example}}

We depict in Fig.~\ref{fig:monster-example} a weighted graph
parametrized with the number $n$ of \emph{layers}, as well as the
greatest weight $W>0$. For both the min-cost reachability objective
(with $\target$ the target) and the total-payoff objectives, the
values of the vertices are as follows: vertices $v_{3k+1}$ and
$v_{3k+2}$ ($k\in\{0,\ldots,n-1\}$) have value $0$, whereas vertices
$v_{3k}$ ($k\in\{1,\ldots,n\}$) have value $W$. In our add-on
prototype of PRISM, we model the min-cost objective with
$\langle\!\langle \MaxPl \rangle\!\rangle \mathsf{R}^{\max=?} [
\mathsf{F}^\infty \target]$
for $\MaxPl$ and
$\langle\!\langle \MinPl\rangle\!\rangle \mathsf{R}^{\min=?} [
\mathsf{F}^\infty \target]$
for $\MinPl$, whereas total-payoff objectives are modelled by
$\langle\!\langle \MaxPl\rangle\!\rangle \mathsf{R}^{\max=?} [
\mathsf{F}^c \bot]$
and
$\langle\!\langle \MinPl\rangle\!\rangle \mathsf{R}^{\min=?} [
\mathsf{F}^c \bot]$.

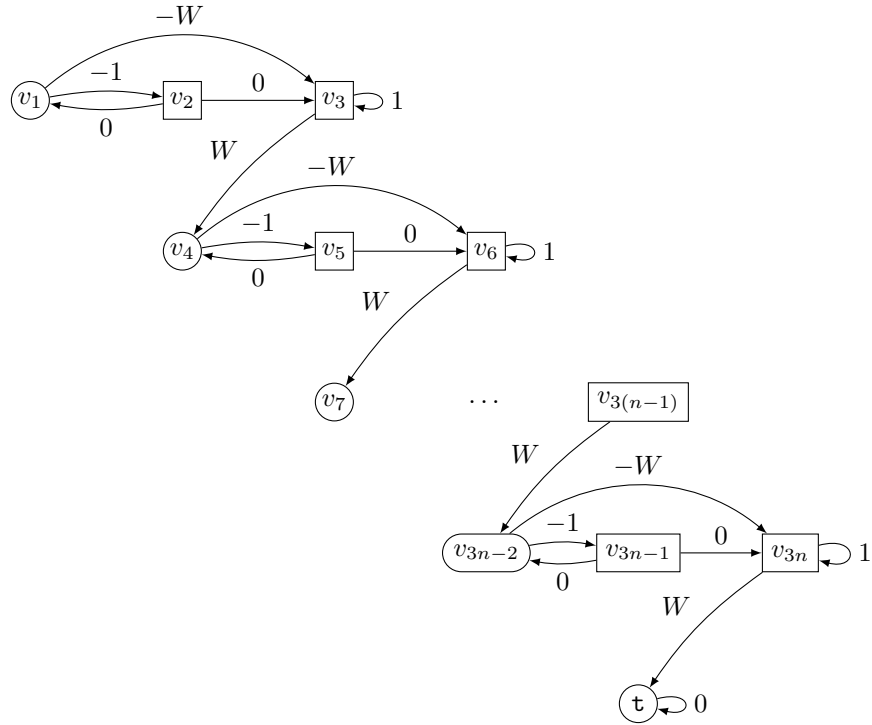
\begin{figure}[tbp]
  \centering
  \begin{tikzpicture}[node distance=2cm,auto,->,>=latex]
    \node[player1](1){\makebox[0mm][c]{$v_1$}}; 
    \node[player2](2)[right of=1]{\makebox[0mm][c]{$v_2$}}; 
    \node[player2](3)[right of=2]{\makebox[0mm][c]{$v_3$}};
    \node[player1](4)[below of=2]{\makebox[0mm][c]{$v_4$}}; 
    \node[player2](5)[right of=4]{\makebox[0mm][c]{$v_5$}}; 
    \node[player2](6)[right of=5]{\makebox[0mm][c]{$v_6$}};
    \node[player1](7)[below of=5]{\makebox[0mm][c]{$v_7$}}; 

    \node()[right of=7]{\ldots};
    \node[player2](3n-3)[right of=7,xshift=2cm]{$v_{3(n-1)}$};
    \node[player1](3n-2)[below of=7, xshift=2cm]{$v_{3n-2}$}; 
    \node[player2](3n-1)[right of=3n-2]{$v_{3n-1}$}; 
    \node[player2](3n)[right of=3n-1]{$v_{3n}$};
    \node[player1](tg)[below of=3n-1]{\makebox[0mm][c]{$\target$}};

    \path 
    (1) edge[bend left=40] node[above]{$-W$} (3)
    (1) edge[bend left=10] node[above]{$-1$} (2) 
    (2) edge[bend left=10] node[below]{$0$} (1)
    edge node[above]{$0$} (3)
    (3) edge[loop right] node[right]{$1$} (3)
    (3) edge[bend right=10] node[above left]{$W$} (4);

    \path 
    (4) edge[bend left=40] node[above]{$-W$} (6)
    (4) edge[bend left=10] node[above]{$-1$} (5) 
    (5) edge[bend left=10] node[below]{$0$} (4)
    edge node[above]{$0$} (6)
    (6) edge[loop right] node[right]{$1$} (6)
    (6) edge[bend right=10] node[above left]{$W$} (7);

    \path 
    (3n-3) edge[bend right=10] node[above left]{$W$} (3n-2)
    (3n-2) edge[bend left=40] node[above]{$-W$} (3n)
    (3n-2) edge[bend left=10] node[above]{$-1$} (3n-1) 
    (3n-1) edge[bend left=10] node[below]{$0$} (3n-2)
    edge node[above]{$0$} (3n)
    (3n) edge[loop right] node[right]{$1$} (3n)
    (3n) edge[bend right=10]
    node[above left]{$W$} (tg)
    (tg) edge[loop right] node[right]{$0$} (tg);
    
  \end{tikzpicture}
  \caption{Parametric weighted graph}
  \label{fig:monster-example}
\end{figure}

We present in the following table the time for resolution (in
seconds), the number of iterations in the external loop, and the total
number of iterations in the internal loops for the total-payoff
resolution, for various values of parameters $W$ and $n$:

\smallskip
{\hspace{-1cm}\smaller[2]
  \begin{tabular}{|c|c|c|c|c|c|c|} \hline
    $W \backslash n$ & 100 & 200 & 300 & 400 & 500 \\\hline %
    50 & 0.52 / 151 / 12603 & 1.90 / 251 / 22703 & 3.84 / 351 / 32803 & 6.05
    / 451 / 42903 & 9.83 / 551 / 53003 \\\hline %
    100 & 1.00 / 201 / 30103 & 3.48 / 301 / 50203 & 8.64 / 401 / 70303 & 13.53 / 501 / 90403 & 22.64 / 601 / 110503 \\\hline %
    150 & 1.89 / 251 / 52603 & 6.02 / 351 / 82703 & 12.88 / 451 / 112803 & 22.13 / 551 / 142903 & 34.16 / 651 / 173003 \\\hline 
    200 & 2.96 / 301 / 80103 & 9.62 / 401 / 120203 & 18.33 / 501 / 160303 & 30.42 / 601 / 200403 & 45.64 / 701 / 240503 \\\hline %
    250 & 3.92 / 351 / 112603 & 13.28 / 451 / 162703 & 25.18 / 551 / 212803 & 46.23 / 651 / 262903 & 71.51 / 751 / 313003 \\\hline
\end{tabular}
}
\medskip

Notice that due to the very little memory consumption of the
algorithm, there is no risk of running out of memory. However, the
execution time can become very large. For instance, in case $W=500$
and $n=1000$, the execution time becomes $536 s$ whereas the total
number of iterations in the internal loop is greater than a million.

On this example, with $n+1$ components, the acceleration heuristics
presented in details in Appendix~\ref{sec:accel-heur-rtp} gives
excellent results. Indeed, by combining both heuristics, we obtain the
following results:

\smallskip
{\smaller[2]
  \begin{tabular}{|c|c|c|c|c|c|c|} \hline
    $W \backslash n$ & 100 & 200 & 300 & 400 & 500 \\\hline %
    50 & 0.01 / 402 / 1404 & 0.08 / 802 / 2804 & 0.22 / 1202 / 4204 & 0.38
    / 1602 / 5604 & 0.42 / 2002 / 7004 \\\hline %
    100 & 0.02 / 402 / 1404 & 0.09 / 802 / 2804 & 0.19 / 1202 / 4204 & 0.33 / 1602 / 5604 & 0.40 / 2002 / 7004 \\\hline %
    150 & 0.03 / 402 / 1404 & 0.09 / 802 / 2804 & 0.18 / 1202 / 4204 & 0.29 / 1602 / 5604 & 0.47 / 2002 / 7004 \\\hline 
    200 & 0.02 / 402 / 1404 & 0.07 / 802 / 2804 & 0.16 / 1202 / 4204 & 0.23 / 1602 / 5604 & 0.47 / 2002 / 7004 \\\hline %
    250 & 0.01 / 402 / 1404 & 0.07 / 802 / 2804 & 0.17 / 1202 / 4204 & 0.29 / 1602 / 5604 & 0.48 / 2002 / 7004 \\\hline
\end{tabular}
}
\medskip

Notice that the number of iterations in both internal and external
loops do no longer depend on the choice of parameter $W$, as well as
the execution time. With respect to the execution time, the decrease
from the case without acceleration is even larger, since the updates
of vector $\mathsf X$ inside the inner loop need only to be performed
on the vertices of the current component. For large instances, the
execution time may again become very large, but in case $W=500$ (as
previously said, this value is independent of the result) and
$n=1000$, it shrinks to $2.3 s$ whereas the total number of iterations
in the internal loop becomes $14 004$, i.e., 5 orders of magnitude
less than for the algorithm without acceleration heuristics.

\section{Acceleration heuristics in MCR games\label{sec:accel-heur-rtp}}
Algorithm~\ref{algo:value-iteration-RT-acc} and
\ref{algo:value-iter-TPO-acc} are enhanced versions of
Algorithm~\ref{algo:value-iteration-RT} and \ref{algo:value-iter-TPO}
respectively, that apply the acceleration heuristics described at the
end of Section~\ref{sec:experiments}.  \medskip

\begin{algorithm}[h]
  \DontPrintSemicolon%
  \KwIn{min-cost reachability game $\gameEx[\MCR[\{\target\}]]$,
    SCC-decomposition
    $\decomposition\colon\vertices\to\{0,1,\ldots,p\}$ and an oracle
    $\mathcal O(q,v)$ outputting sets
    $(S_v)_{v\in\decomposition^{-1}(q)}$}%
  \SetKw{value}{\ensuremath{\mathsf{X}}}
  \SetKw{prevvalue}{\ensuremath{\mathsf{X}_{pre}}}

  \BlankLine

  $\value(\target):=0$\;
  \For{$q=1$ \KwTo $p$}%
  {
    $(S_v)_{v\in\decomposition^{-1}(q)} := \mathcal O(q,\value)$
    \tcc*[r]{Use of the oracle}%
    \lForEach{$v\in\decomposition^{-1}(q)$}{$\value(v):=\max S_v$}
    
    \Repeat{$\value=\prevvalue$}{%
      $\prevvalue := \value$\;%
      \ForEach{$v\in\decomposition^{-1}(q)\cap\maxvertices$}%
      {$\value(v) := \max_{v'\in\edges(v)}
        \big(\edgeweights(v,v')+\prevvalue(v')\big)$}%
      \ForEach{$v\in\decomposition^{-1}(q)\cap\minvertices$}%
      {$\value(v) := \min_{v'\in\edges(v)}
        \big(\edgeweights(v,v')+\prevvalue(v')\big)$}%
      \lForEach{$v\in\decomposition^{-1}(q)$}{$\value(v):=\max \big(S_v \cap
        [-\infty,\value(v)]\big)$}%
    }%
  }%
  
  \Return{$\value$}
  \caption{Accelerated value iteration algorithm for min-cost
    reachability games}\label{algo:value-iteration-RT-acc}
\end{algorithm}

\begin{algorithm}[tbp]
  \DontPrintSemicolon %
  \KwIn{Total-payoff game $\game=\gameEx[\TP]$, SCC-decomposition
    $\decomposition\colon\vertices\to\{0,1,\ldots,p\}$ and an oracle
    $\mathcal O(q,v)$ outputting sets
    $(S_v)_{v\in\decomposition^{-1}(q)}$}
  \SetKw{value}{\ensuremath{\mathsf{Y}}}%
  \SetKw{prevvalue}{\ensuremath{\mathsf{Y}_{pre}}}%
  \SetKw{valueint}{\ensuremath{\mathsf{X}}}%
  \SetKw{prevvalueint}{\ensuremath{\mathsf{X}_{pre}}}%
  \BlankLine
  
  \lForEach{$v\in\vertices$}{$\value(v) := -\infty$}%
  \For{$q=0$ \KwTo $p$}%
  {%
    $(S_v)_{v\in\decomposition^{-1}(q)} := \mathcal
    O(q,\value)$\tcc*[r]{Use of the oracle} %
    \Repeat{$\value=\prevvalue$}{ %
      \ForEach{$v\in\decomposition^{-1}(q)$}{$\prevvalue(v):=\value(v)$;
        $\value(v):= \max (0,\value(v))$; $\valueint(v):=\max S_v$}
      \Repeat{$\valueint=\prevvalueint$}{%
        $\prevvalueint:=\valueint$\;%
        \ForEach{$v\in\maxvertices\cap
          \decomposition^{-1}(q)$}{$\valueint(v)
          := \max_{v'\in \edges(v)}
          \big[\edgeweights(v,v')+\min(\prevvalueint(v'),\value(v'))\big]$} %
        \ForEach{$v\in\minvertices\cap
          \decomposition^{-1}(q)$}{$\valueint(v)
          := \min_{v'\in\edges(v)}
          \big[\edgeweights(v,v')+\min(\prevvalueint(v'),\value(v'))\big]$} %
        \lForEach{$v\in\decomposition^{-1}(q)$}%
        {$\valueint(v) := \max \big(S_v \cap
          [-\infty,\valueint(v)]\big)$} %
      } %
      $\value:=\valueint$\;%
      \lForEach{$v\in\vertices$ \emph{such that}
        $\value(v) > (|\vertices|-1) W$}%
      {$\value(v) := +\infty$}%
    }%
  }%
  
  \Return{$\value$}
  \caption{Accelerated value iteration algorithm for total-payoff
    games\label{algo:value-iter-TPO-acc}}
\end{algorithm}

\end{document}